\pdfoutput=1

\documentclass[copyright,creativecommons]{eptcs}

\submissiontrue
\providecommand{\event}{QPL 2020}

\usepackage{bookmark}

\usepackage[table]{xcolor}

\usepackage[all,cmtip]{xy} 

\usepackage{tikz}
\usepackage[oldvoltagedirection]{circuitikz}

\usepackage{amsmath}

\usepackage{stmaryrd}

\usepackage{mathtools}

\usepackage{mdframed}

\usepackage{arydshln}

\usepackage{multicol}

\usepackage{float}

\renewcommand{\tilde}{\widetilde}

\usepackage{everypage}
\usepackage{lipsum}

\usepackage{amsmath}
\usepackage{amsfonts}
\usepackage{amsthm}
\usepackage[inline]{enumitem}

\usepackage{scalerel,stackengine}
\stackMath
\renewcommand\hat[1]{%
\savestack{\tmpbox}{\stretchto{%
  \scaleto{%
    \scalerel*[\widthof{\ensuremath{#1}}]{\kern-.6pt\bigwedge\kern-.6pt}%
    {\rule[-\textheight/2]{1ex}{\textheight}}
  }{\textheight}%
}{0.5ex}}%
\stackon[1pt]{#1}{\tmpbox}%
}
\newcommand{\bcell}{\cellcolor{black!10}}

\makeatletter
\newcommand{\inlineitem}[1][]{%
\ifnum\enit@type=\tw@
    {\descriptionlabel{#1}}
  \hspace{\labelsep}%
\else
  \ifnum\enit@type=\z@
       \refstepcounter{\@listctr}\fi
    \quad\@itemlabel\hspace{\labelsep}%
\fi}
\makeatother
\newdir{|>}{-<5pt,0pt>{
\begin{tikzpicture}[scale=.7]
	\begin{pgfonlayer}{nodelayer}
		\node [style=none] (0) at (1.5, 0) {};
		\node [style=none] (1) at (2.5, 0) {};
		\node [style=none] (2) at (0.5, -0.5) {};
	\end{pgfonlayer}
	\begin{pgfonlayer}{edgelayer}
		\draw (2.center) to (0.center);
		\draw (0.center) to (1.center);
	\end{pgfonlayer}
\end{tikzpicture}
}}
\newdir{||>}{-<5pt,0pt>{
\begin{tikzpicture}[scale=.7]
	\begin{pgfonlayer}{nodelayer}
		\node [style=none] (0) at (1.25, 0.5) {};
		\node [style=none] (1) at (2.5, 0.25) {};
		\node [style=none] (2) at (0.5, 0) {};
		\node [style=none] (3) at (2, -1) {};
	\end{pgfonlayer}
	\begin{pgfonlayer}{edgelayer}
		\draw (2.center) to (0.center);
		\draw (0.center) to (1.center);
	\end{pgfonlayer}
\end{tikzpicture}
}}
\newdir{|<}{-<5pt,0pt>{
\begin{tikzpicture}[scale=.9]
	\begin{pgfonlayer}{nodelayer}
		\node [style=none] (0) at (0, -0.25) {};
		\node [style=none] (1) at (-1, -0.25) {};
		\node [style=none] (2) at (1, 0) {};
	\end{pgfonlayer}
	\begin{pgfonlayer}{edgelayer}
		\draw (2.center) to (0.center);
		\draw (0.center) to (1.center);
	\end{pgfonlayer}
\end{tikzpicture}
}}

\newcommand{\xrightarrowtail}[1]{\!\!{\xymatrix@C=1em{\ar@{>->}[r]^{#1}&}}\!\!\!}
\newcommand{\xleftarrowtail}[1]{\!\!\!{\xymatrix@C=1em{&\ar@{>->}[l]_{#1}}}\!\!}

\usetikzlibrary{shapes.geometric}
\usetikzlibrary{patterns}
\usetikzlibrary{fit}
\usetikzlibrary{positioning}
\usetikzlibrary{calc}
\usetikzlibrary{arrows}
\usetikzlibrary{decorations.markings}
\usetikzlibrary{decorations.pathreplacing}
\usetikzlibrary{shapes}

\pgfdeclarelayer{background}
\pgfdeclarelayer{nodelayer}
\pgfdeclarelayer{edgelayer}
\pgfsetlayers{background,edgelayer,nodelayer,main}

\tikzset{H/.style={draw,color=black,fill={rgb:black,1;white,3}, rectangle}}

\tikzset{rn/.style={}}
\tikzset{simple/.style={}}
\tikzset{none/.style={}}
\tikzset{nothing/.style={}}
\tikzset{thick/.style={draw, line width=0.5mm }}
\tikzset{X/.style={draw,fill=white, circle,scale=1, inner sep=0pt, minimum size=10pt}}
\tikzset{Z/.style={draw,fill={rgb:black,1;white,3}, text=black, circle,scale=1, inner sep=0pt, minimum size=10pt }}
\tikzset{Xthick/.style={draw,fill=white, circle,scale=1, inner sep=0pt, minimum size=15pt,line width=0.5mm}}
\tikzset{Zthick/.style={draw,fill={rgb:black,1;white,3}, text=black, circle,scale=1, inner sep=0pt, minimum size=15pt,line width=0.5mm }}

\tikzset{phase/.style={draw,fill=white, diamond,scale=1, inner sep=0pt, minimum size=10pt}}

\tikzset{discard/.style={draw, xscale=2.2,ground, rotate=90}}
\tikzset{mmixed/.style={draw, quantum, yscale=-2.2,ground, rotate=180}}
\tikzset{quantum/.style={line width=.6mm}}
\tikzset{map/.style={draw,color=black,fill=white, rectangle}}
\tikzset{mapthick/.style={draw,color=black,fill=white, rectangle, inner sep=0pt, minimum size=15pt,line width=0.5mm }}

\tikzset{otimes/.style={draw,fill=white,rotate=45, scale=0.9,minimum height=.1cm,circle,append after command={
[shorten >=\pgflinewidth, shorten <=\pgflinewidth,]
(\tikzlastnode.north) edge (\tikzlastnode.south)
(\tikzlastnode.east) edge (\tikzlastnode.west)
}
}
}

\tikzset{dot/.style={thick, fill=black, circle, scale=1, inner sep = .05cm}}

\tikzset{oplus/.style={draw, scale=0.9,minimum height=.1cm,circle,append after command={
[shorten >=\pgflinewidth, shorten <=\pgflinewidth,]
(\tikzlastnode.north) edge (\tikzlastnode.south)
(\tikzlastnode.east) edge (\tikzlastnode.west)
}
}
}

\usetikzlibrary{arrows,shapes.gates.logic.US,shapes.gates.logic.IEC,calc}

\tikzset{andin/.style={
draw,
and gate US,
scale=1,
fill=white,
label={center:{\it \&}}
}}

\tikzset{andout/.style={
draw,
and gate US,
rotate=-180,
scale=1,
fill=white,
label={center:{\it \&}}
}}

\tikzset{tri/.style={
draw,
shape border rotate=-30,
regular polygon,
regular polygon sides=3,
fill={rgb:black,1;white,3},
inner sep = .1cm
}
}

\tikzset{triflip/.style={
draw,
shape border rotate=30,
regular polygon,
regular polygon sides=3,
fill={rgb:black,1;white,3},
inner sep = .1cm
}
}

\tikzset{fanin/.style={
draw,
shape border rotate=30,
regular polygon,
regular polygon sides=3,
fill=white,
inner sep = .1cm
}
}

\tikzset{fanout/.style={
draw,
shape border rotate=-30,
regular polygon,
regular polygon sides=3,
fill=white,
inner sep = .1cm
}
}

\tikzset{onein/.style={
draw,
shape border rotate=30,
regular polygon,
regular polygon sides=3,
fill=black,
inner sep = .04cm,
scale=1.2
}
}

\tikzset{oneout/.style={
draw,
shape border rotate=-30,
regular polygon,
regular polygon sides=3,
fill=black,
inner sep = .04cm,
scale=1.2
}
}

\tikzset{zeroin/.style={
draw,
shape border rotate=30,
regular polygon,
regular polygon sides=3,
fill=white,
inner sep = .04cm,
scale=1.2
}
}

\tikzset{zeroout/.style={
draw,
shape border rotate=-30,
regular polygon,
regular polygon sides=3,
fill=white,
inner sep = .04cm,
scale=1.2
}
}

\tikzstyle{strings}=[baseline={([yshift=-.5ex]current bounding box.center)}]

\tikzset{every picture/.append style={scale=.75}, transform shape,strings}

  \newtheorem{thm}{Theorem}[section]
  \newtheorem{corollary}[thm]{Corollary}
  \newtheorem{lemma}[thm]{Lemma}
  \newtheorem{proposition}[thm]{Proposition}

  \newtheorem{definition}[thm]{Definition}
  \newtheorem{example}[thm]{Example}
  \newtheorem{remark}[thm]{Remark}

\newcommand{\cnot}{\mathsf{cnot}}
\newcommand{\tof}{\mathsf{tof}}
\newcommand{\Not}{\mathsf{not}}
\newcommand{\zeroin}{|0\rangle}
\newcommand{\zeroout}{\langle 0|}
\newcommand{\CNOT}{\mathsf{CNOT}}
\newcommand{\Sets}{\mathsf{Set}}

\newcommand{\FinOrd}{\mathsf{FinOrd}}
\newcommand{\TOF}{\mathsf{TOF}}
\newcommand{\Span}{\mathsf{Span}}

\newcommand{\op}{\mathsf{op}}
\newcommand{\Hilb}{\mathsf{Hilb}}

\newcommand{\FHilb}{\mathsf{FHilb}}

\newcommand{\CP}{\mathsf{CP}}
\newcommand{\FPinj}{\mathsf{FPinj}}
\newcommand{\Pinj}{\mathsf{Pinj}}
\newcommand{\Par}{\mathsf{Par}}
\newcommand{\ParIso}{\mathsf{ParIso}}
\newcommand{\Total}{\mathsf{Total}}

\newcommand{\Mat}{\mathsf{Mat}}
\newcommand{\STOCH}{\mathsf{STOCH}}

\newcommand{\X}{\mathbb{X}}

\newcommand{\N}{\mathbb{N}}
\newcommand{\F}{\mathbb{F}}

\newcommand{\M}{\mathcal{M}}

\newcommand{\eq}[1]{\stackrel{\scalebox{.6}{#1}}{=}}

\newcommand{\ZXA}{\mathsf{ZX}\textit{\&}}

\newcommand{\ZX}{\mathsf{ZX}}
\newcommand{\ZH}{\mathsf{ZH}}
\DeclareMathSymbol{\bot}{\mathord}{symbols}{"3F}

\renewcommand{\epsilon}{\varepsilon}
\renewcommand{\bar}[1]{\overline{#1}\hspace*{.01cm}}

\xymatrixrowsep{.05cm}
\xymatrixcolsep{.3cm}

\title{The $\ZXA$-calculus: A complete graphical calculus for classical circuits using spiders}
\def\titlerunning{The $\ZXA$-calculus}
\date{\today}
\author{Cole Comfort\\ Department of Computer Science, University of Oxford}
\def\authorrunning{Cole Comfort}

\begin{document}


\maketitle

\begin{abstract}
We give a complete presentation for the fragment, $\ZXA$, of the $\ZX$-calculus generated by the Z and X spiders (corresponding to copying and addition) along with the {\sf not} gate and the {\sf and} gate.  To prove completeness, we freely add a unit and counit to the category $\TOF$ generated by the Toffoli gate and ancillary bits, showing that this yields the full subcategory of finite ordinals and functions with objects powers of two; and then perform a two way translation between this category and $\ZXA$.  A translation to some extension of $\TOF$, as opposed to some fragment of the $\ZX$-calculus, is a natural choice because of the multiplicative nature of the Toffoli gate.  To this end, we show that freely adding counits to the semi-Frobenius algebras of a discrete inverse category is the same as constructing the Cartesian completion. In particular, for a discrete inverse category, the category of classical channels, the Cartesian completion and adding counits all produce the same category.  Therefore, applying these constructions to $\TOF$ produces the full subcategory of finite ordinals and partial maps with objects powers of two.  By glueing together the  free counit completion and the free unit completion, this yields ``qubit multirelations.''
\end{abstract}


\section{Introduction}

In this paper a complete set of identities is provided for the fragment, $\ZXA$, of the $\ZX$-calculus, generated by black and white spiders, the not gate and the {\sf and} gate. We show that this is a universal and complete presentation of ``qubit multirelations,'' or equivalently $2^n \times 2^m$ dimensional matrices over $\N$.
 To prove completeness and universality requires much exposition.  Along the way we show that the category of classical channels of a discrete inverse category is the Cartesian completion of that discrete inverse category.  We then show that the corresponding environment structure is precisely the free counit completion of the chosen Frobenius structure.  This allows us to present the Cartesian completion of, $\TOF$, the category generated by the Toffoli gate, $|1\rangle$ and $\langle 1|$ by only adding the $|+\rangle$ state and the unitality equation.  By freely adding both the unit and counit to $\TOF$, corresponding to $\sqrt{2}|+\rangle$ and  $\sqrt{2}\langle +|$, this yields an isomorphism with spans between ordinals $2^n$, $n\in \N$, or equivalently, ``qubit multirelations.''

The identities which are given by this two way translation are {\em almost} the union of the complete identities for Boolean functions \cite[Thm. 10]{lafont} (functions of type $\F_2^n \to \F_2$) and the identities for $\Span^\sim(\Mat(\F_2))$ \cite[Def. 5.1]{linrel}.  These classes of circuits, and these identities for that matter, are nothing new; however, we provide a completeness result, as well as a structural account of how the full classical qubit fragment of $\FHilb$ can be obtained from adding discarding and codiscarding to the full classically reversible Boolean fragment.  In fact, some of these identities are presented in \cite[Chap. 5]{herrmann},  and they are used in the $\ZH$-calculus \cite{zh,zhpi}, as well as in some presentations of the $\ZX$-calculus with the triangle generator as a primitive \cite{munson2019note,ringZX}.  This is particularity unsurprising for the latter, \cite{ringZX}, where the author proves completeness of the $\ZX$-calculus over arbitrary semirings, which subsumes the completeness result herein.  Albeit, the presentation given here is substantially simpler.  It worth mentioning that $\ZXA$ is not a $\ZX*$-calculus in the sense of \cite{zxstar}, because the {\sf and} gate is not a spider.  $\ZXA$ should be instead though of as the ``classical fragment'' of the phase-free $\ZH$-calculus: retaining the monoid for ``and'' without $H$-boxes.   From this presentation only natural-number H-boxes can be derived.

We assume familiarity with the theory of  monoidal categories and categorical quantum mechanics.
Most of the paper will be devoted to reviewing the required categorical machinery of restriction and inverse categories, and developing it further, in order to prove the main result.  With all of mathematics reviewed and developed in generality, the desired result follows from abstract nonsense after a mechanical calculation. 

In Section \ref{sec:rest}, the theory of restriction categories and inverse categories is reviewed.  In Section \ref{sec:cpm}, we construct classical channels in the setting of discrete inverse categories, showing that the ``environment structures'' of the classical channels corresponds to adding counits to the base discrete inverse category.  Finally, in Section \ref{sec:ZXA}, we actually compute the (co)unit completion of $\TOF$.  We show that this category has a much more canonical presentation, $\ZXA$, in terms of interacting monoids/comonoids which very much resembles the $\ZH$-calculus.  We also show that this category is isomorphic to the category spans between ordinals $2^n$.
\section{Restriction and Inverse Categories}

\label{sec:rest}

Restriction and inverse categories provide a categorical semantics for partial computing and reversible computing, respectively.  We review how weakened products can be constructed in both settings; relating one to the other.

\begin{definition}\cite[\S 2.1.1]{resti}
A {\bf restriction category} is a category along with a restriction operator:

\hfil
$
(A \xrightarrow{f} B )\mapsto (A \xrightarrow{\bar f} A)
$\\
such that:\footnote{Using diagrammatic composition.}

\begin{center}
\begin{multicols}{4}
\begin{enumerate}[label={\bf [R.\arabic*]}, ref={\bf [R.\arabic*]}]
\item $\bar f f  = f$
\label{R.1}
\item $\bar f \bar g = \bar g \bar f$
\label{R.2}
\item $\bar f \bar g = \bar{\bar f g}$
\label{R.3}
\item $f \bar g = \bar{fg} f$
\label{R.4}
\end{enumerate}
\end{multicols}
\end{center}

Maps of the form $\bar f$ are called restriction idempotents.
The canonical example of a restriction category is $\Par$, sets and partial maps.  The restriction in this case, just restricts partial functions to their domain of definition.

Restriction categories have a partial order on homsets given by $f \leq g \iff \bar f g = f$.

A map $f$ in a restriction category is called a {\bf partial isomorphism}, in case there exists a map $g$ called the partial inverse of $f$ so that $fg=\bar f$ and $gf = \bar g$.  Similarly, a map $f$ in a restriction category is {\bf total} if $\bar f =1$.  Denote the subcategories of partial isomorphisms and total maps of a restriction category $\X$, respectively by $\ParIso(\X)$ and $\Total(\X)$.

\end{definition}

\begin{example} \cite[p. 101]{pcat} \cite[\S 5]{restiii}
A {\bf counital copy category} (or a p-category with a one element object) is a monoidal category with a family of commutative comonoids on every object compatible with the monoidal structure, with a natural comultiplication.  This gives a restriction via copying and then discarding:
$$
\begin{tikzpicture}
	\begin{pgfonlayer}{nodelayer}
		\node [style=none] (0) at (0, 0.75) {};
		\node [style=none] (1) at (-2, 0.75) {};
		\node [style=map] (2) at (-1, 0.75) {$\bar f$};
	\end{pgfonlayer}
	\begin{pgfonlayer}{edgelayer}
		\draw [style=simple] (0.center) to (2);
		\draw [style=simple] (2) to (1.center);
	\end{pgfonlayer}
\end{tikzpicture}
:=
\begin{tikzpicture}
	\begin{pgfonlayer}{nodelayer}
		\node [style=map] (0) at (0, -0) {$f$};
		\node [style=X] (1) at (1, -0) {};
		\node [style=X] (2) at (-1, -0.5) {};
		\node [style=none] (3) at (1, -1) {};
		\node [style=none] (4) at (-2, -0.5) {};
	\end{pgfonlayer}  
	\begin{pgfonlayer}{edgelayer}
		\draw [style=simple] (1) to (0);
		\draw [style=simple, in=27, out=180, looseness=1.00] (0) to (2);
		\draw [style=simple] (2) to (4.center);
		\draw [style=simple, in=180, out=-30, looseness=1.00] (2) to (3.center);
	\end{pgfonlayer}
\end{tikzpicture}
$$
\end{example}

\begin{definition}\cite[\S 3.1]{resti}
A {\bf stable system of monics} $\M$ of $\X$ is a collection of monics in $\X$ containing all isomorphisms; where for any cospan $ X\xrightarrow{f} Z \xleftarrowtail{m} Y$  in $\X$, where $m'$ is in $\M$, the following pullback exists:


\hfil$
\xymatrixrowsep{.005in}
\xymatrixcolsep{.13in}
  \xymatrix{
  	& W \ar@{>->}[dl]_{m'} \ar[dr]^{f'}\\
  	X \ar[dr]_f &  & Y \ar@{>->}[dl]^m\\
  	& Z
  }
$

Where $m'$ is in $\M$.

\end{definition}

Stable systems of monics allow one to represent the domains of definition of a partial functions as a subobjects:

\begin{definition}\cite[\S 3.1]{resti}
Given a stable system of monics $\M$ in a category $\X$, the {\bf partial map category} $\Par(\X,\M)$ is given by the same objects as in $\X$ where morphisms $X\to Y$, given by isomorphism classes of spans $X\xleftarrowtail{m} Z \xrightarrow{f} Y$ where $f$ is a map in $\X$ and $m$ is a map in $\M$.  Composition is given by pullback and the identity is given by the trivial span.

Partial map categories have a restriction structure given by:  $(X\xleftarrowtail{m} Z \xrightarrow{f} Y) \mapsto (X\xleftarrowtail{m} Z \xrightarrowtail{m} X)$.  Moreover, a partial isomorphism is a span $X\xleftarrowtail{e} Z \xrightarrowtail{m} Y$ where $e,m \in \M$; the partial inverse given by  $Y\xleftarrowtail{m} Z \xrightarrowtail{e} X$.
\end{definition}

$\Par$ is equivalently the partial map category $\Par(\Sets,\M)$ where $\M$ is all monics in $\Sets$.



Let $\Span^\sim(\X)$ denote the category given by isomorphism classes of spans over $\X$. Given a stable system of monics $\M$ over $\X$, if $\X$ is finitely complete, then $\Span^\sim(\X)$ exists, and thus, there is a faithful functor $\Par(\X,\M)\to \Span^\sim(\X)$.

\begin{definition}\cite[\S 2.3.2]{resti}
An {\bf inverse category} is a restriction category in which all maps are partial isomorphisms.  The subcategory of partial isomorphisms of $\Par$ is called $\Pinj$.
\end{definition}

Inverse categories can be presented with a dagger functor taking maps to their partial inverses:

\begin{thm}\cite[Thm. 2.20]{resti}
A restriction category $\X$ is an inverse category if and only if there is a dagger functor $(\_)^\circ:\X^\op\to\X$ such that for all $X\xleftarrow{f} Z \xrightarrow{g} Y$:
\begin{center}
\begin{tabular}{cc}
 $f f^\circ f = f$ & 
 $f f ^\circ gg^\circ = gg^\circ f f ^\circ $
\end{tabular}
\end{center}
\end{thm}

Since restriction categories  and inverse categories give a categorical semantics for partial computing  and reversible computing, respectively, it is natural to ask when these categories have copying.

In the case of restriction categories, one must weaken the notion of the product to lax products using the partial order enrichment:

\begin{definition}\cite{restiii}
A restriction category has {\bf binary restriction products}, when for all objects  $X,Y$, there exists an object $X\times Y$ and total maps $X \xleftarrow{\pi_0}  X\times Y \xrightarrow{\pi_1} Y$, so that for all objects $Z$ and all maps $X \xleftarrow{f} Z \xrightarrow{g} Y$, the following diagram commutes there exists a unique $Z\xrightarrow{\langle f,g \rangle} X\times Y$ making the diagram commute:
\hfil
$
\xymatrixrowsep{0.2cm}
\xymatrixcolsep{0.4cm}
\xymatrix{
&& Z\ar@{..>}[dd]|-{\langle f, g\rangle} \ar@/_/[ddll]_f \ar@/^/[ddrr]^g &&\\
& \ar@{}[dr]|-{\geq} && \ar@{}[dl] |-{\leq} &\\
X &&  X\times Y \ar[rr]_{\pi_1} \ar[ll]^{\pi_0}  && Y
}
$

so that $\bar{\langle f, g\rangle \pi_0} f = \langle f, g\rangle \pi_0$ and $\bar{\langle f, g\rangle \pi_1} g = \langle f, g\rangle \pi_1$;
where additionally $\bar{\langle f, g\rangle} =  \bar f \bar g$.


A restriction category has a {\bf restriction terminal object} $\top$ when for all objects $X$, there exists a unique total map $!_X:X\to\top$ such that $f !_Y = \bar f !_X$.

A restriction category with a restriction terminal object and binary restriction products is a {\bf Cartesian restriction category}.

An object $A$ in a restriction category with restriction products is {\bf discrete} when the diagonal map $\Delta_X:=\langle 1_X, 1_X\rangle$ is a partial isomorphism. A restriction category is discrete when all objects are discrete.  Discrete Cartesian restriction categories are said to have restriction products.
\end{definition}

\begin{thm}\cite[Thm. 5.2]{restiii}
The structure of a  counital copy category structure is precisely that of a Cartesian restriction category.
\end{thm}

\begin{proposition} \cite[\S 5.1]{restiii}
\label{prop:cartesian}

If $\X$ is a discrete Cartesian restriction category, then $\Total(\X)$ is Cartesian.
\end{proposition}

$\Par$ is a canonical example of a discrete Cartesian restriction category; the restriction product is given by the Cartesian product on underlying sets and the terminal object is  the singleton set.

The weakened notion of products in restriction categories is not satisfying for inverse categories because it does not impose enough equations governing the interaction between the diagonal map and its partial inverse.

\begin{definition}\cite[Def. 4.3.1]{giles}
A symmetric monoidal inverse category $\X$ is a {\bf discrete inverse category} when there is a natural, special commutative $\dag$-semi-Frobenius algebra\footnote{The ``semi'' adjective on Frobenius just means that the a semigroup and cosemigroup are interacting instead of a monoid and comonoid.} on every object (where the (co)multiplications are drawn as white bubbles)  compatible with the tensor product:

$$
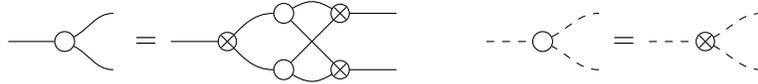
\begin{tikzpicture}
	\begin{pgfonlayer}{nodelayer}
		\node [style=none] (0) at (0, -0) {};
		\node [style=none] (1) at (0, -1) {};
		\node [style=X] (2) at (-1, -0.5) {};
		\node [style=none] (3) at (-2, -0.5) {};
	\end{pgfonlayer}
	\begin{pgfonlayer}{edgelayer}
		\draw [style=simple] (3.center) to (2);
		\draw [style=simple, in=180, out=27, looseness=1.00] (2) to (0);
		\draw [style=simple, in=-27, out=180, looseness=1.00] (1) to (2);
	\end{pgfonlayer}
\end{tikzpicture}
=
\begin{tikzpicture}
	\begin{pgfonlayer}{nodelayer}
		\node [style=X] (0) at (0, -0) {};
		\node [style=X] (1) at (0, -1) {};
		\node [style=none] (2) at (-1, -0.5) {};
		\node [style=none] (3) at (-2, -0.5) {};
		\node [style=none] (4) at (1, -0) {};
		\node [style=none] (5) at (1, -1) {};
		\node [style=none] (6) at (2, -0) {};
		\node [style=none] (7) at (2, -1) {};
		\node [style=otimes] (8) at (-1, -0.5) {};
		\node [style=otimes] (9) at (1, -1) {};
		\node [style=otimes] (10) at (1, -0) {};
	\end{pgfonlayer}
	\begin{pgfonlayer}{edgelayer}
		\draw [style=simple] (3.center) to (2.center);
		\draw [style=simple, in=180, out=45, looseness=1.00] (2.center) to (0);
		\draw [style=simple] (0) to (5.center);
		\draw [style=simple, in=30, out=150, looseness=1.25] (4.center) to (0);
		\draw [style=simple, in=-150, out=-30, looseness=1.25] (1) to (5.center);
		\draw [style=simple] (1) to (4.center);
		\draw [style=simple, in=-45, out=180, looseness=1.00] (1) to (2.center);
		\draw [style=simple] (4.center) to (6.center);
		\draw [style=simple] (5.center) to (7.center);
	\end{pgfonlayer}
\end{tikzpicture}
\hspace*{1cm}
\begin{tikzpicture}
	\begin{pgfonlayer}{nodelayer}
		\node [style=X] (0) at (0, -0) {};
		\node [style=none] (1) at (1, 0.5) {};
		\node [style=none] (2) at (1, -0.5) {};
		\node [style=none] (3) at (-1, -0) {};
	\end{pgfonlayer}
	\begin{pgfonlayer}{edgelayer}
		\draw [style=dashed] (3.center) to (0);
		\draw [style=dashed, in=180, out=27, looseness=1.00] (0) to (1.center);
		\draw [style=dashed, in=-27, out=180, looseness=1.00] (2.center) to (0);
	\end{pgfonlayer}
\end{tikzpicture}
=
\begin{tikzpicture}
	\begin{pgfonlayer}{nodelayer}
		\node [style=none] (0) at (0, -0) {};
		\node [style=none] (1) at (1, 0.5) {};
		\node [style=none] (2) at (1, -0.5) {};
		\node [style=none] (3) at (-1, -0) {};
		\node [style=otimes] (4) at (0, -0) {};
	\end{pgfonlayer}
	\begin{pgfonlayer}{edgelayer}
		\draw [style=dashed] (3.center) to (0);
		\draw [style=dashed, in=180, out=27, looseness=1.00] (0) to (1.center);
		\draw [style=dashed, in=-27, out=180, looseness=1.00] (2.center) to (0);
	\end{pgfonlayer}
\end{tikzpicture}
$$

Where the tensor product is also required to preserve restriction in both components.
\end{definition}

In a discrete inverse category, restriction idempotents are prephases for the Frobenius algebra, so that:
$$
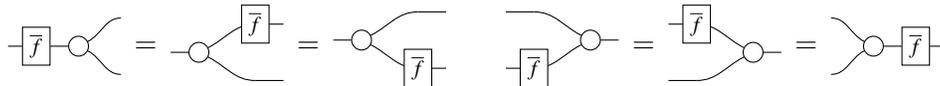
\begin{tikzpicture}
	\begin{pgfonlayer}{nodelayer}
		\node [style=X] (0) at (1.75, -3) {};
		\node [style=map] (1) at (1, -3) {$\bar f$};
		\node [style=none] (2) at (0.5, -3) {};
		\node [style=none] (3) at (2.5, -2.5) {};
		\node [style=none] (4) at (2.5, -3.5) {};
	\end{pgfonlayer}
	\begin{pgfonlayer}{edgelayer}
		\draw [style=simple, in=-27, out=180, looseness=1.00] (4.center) to (0);
		\draw [style=simple, in=180, out=27, looseness=1.00] (0) to (3.center);
		\draw [style=simple] (1) to (0);
		\draw [style=simple] (1) to (2.center);
	\end{pgfonlayer}
\end{tikzpicture}
=
\begin{tikzpicture}
	\begin{pgfonlayer}{nodelayer}
		\node [style=X] (0) at (2, -3) {};
		\node [style=none] (1) at (1.5, -3) {};
		\node [style=none] (2) at (3, -2.5) {};
		\node [style=none] (3) at (3, -3.5) {};
		\node [style=map] (4) at (3, -2.5) {$\bar f$};
		\node [style=none] (5) at (3.5, -3.5) {};
		\node [style=none] (6) at (3.5, -2.5) {};
	\end{pgfonlayer}
	\begin{pgfonlayer}{edgelayer}
		\draw [style=simple, in=-27, out=180, looseness=1.00] (3.center) to (0);
		\draw [style=simple, in=180, out=27, looseness=1.00] (0) to (2.center);
		\draw [style=simple] (6.center) to (2.center);
		\draw [style=simple] (5.center) to (3.center);
		\draw [style=simple] (0) to (1.center);
	\end{pgfonlayer}
\end{tikzpicture}
=
\begin{tikzpicture}
	\begin{pgfonlayer}{nodelayer}
		\node [style=X] (0) at (2, -3) {};
		\node [style=none] (1) at (1.5, -3) {};
		\node [style=none] (2) at (3, -3.5) {};
		\node [style=none] (3) at (3, -2.5) {};
		\node [style=map] (4) at (3, -3.5) {$\bar f$};
		\node [style=none] (5) at (3.5, -2.5) {};
		\node [style=none] (6) at (3.5, -3.5) {};
	\end{pgfonlayer}
	\begin{pgfonlayer}{edgelayer}
		\draw [style=simple, in=27, out=180, looseness=1.00] (3.center) to (0);
		\draw [style=simple, in=180, out=-27, looseness=1.00] (0) to (2.center);
		\draw [style=simple] (6.center) to (2.center);
		\draw [style=simple] (5.center) to (3.center);
		\draw [style=simple] (0) to (1.center);
	\end{pgfonlayer}
\end{tikzpicture}
\hspace*{.6cm}
\begin{tikzpicture}
	\begin{pgfonlayer}{nodelayer}
		\node [style=X] (0) at (3, -3) {};
		\node [style=none] (1) at (3.5, -3) {};
		\node [style=none] (2) at (2, -3.5) {};
		\node [style=none] (3) at (2, -2.5) {};
		\node [style=map] (4) at (2, -3.5) {$\bar f$};
		\node [style=none] (5) at (1.5, -2.5) {};
		\node [style=none] (6) at (1.5, -3.5) {};
	\end{pgfonlayer}
	\begin{pgfonlayer}{edgelayer}
		\draw [style=simple, in=153, out=0, looseness=1.00] (3.center) to (0);
		\draw [style=simple, in=0, out=-153, looseness=1.00] (0) to (2.center);
		\draw [style=simple] (6.center) to (2.center);
		\draw [style=simple] (5.center) to (3.center);
		\draw [style=simple] (0) to (1.center);
	\end{pgfonlayer}
\end{tikzpicture}
=
\begin{tikzpicture}
	\begin{pgfonlayer}{nodelayer}
		\node [style=X] (0) at (3, -3) {};
		\node [style=none] (1) at (3.5, -3) {};
		\node [style=none] (2) at (2, -2.5) {};
		\node [style=none] (3) at (2, -3.5) {};
		\node [style=map] (4) at (2, -2.5) {$\bar f$};
		\node [style=none] (5) at (1.5, -3.5) {};
		\node [style=none] (6) at (1.5, -2.5) {};
	\end{pgfonlayer}
	\begin{pgfonlayer}{edgelayer}
		\draw [style=simple, in=-153, out=0, looseness=1.00] (3.center) to (0);
		\draw [style=simple, in=0, out=153, looseness=1.00] (0) to (2.center);
		\draw [style=simple] (6.center) to (2.center);
		\draw [style=simple] (5.center) to (3.center);
		\draw [style=simple] (0) to (1.center);
	\end{pgfonlayer}
\end{tikzpicture}
=
\begin{tikzpicture}
	\begin{pgfonlayer}{nodelayer}
		\node [style=X] (0) at (1.25, -3) {};
		\node [style=map] (1) at (2, -3) {$\bar f$};
		\node [style=none] (2) at (2.5, -3) {};
		\node [style=none] (3) at (0.5, -2.5) {};
		\node [style=none] (4) at (0.5, -3.5) {};
	\end{pgfonlayer}
	\begin{pgfonlayer}{edgelayer}
		\draw [style=simple, in=-153, out=0, looseness=1.00] (4.center) to (0);
		\draw [style=simple, in=0, out=153, looseness=1.00] (0) to (3.center);
		\draw [style=simple] (1) to (0);
		\draw [style=simple] (1) to (2.center);
	\end{pgfonlayer}
\end{tikzpicture}
$$

Discrete inverse categories are the ``right'' notion of weakened products for monoidal inverse categories:

\begin{thm}\cite[Thm. 5.2.6]{giles}
There is an equivalence of categories between the category of discrete inverse categories and the category of discrete Cartesian categories.
\end{thm}

To go from  discrete Cartesian restriction categories to discrete inverse categories, one takes the subcategory of partial isomorphisms.
The other direction is less trivial; in particular, this involves adding a restriction terminal object via the following construction which ``adds a history'' to a partial isomorphism:

\begin{definition}\cite[Def. 5.1.1]{giles}
Given a discrete inverse category $\X$, define its {\bf Cartesian completion} $\tilde \X$ as the category with:

\begin{description}
\item[Objects:] The same objects as $\X$.
\item[Maps:]
\hfil
$
\dfrac{ X\xrightarrow{f} Y \otimes S \in \X}{ X\xrightarrow{(f,S)} Y \in \tilde \X}
$

Where two parallel maps $X\xrightarrow{(f,S), (g,T)} Y $ are equivalent when either (both conditions are equivalent):
$$
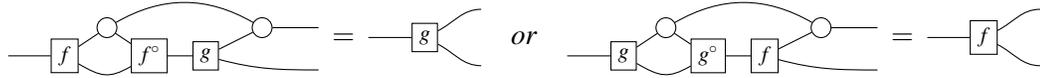
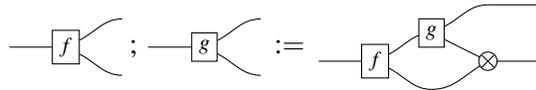
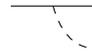
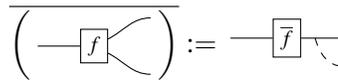
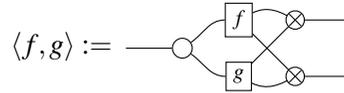
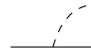
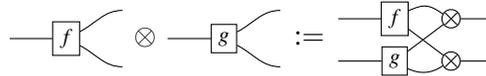
\begin{tikzpicture}
	\begin{pgfonlayer}{nodelayer}
		\node [style=map] (0) at (0, -0) {$f$};
		\node [style=none] (1) at (-1, -0) {};
		\node [style=map] (2) at (1.5, -0) {$f^\circ$};
		\node [style=map] (3) at (2.5, -0) {$g$};
		\node [style=X] (4) at (0.75, 0.5) {};
		\node [style=X] (5) at (3.5, 0.5) {};
		\node [style=none] (6) at (4.5, 0.5) {};
		\node [style=none] (7) at (4.5, -0.25) {};
	\end{pgfonlayer}
	\begin{pgfonlayer}{edgelayer}
		\draw (0) to (1.center);
		\draw [in=-15, out=180, looseness=1.00] (7.center) to (3);
		\draw (6.center) to (5);
		\draw [in=30, out=150, looseness=1.00] (5) to (4);
		\draw (4) to (2);
		\draw [in=-30, out=-150, looseness=1.25] (2) to (0);
		\draw (0) to (4);
		\draw (3) to (2);
		\draw (3) to (5);
	\end{pgfonlayer}
\end{tikzpicture}
=
\begin{tikzpicture}
	\begin{pgfonlayer}{nodelayer}
		\node [style=map] (0) at (0, -0) {$g$};
		\node [style=none] (1) at (1, 0.5) {};
		\node [style=none] (2) at (1, -0.5) {};
		\node [style=none] (3) at (-1, -0) {};
	\end{pgfonlayer}
	\begin{pgfonlayer}{edgelayer}
		\draw [in=27, out=180, looseness=1.00] (1.center) to (0);
		\draw [in=180, out=-27, looseness=1.00] (0) to (2.center);
		\draw (0) to (3.center);
	\end{pgfonlayer}
\end{tikzpicture}
\hspace*{.3cm}
or
\hspace*{.3cm}
\begin{tikzpicture}
	\begin{pgfonlayer}{nodelayer}
		\node [style=map] (0) at (0, -0) {$g$};
		\node [style=none] (1) at (-1, -0) {};
		\node [style=map] (2) at (1.5, -0) {$g^\circ$};
		\node [style=map] (3) at (2.5, -0) {$f$};
		\node [style=X] (4) at (0.75, 0.5) {};
		\node [style=X] (5) at (3.5, 0.5) {};
		\node [style=none] (6) at (4.5, 0.5) {};
		\node [style=none] (7) at (4.5, -0.25) {};
	\end{pgfonlayer}
	\begin{pgfonlayer}{edgelayer}
		\draw (0) to (1.center);
		\draw [in=-15, out=180, looseness=1.00] (7.center) to (3);
		\draw (6.center) to (5);
		\draw [in=30, out=150, looseness=1.00] (5) to (4);
		\draw (4) to (2);
		\draw [in=-30, out=-150, looseness=1.25] (2) to (0);
		\draw (0) to (4);
		\draw (3) to (2);
		\draw (3) to (5);
	\end{pgfonlayer}
\end{tikzpicture}
=
\begin{tikzpicture}
	\begin{pgfonlayer}{nodelayer}
		\node [style=map] (0) at (0, -0) {$f$};
		\node [style=none] (1) at (1, 0.5) {};
		\node [style=none] (2) at (1, -0.5) {};
		\node [style=none] (3) at (-1, -0) {};
	\end{pgfonlayer}
	\begin{pgfonlayer}{edgelayer}
		\draw [in=27, out=180, looseness=1.00] (1.center) to (0);
		\draw [in=180, out=-27, looseness=1.00] (0) to (2.center);
		\draw (0) to (3.center);
	\end{pgfonlayer}
\end{tikzpicture}
$$

\item[Composition:]
\hfil
$
\begin{tikzpicture}
	\begin{pgfonlayer}{nodelayer}
		\node [style=map] (0) at (0, -0) {$f$};
		\node [style=none] (1) at (1, 0.5) {};
		\node [style=none] (2) at (1, -0.5) {};
		\node [style=none] (3) at (-1, -0) {};
	\end{pgfonlayer}
	\begin{pgfonlayer}{edgelayer}
		\draw [in=27, out=180, looseness=1.00] (1.center) to (0);
		\draw [in=180, out=-27, looseness=1.00] (0) to (2.center);
		\draw (0) to (3.center);
	\end{pgfonlayer}
\end{tikzpicture}
;
\begin{tikzpicture}
	\begin{pgfonlayer}{nodelayer}
		\node [style=map] (0) at (0, -0) {$g$};
		\node [style=none] (1) at (1, 0.5) {};
		\node [style=none] (2) at (1, -0.5) {};
		\node [style=none] (3) at (-1, -0) {};
	\end{pgfonlayer}
	\begin{pgfonlayer}{edgelayer}
		\draw [in=27, out=180, looseness=1.00] (1.center) to (0);
		\draw [in=180, out=-27, looseness=1.00] (0) to (2.center);
		\draw (0) to (3.center);
	\end{pgfonlayer}
\end{tikzpicture}
:=
\begin{tikzpicture}
	\begin{pgfonlayer}{nodelayer}
		\node [style=map] (0) at (0, -0) {$f$};
		\node [style=none] (1) at (1, -0.5) {};
		\node [style=none] (2) at (-1, -0) {};
		\node [style=map] (3) at (1, 0.5) {$g$};
		\node [style=none] (4) at (2, 1) {};
		\node [style=otimes] (5) at (2, -0) {};
		\node [style=none] (6) at (1, 0.5) {};
		\node [style=none] (7) at (3, 1) {};
		\node [style=none] (8) at (3, -0) {};
	\end{pgfonlayer}
	\begin{pgfonlayer}{edgelayer}
		\draw [in=180, out=-27, looseness=1.00] (0) to (1.center);
		\draw (0) to (2.center);
		\draw [in=27, out=180, looseness=1.00] (4.center) to (3);
		\draw (3) to (5.center);
		\draw [in=27, out=180, looseness=1.00] (6.center) to (0);
		\draw [in=-153, out=0, looseness=1.00] (1.center) to (5.center);
		\draw (5.center) to (8.center);
		\draw (4.center) to (7.center);
	\end{pgfonlayer}
\end{tikzpicture}
$

\item[Identity:]
\hfil
$
\begin{tikzpicture}
	\begin{pgfonlayer}{nodelayer}
		\node [style=none] (0) at (-1, -0) {};
		\node [style=none] (1) at (0.5, -0) {};
		\node [style=none] (2) at (0.5, -0.75) {};
		\node [style=none] (3) at (-0.25, -0) {};
	\end{pgfonlayer}
	\begin{pgfonlayer}{edgelayer}
		\draw [style=dashed, in=180, out=-75, looseness=1.00] (3.center) to (2.center);
		\draw [style=simple] (0.center) to (1.center);
	\end{pgfonlayer}
\end{tikzpicture}
$

\item[Restriction: ] 
\hfil
$
\bar{\left(
\begin{tikzpicture}
	\begin{pgfonlayer}{nodelayer}
		\node [style=map] (0) at (0, -0) {$f$};
		\node [style=none] (1) at (-1, -0) {};
		\node [style=none] (2) at (1, 0.5) {};
		\node [style=none] (3) at (1, -0.5) {};
	\end{pgfonlayer}
	\begin{pgfonlayer}{edgelayer}
		\draw [style=simple] (1.center) to (0);
		\draw [style=simple, in=27, out=180, looseness=1.00] (2.center) to (0);
		\draw [style=simple, in=-27, out=180, looseness=1.00] (3.center) to (0);
	\end{pgfonlayer}
\end{tikzpicture}
\right)}
:=
\begin{tikzpicture}
	\begin{pgfonlayer}{nodelayer}
		\node [style=map] (0) at (0, -0) {$\bar f$};
		\node [style=none] (1) at (-1, -0) {};
		\node [style=none] (2) at (1, -0) {};
		\node [style=none] (3) at (0.5, -0) {};
		\node [style=none] (4) at (1, -0.5) {};
	\end{pgfonlayer}
	\begin{pgfonlayer}{edgelayer}
		\draw [style=simple] (1.center) to (0);
		\draw [style=simple] (2.center) to (0);
		\draw [style=dashed, in=180, out=-75, looseness=1.00] (3.center) to (4.center);
	\end{pgfonlayer}
\end{tikzpicture}
$

\item[Restriction product:]
\hfil
$
\langle f,g \rangle:=
\begin{tikzpicture}
	\begin{pgfonlayer}{nodelayer}
		\node [style=map] (0) at (0, 0.25) {$f$};
		\node [style=none] (1) at (1, 0.25) {};
		\node [style=none] (2) at (1, -0.75) {};
		\node [style=none] (3) at (1, 0.25) {};
		\node [style=map] (4) at (0, -0.75) {$g$};
		\node [style=none] (5) at (1, -0.75) {};
		\node [style=otimes] (6) at (1, -0.75) {};
		\node [style=otimes] (7) at (1, 0.25) {};
		\node [style=X] (8) at (-1, -0.25) {};
		\node [style=none] (9) at (2, 0.25) {};
		\node [style=none] (10) at (2, -0.75) {};
		\node [style=none] (11) at (-2, -0.25) {};
	\end{pgfonlayer}
	\begin{pgfonlayer}{edgelayer}
		\draw [style=simple, in=27, out=150, looseness=1.00] (1.center) to (0);
		\draw [style=simple] (2.center) to (0);
		\draw [style=simple] (3.center) to (4);
		\draw [style=simple, in=-27, out=-150, looseness=1.00] (5.center) to (4);
		\draw [style=simple, in=-34, out=180, looseness=1.00] (4) to (8);
		\draw [style=simple, in=180, out=34, looseness=1.00] (8) to (0);
		\draw [style=simple] (9.center) to (1.center);
		\draw [style=simple] (2.center) to (10.center);
		\draw [style=simple] (8) to (11.center);
	\end{pgfonlayer}
\end{tikzpicture}
$

\item[Restriction terminal map:]
\hfil
$
\begin{tikzpicture}[yscale=-1]
	\begin{pgfonlayer}{nodelayer}
		\node [style=none] (0) at (-1, -0) {};
		\node [style=none] (1) at (0.5, -0) {};
		\node [style=none] (2) at (0.5, -0.75) {};
		\node [style=none] (3) at (-0.25, -0) {};
	\end{pgfonlayer}
	\begin{pgfonlayer}{edgelayer}
		\draw [style=dashed, in=180, out=-75, looseness=1.00] (3.center) to (2.center);
		\draw [style=simple] (0.center) to (1.center);
	\end{pgfonlayer}
\end{tikzpicture}
$

\item[Tensor product:]
\hfil
$
\begin{tikzpicture}
	\begin{pgfonlayer}{nodelayer}
		\node [style=map] (0) at (0, -0) {$f$};
		\node [style=none] (1) at (1, 0.5) {};
		\node [style=none] (2) at (1, -0.5) {};
		\node [style=none] (3) at (-1, -0) {};
	\end{pgfonlayer}
	\begin{pgfonlayer}{edgelayer}
		\draw [in=27, out=180, looseness=1.00] (1.center) to (0);
		\draw [in=180, out=-27, looseness=1.00] (0) to (2.center);
		\draw (0) to (3.center);
	\end{pgfonlayer}
\end{tikzpicture}
\otimes
\begin{tikzpicture}
	\begin{pgfonlayer}{nodelayer}
		\node [style=map] (0) at (0, -0) {$g$};
		\node [style=none] (1) at (1, 0.5) {};
		\node [style=none] (2) at (1, -0.5) {};
		\node [style=none] (3) at (-1, -0) {};
	\end{pgfonlayer}
	\begin{pgfonlayer}{edgelayer}
		\draw [in=27, out=180, looseness=1.00] (1.center) to (0);
		\draw [in=180, out=-27, looseness=1.00] (0) to (2.center);
		\draw (0) to (3.center);
	\end{pgfonlayer}
\end{tikzpicture}
:=
\begin{tikzpicture}
	\begin{pgfonlayer}{nodelayer}
		\node [style=map] (0) at (0, 0.25) {$f$};
		\node [style=none] (1) at (-1, 0.25) {};
		\node [style=map] (2) at (0, -0.5) {$g$};
		\node [style=none] (3) at (-1, -0.5) {};
		\node [style=otimes] (4) at (1, 0.25) {};
		\node [style=otimes] (5) at (1, -0.5) {};
		\node [style=none] (6) at (1.75, 0.25) {};
		\node [style=none] (7) at (1.75, -0.5) {};
	\end{pgfonlayer}
	\begin{pgfonlayer}{edgelayer}
		\draw (0) to (1.center);
		\draw (2) to (3.center);
		\draw [style=simple] (7.center) to (5);
		\draw [style=simple] (4) to (6.center);
		\draw [style=simple] (5) to (0);
		\draw [style=simple] (2) to (4);
		\draw [style=simple, bend left, looseness=1.00] (5) to (2);
		\draw [style=simple, bend left, looseness=1.25] (0) to (4);
	\end{pgfonlayer}
\end{tikzpicture}
$

\item[Tensor unit:]  The same as in $\X$.
\end{description}

\end{definition}

\begin{example}\cite[Ex. 5.3.3]{giles}
$\tilde \Pinj$ is $\Par$.
\end{example}
\begin{proof}
For a partial function $f:X\to Y$, $\{(x,(y,x)) | (x,y) \in f \}/\sim$ is a partial isomorphism.
\end{proof}

\begin{lemma}
\label{lemma:xtildefaithful}
The canonical functor $\iota:\X\to \tilde \X$ is faithful.
\end{lemma}
The proof is contained in \S \ref{proof:xtildefaithful}.

\begin{lemma}
The induced Frobenius algebra structure in $\tilde \X$ is counital.
\end{lemma}
\begin{proof}
For all $X$, the map $X \to (X\otimes X) \otimes I$ in $\tilde\X$ induced by the Frobenius algebra in $\X$ has a counit given by the  unitor $X\to I\otimes X$ since, in $\X$:
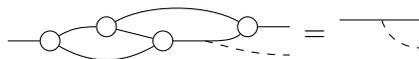
$$
\begin{tikzpicture}
	\begin{pgfonlayer}{nodelayer}
		\node [style=X] (0) at (3.75, -0) {};
		\node [style=none] (1) at (3, -0) {};
		\node [style=X] (2) at (4.75, 0.25) {};
		\node [style=X] (3) at (5.75, -0) {};
		\node [style=X] (4) at (7.25, 0.25) {};
		\node [style=none] (5) at (8, 0.25) {};
		\node [style=none] (6) at (8, -0.25) {};
		\node [style=none] (7) at (6.5, -0) {};
	\end{pgfonlayer}
	\begin{pgfonlayer}{edgelayer}
		\draw (0) to (1.center);
		\draw [in=180, out=30, looseness=1.00] (0) to (2);
		\draw (2) to (3);
		\draw (5.center) to (4);
		\draw [in=30, out=150, looseness=0.75] (4) to (2);
		\draw [in=-30, out=-150, looseness=1.00] (3) to (0);
		\draw (3) to (7.center);
		\draw [in=-135, out=0, looseness=0.75] (7.center) to (4);
		\draw [style=dashed, in=180, out=-15, looseness=1.00] (7.center) to (6.center);
	\end{pgfonlayer}
\end{tikzpicture}
=
\begin{tikzpicture}
	\begin{pgfonlayer}{nodelayer}
		\node [style=none] (0) at (3, -0) {};
		\node [style=none] (1) at (4.5, -0) {};
		\node [style=none] (2) at (4.5, -0.5) {};
		\node [style=none] (3) at (3.75, -0) {};
	\end{pgfonlayer}
	\begin{pgfonlayer}{edgelayer}
		\draw [style=dashed, in=180, out=-75, looseness=1.00] (3.center) to (2.center);
		\draw [style=simple] (3.center) to (0.center);
		\draw [style=simple] (3.center) to (1.center);
	\end{pgfonlayer}
\end{tikzpicture}
$$
\end{proof}

\section{Categorical quantum mechanics and completely positive maps}
\label{sec:cpm}
The $\sf CPM$ construction gives a notion of quantum channels for any $\dag$-compact closed category \cite{selinger2007dagger}.
The \dag-Frobenius algebras in the base category induce idempotents in $\sf CPM$ corresponding to decohering quantum channels.  By considering the full subcategory of the Karoubi envelope whose objects are such idempotents one obtains the $\STOCH$ construction of \cite{coecke2016categories}: yielding classical channels between finite dimensional $C^*$-algebras when applied to $\FHilb$.   However, the $\sf CPM$ construction  can not be applied to $\Hilb$ in general because unlike $\FHilb$, it is not compact closed. 
The $\CP^\infty$ construction  \cite{coecke2016pictures} generalizes the $\sf CPM$ construction to (non compact closed) $\dag$-symmetric monoidal categories, by unbending the cups/caps and, identifying two super-maps  when they act the same on all positive test maps: recovering the usual notion of purely quantum channels.

\begin{figure}

$$
\begin{tikzpicture}
	\begin{pgfonlayer}{nodelayer}
		\node [style=map] (0) at (-4, -0.25) {$f$};
		\node [style=none] (1) at (-4.75, -0.25) {};
		\node [style=none] (2) at (-3.25, -0) {};
		\node [style=none] (3) at (-3.25, -0.5) {};
	\end{pgfonlayer}
	\begin{pgfonlayer}{edgelayer}
		\draw [style=simple] (1.center) to (0);
		\draw [style=simple, in=180, out=34, looseness=1.00] (0) to (2.center);
		\draw [style=simple, in=-34, out=180, looseness=1.00] (3.center) to (0);
	\end{pgfonlayer}
\end{tikzpicture}
;
\begin{tikzpicture}
	\begin{pgfonlayer}{nodelayer}
		\node [style=map] (0) at (-4, -0.25) {$g$};
		\node [style=none] (1) at (-4.75, -0.25) {};
		\node [style=none] (2) at (-3.25, -0) {};
		\node [style=none] (3) at (-3.25, -0.5) {};
	\end{pgfonlayer}
	\begin{pgfonlayer}{edgelayer}
		\draw [style=simple] (1.center) to (0);
		\draw [style=simple, in=180, out=34, looseness=1.00] (0) to (2.center);
		\draw [style=simple, in=-34, out=180, looseness=1.00] (3.center) to (0);
	\end{pgfonlayer}
\end{tikzpicture}
:=
\begin{tikzpicture}
	\begin{pgfonlayer}{nodelayer}
		\node [style=map] (0) at (-4, -0.25) {$f$};
		\node [style=none] (1) at (-4.75, -0.25) {};
		\node [style=none] (2) at (-3.25, -0) {};
		\node [style=none] (3) at (-2.5, -0.5) {};
		\node [style=none] (4) at (-2.5, -0.5) {};
		\node [style=none] (5) at (-2, 0.25) {};
		\node [style=map] (6) at (-3.25, -0) {$g$};
		\node [style=otimes] (7) at (-2.5, -0.5) {};
		\node [style=none] (8) at (-2, -0.5) {};
	\end{pgfonlayer}
	\begin{pgfonlayer}{edgelayer}
		\draw [style=simple] (1.center) to (0);
		\draw [style=simple, in=180, out=34, looseness=1.00] (0) to (2.center);
		\draw [style=simple, in=-34, out=180, looseness=1.00] (3.center) to (0);
		\draw [style=simple, in=180, out=34, looseness=1.00] (6) to (5.center);
		\draw [style=simple, in=-34, out=150, looseness=1.00] (4.center) to (6);
		\draw [style=simple] (8.center) to (3.center);
	\end{pgfonlayer}
\end{tikzpicture}
\hspace*{.5cm}
\begin{tikzpicture}
	\begin{pgfonlayer}{nodelayer}
		\node [style=map] (0) at (-4, -0.25) {$h$};
		\node [style=none] (1) at (-4.75, -0.25) {};
		\node [style=none] (2) at (-2, 0.25) {};
		\node [style=map] (3) at (-2.75, -0.25) {$h^\circ$};
		\node [style=none] (4) at (-2, -0.25) {};
		\node [style=none] (5) at (-4.75, 0.25) {};
	\end{pgfonlayer}
	\begin{pgfonlayer}{edgelayer}
		\draw [style=simple] (1.center) to (0);
		\draw [style=simple, in=180, out=34, looseness=0.75] (0) to (2.center);
		\draw [style=simple] (4.center) to (3);
		\draw [style=simple, in=0, out=146, looseness=0.75] (3) to (5.center);
		\draw [bend right, looseness=0.75] (0) to (3);
	\end{pgfonlayer}
\end{tikzpicture}
=
\begin{tikzpicture}
	\begin{pgfonlayer}{nodelayer}
		\node [style=map] (0) at (-4, -0.25) {$k$};
		\node [style=none] (1) at (-4.75, -0.25) {};
		\node [style=none] (2) at (-2, 0.25) {};
		\node [style=map] (3) at (-2.75, -0.25) {$k^\circ$};
		\node [style=none] (4) at (-2, -0.25) {};
		\node [style=none] (5) at (-4.75, 0.25) {};
	\end{pgfonlayer}
	\begin{pgfonlayer}{edgelayer}
		\draw [style=simple] (1.center) to (0);
		\draw [style=simple, in=180, out=34, looseness=0.75] (0) to (2.center);
		\draw [style=simple] (4.center) to (3);
		\draw [style=simple, in=0, out=146, looseness=0.75] (3) to (5.center);
		\draw [bend right, looseness=0.75] (0) to (3);
	\end{pgfonlayer}
\end{tikzpicture}
\hspace*{.5cm}
\begin{tikzpicture}
	\begin{pgfonlayer}{nodelayer}
		\node [style=X] (0) at (-7, -0) {};
		\node [style=none] (1) at (-7.5, -0) {};
		\node [style=none] (2) at (-6.25, 0.25) {};
		\node [style=none] (3) at (-6.25, -0.25) {};
	\end{pgfonlayer}
	\begin{pgfonlayer}{edgelayer}
		\draw (1.center) to (0);
		\draw [in=180, out=18, looseness=1.00] (0) to (2.center);
		\draw [in=-18, out=180, looseness=0.75] (3.center) to (0);
	\end{pgfonlayer}
\end{tikzpicture}
$$

\caption{
Composition of representatives $f;g$;  equivalence relation $h\sim k$; decoherence map.}
\label{fig:kraus}
\end{figure}

To generalize the $\STOCH$ construction to  \dag-semi-Frobenius algebras, one must combine the  
 $\STOCH$ and $\CP^\infty$ constructions, as the compact closed structure is no longer taken for granted.   We show that the Cartesian completion is the same as first applying a modified version of the  $\CP^\infty$ construction (without quantifying over all test maps, as seen in Figure \ref{fig:kraus}) to a discrete inverse category and then taking the full subcategory of the Karoubi envelope whose objects are  the decoherence maps \footnote{Although, composition in this version of the ${\sf CP}^\infty$ construction, without universally quantifying over test maps, when applied to a discrete inverse category is not obviously well-defined unless the base category embeds in a compact closed category.}.  The following Lemma is needed to prove this fact:

\begin{lemma}
\label{lem:latching}

Given two parallel maps $X\xrightarrow{f,g} Y\otimes Z$ in a discrete inverse category:

$$

$$
\end{lemma}

\begin{proof}
The one direction is trivial, for the other direction:

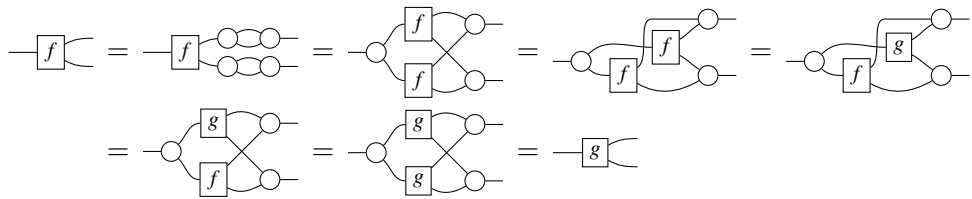
\begin{align*}
%
\end{align*}
\end{proof}

\begin{lemma}
\label{theorem:cpstartheorem}
Given two maps $X \xrightarrow{f} Y\otimes S$ and $X \xrightarrow{g} Y\otimes T$, in a discrete inverse category:

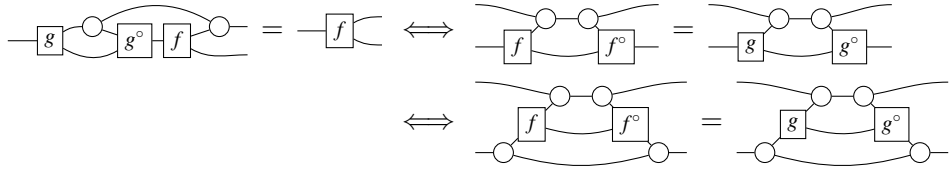
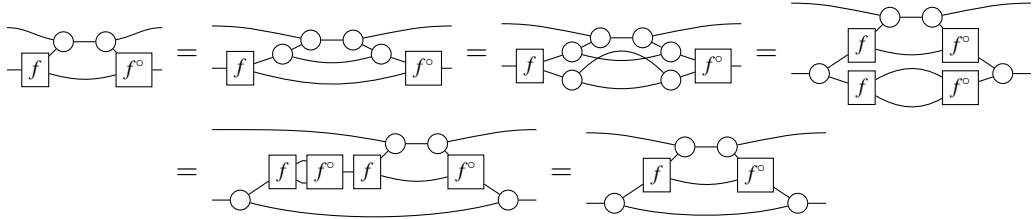
\begin{align*}
%
\end{align*}

So that we only have to prove the first biconditional.
Suppose that the left hand side holds, then:

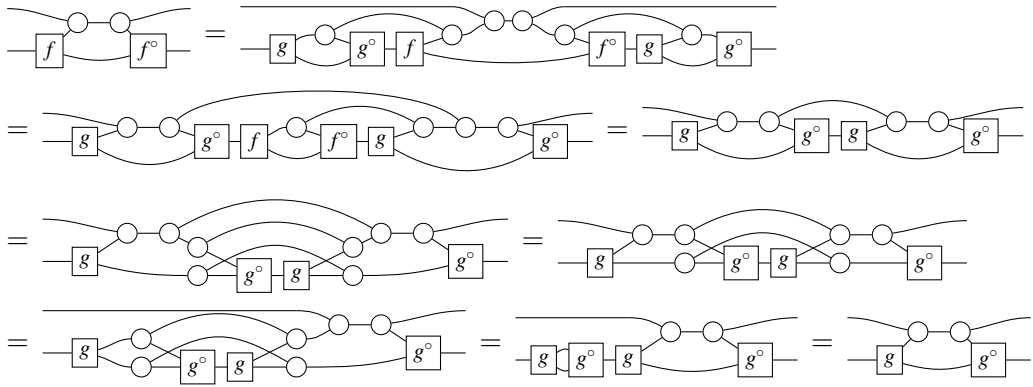
\begin{align*}
&
%
\end{align*}

Conversely, suppose that the right hand side holds.  Then:

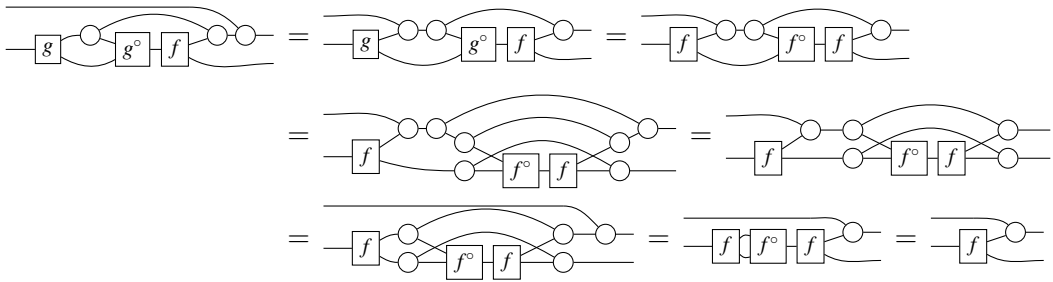
\begin{align*}
%
\\
\\
\text{Thus, by Lemma \ref{lem:latching}}:&
%
\end{align*}

\end{proof}

\label{sec:env}

The natural question arises: can we characterize classical channels in this setting, algebraically in terms of a discarding morphism, without performing any doubling.  In other words, is there some notion of ``environment structure'' \cite{coecke2010environment} for the {\em classical} channels of discrete inverse categories:

\begin{definition}
Given a discrete inverse category $\X$, define the counital completion of $\X$, $c(\X)$ to have the same objects and maps of $\X$, except with a freely adjoined counit $!_X:X\to I$ to the chosen semi-Frobenius algebra on $X$, for each object in $\X$ compatible with the monoidal structure.
\end{definition}

\begin{lemma}
$c(\X)$ is a discrete Cartesian restriction category.
\end{lemma}
\begin{proof}
This is clearly a counital copy category, with a restriction terminal object given by the tensor unit.  Moreover, because the Frobenius structure is special, it is also discrete.
\end{proof}

\begin{lemma}
\label{lemma:envstruct}
Given a discrete inverse category $\X$, $c(\X)$ and $\tilde \X$ are isomorphic as discrete Cartesian restriction categories.
\end{lemma}

The proof is contained in \S \ref{proof:envstruct}.

\section{\texorpdfstring{$\ZXA$}{ZX\&}}
\label{sec:ZXA}
In this section, we add a unit and counit to the Frobenius algebra in $\TOF$ by glueing its counital completion and unital completion together.  We then give a presentation of this category in terms of the self-dual compact closed prop $\ZXA$ generated by the copy and addition spiders, the not gate and the {\sf and} gate via a two-way translation.

\begin{definition}\cite{tof}
The category $\TOF$ is the prop generated by the Toffoli gate and ancillary bits, satisfying the equations in \S \ref{sec:tof} Figure \ref{fig:TOF}.
\end{definition}

\begin{thm}\cite{tof}
$\TOF$ is isomorphic to the category of partial isomorphisms between ordinals $2^n$, $n\in \N$.
\end{thm}

By adding a unit and counit, we obtain a full subcategory of spans of sets and finite ordinals:

\begin{lemma}
\label{lemma:unitcounit}

 The full subcategory of $\Span^\sim(\FinOrd)$ generated by powers of 2 is presented by the pushout,  $\hat \TOF$, of the following diagram of props:

$$c(\TOF)^\op \leftarrow \TOF \rightarrow c(\TOF)$$
\end{lemma}

The proof is contained in \S\ref{proof:unitcounit}.

If $f$ is a partial isomorphism between finite sets, then the white spiders correspond to the classical structure for the chosen computational basis.  For the interpretation into $\FHilb$ via the $\ell_2$ functor, this means that in the  qubit case, the unit and counit correspond to $\sqrt{2}|+\rangle$ and $\sqrt{2}\langle +|$.

We give a more elegant presentation of this category in terms of interacting monoids and 
 comonoids:

\begin{definition}
Consider the self dual prop $\ZXA$ generated by the addition spider with phases in $\{0,\pi\}$, the copy spider and the monoid for conjunction satisfying the  identities given in Figure \ref{fig:ZXA}.

\begin{figure}
	\noindent
	\scalebox{1.0}{%
		\vbox{%
			\begin{mdframed}
				\begin{multicols}{2}
					\begin{enumerate}[label={\bf [ZX{\it \&}.\arabic*]}, ref={\bf [ZX{\it \&}.\arabic*]}, wide = 0pt, leftmargin = 2em]
						\item
						\label{ZXA.1}
						{\hfil
							$

							$
						}

					\end{enumerate}
				\end{multicols}
				\
			\end{mdframed}
	}}
	\caption{The identities of \texorpdfstring{$\ZXA$}{ZX\&}, where \texorpdfstring{$\alpha,\beta \in \{0,\pi\}$}{alpha and beta are either 0 or pi} and a blank grey spider has angle 0.}
	\label{fig:ZXA}
\end{figure}

\end{definition}
One can interpret the generators as logical connectives and open wires as variables, similar to the regular logic \cite{butz}, or the logic of a Cartesian bicategory \cite{carboni}, except we forget the 2-cells in $\ZXA$.  The decorated black spiders correspond to fixed variables and xor.  White (co)multiplications (co)copy variables; the white unit is existential quantification and the counit is discarding. The relations are open $\Sigma_1$ Boolean formulas augmented with copying and discarding as well as duals; the open variables correspond to distinguished inputs and outputs.


The identities of $\ZXA$ can also be interpreted by freely taking the coproduct of the free prop of commutative (co)monoids \dag-PROP $3\times 2$ times, modulo various (undirected) distributive laws, and monoid maps.  The distributive laws are summarized in Figure \ref{fig:table} (the duals under diagonal are omitted). Te spider rules implicitly identify the (co)units of the \dag-compact closed structure induced by $Z$ and $X$; which is needed for completeness.
%
%

\begin{figure}[H]

\begin{minipage}[b]{\textwidth}
\setlength\dashlinedash{0.2pt}
\setlength\dashlinegap{1.5pt}
\setlength\arrayrulewidth{0.3pt}
\resizebox{\textwidth}{!}{%
\begin{tabular}{l|l:l:l:p{45mm}:p{38mm}:l}
    $\lambda$    & $Z$    & $X$    & $\&$      & $Z^\dag$                                        & $X^\dag$                           & $\&^\dag$\\ \hline
$Z$       & Comm. monoid &        &           & \noindent\begin{tabular}{@{}l} Extra special comm.\\ \dag-Frobenius algebra\end{tabular} &                                   Hopf algebra with $s=1$ &  Special bialgebra  \\ \hdashline
$X$       & \bcell & Comm. monoid &           & Hopf algebra with $s=1$              &  \noindent\begin{tabular}{@{}l} Comm. \dag-Frobenius\\ algebra \end{tabular}&         \\ \hdashline
$\&$      & \bcell & \bcell & Comm. monoid    & Special bialgebra                                       &                                    &         \\ \hdashline
$Z^\dag$  & \bcell & \bcell &   \bcell  & Cocomm. comonoid                                          &                                    &         \\ \hdashline
$X^\dag$  & \bcell & \bcell &   \bcell  &             \bcell                              & Cocomm. comonoid                             &         \\ \hdashline
$\&^\dag$ & \bcell & \bcell &   \bcell  &              \bcell                             &           \bcell                   &  Cocomm. comonoid       \\
\end{tabular}
}
\end{minipage}
\caption{Generating distributive laws of \texorpdfstring{$\ZXA$}{ZX\&}.}
\label{fig:table}
\end{figure}

Additionally, \ref{ZXA.16} states that the counit of $\&^\dag$ is copied by $\&$; ie. the counit is a monad map from $\&$ to the trivial monad.  
\ref{ZXA.17} expreses the multiplication part of the distributive law of Lawvere theories between the props for multiplication and addition mod 2 (see \cite{lawvere} for distributive laws of Lawvere theories).

\begin{proposition}
\label{prop:TOFZXA}
Consider the interpretation $\llbracket\_\rrbracket_{\ZXA}:\ZXA\to\hat \TOF$ taking:

\begin{center}
\begin{tabular}{c}
$
\begin{tikzpicture}
	\begin{pgfonlayer}{nodelayer}
		\node [style=none] (0) at (2, -0.25) {};
		\node [style=none] (1) at (2, -0.75) {};
		\node [style=Z] (2) at (1.25, -0.5) {};
		\node [style=none] (3) at (0.5, -0.5) {};
	\end{pgfonlayer}
	\begin{pgfonlayer}{edgelayer}
		\draw [style=simple, in=180, out=34, looseness=1.00] (2) to (0.center);
		\draw [style=simple, in=180, out=-34, looseness=1.00] (2) to (1.center);
		\draw [style=simple] (2) to (3.center);
	\end{pgfonlayer}
\end{tikzpicture}
\mapsto
\begin{tikzpicture}
	\begin{pgfonlayer}{nodelayer}
		\node [style=none] (0) at (2, -0) {};
		\node [style=none] (1) at (0.5, -0) {};
		\node [style=dot] (2) at (1.25, -0.5) {};
		\node [style=oplus] (3) at (1.25, -0) {};
		\node [style=X] (4) at (0.5, -0.5) {};
		\node [style=none] (5) at (2, -0.5) {};
	\end{pgfonlayer}
	\begin{pgfonlayer}{edgelayer}
		\draw [style=simple] (5.center) to (2);
		\draw [style=simple] (2) to (4);
		\draw [style=simple] (2) to (3);
		\draw [style=simple] (3) to (1.center);
		\draw [style=simple] (3) to (0.center);
	\end{pgfonlayer}
\end{tikzpicture}
\hspace*{.5cm}
\begin{tikzpicture}[xscale=-1]
	\begin{pgfonlayer}{nodelayer}
		\node [style=none] (0) at (2, -0.25) {};
		\node [style=none] (1) at (2, -0.75) {};
		\node [style=Z] (2) at (1.25, -0.5) {};
		\node [style=none] (3) at (0.5, -0.5) {};
	\end{pgfonlayer}
	\begin{pgfonlayer}{edgelayer}
		\draw [style=simple, in=180, out=34, looseness=1.00] (2) to (0.center);
		\draw [style=simple, in=180, out=-34, looseness=1.00] (2) to (1.center);
		\draw [style=simple] (2) to (3.center);
	\end{pgfonlayer}
\end{tikzpicture}
\mapsto
\begin{tikzpicture}[xscale=-1]
	\begin{pgfonlayer}{nodelayer}
		\node [style=none] (0) at (2, -0) {};
		\node [style=none] (1) at (0.5, -0) {};
		\node [style=dot] (2) at (1.25, -0.5) {};
		\node [style=oplus] (3) at (1.25, -0) {};
		\node [style=X] (4) at (0.5, -0.5) {};
		\node [style=none] (5) at (2, -0.5) {};
	\end{pgfonlayer}
	\begin{pgfonlayer}{edgelayer}
		\draw [style=simple] (5.center) to (2);
		\draw [style=simple] (2) to (4);
		\draw [style=simple] (2) to (3);
		\draw [style=simple] (3) to (1.center);
		\draw [style=simple] (3) to (0.center);
	\end{pgfonlayer}
\end{tikzpicture}
\hspace*{.5cm}
\begin{tikzpicture}
	\begin{pgfonlayer}{nodelayer}
		\node [style=Z] (0) at (1, -0) {};
		\node [style=none] (1) at (2, -0) {};
	\end{pgfonlayer}
	\begin{pgfonlayer}{edgelayer}
		\draw [style=simple] (1.center) to (0);
	\end{pgfonlayer}
\end{tikzpicture}
\mapsto
\begin{tikzpicture}
	\begin{pgfonlayer}{nodelayer}
		\node [style=zeroin] (0) at (1, -0) {};
		\node [style=none] (1) at (2, -0) {};
	\end{pgfonlayer}
	\begin{pgfonlayer}{edgelayer}
		\draw [style=simple] (1.center) to (0);
	\end{pgfonlayer}
\end{tikzpicture}
\hspace*{.5cm}
\begin{tikzpicture}[xscale=-1]
	\begin{pgfonlayer}{nodelayer}
		\node [style=Z] (0) at (1, -0) {};
		\node [style=none] (1) at (2, -0) {};
	\end{pgfonlayer}
	\begin{pgfonlayer}{edgelayer}
		\draw [style=simple] (1.center) to (0);
	\end{pgfonlayer}
\end{tikzpicture}
\mapsto
\begin{tikzpicture}[xscale=-1]
	\begin{pgfonlayer}{nodelayer}
		\node [style=zeroin] (0) at (1, -0) {};
		\node [style=none] (1) at (2, -0) {};
	\end{pgfonlayer}
	\begin{pgfonlayer}{edgelayer}
		\draw [style=simple] (1.center) to (0);
	\end{pgfonlayer}
\end{tikzpicture}
$
\\\\
$
\begin{tikzpicture}
	\begin{pgfonlayer}{nodelayer}
		\node [style=none] (0) at (2, -0.25) {};
		\node [style=none] (1) at (2, -0.75) {};
		\node [style=X] (2) at (1.25, -0.5) {};
		\node [style=none] (3) at (0.5, -0.5) {};
	\end{pgfonlayer}
	\begin{pgfonlayer}{edgelayer}
		\draw [style=simple, in=180, out=34, looseness=1.00] (2) to (0.center);
		\draw [style=simple, in=180, out=-34, looseness=1.00] (2) to (1.center);
		\draw [style=simple] (2) to (3.center);
	\end{pgfonlayer}
\end{tikzpicture}
\mapsto
\begin{tikzpicture}
	\begin{pgfonlayer}{nodelayer}
		\node [style=none] (0) at (2, -0) {};
		\node [style=none] (1) at (0.5, -0) {};
		\node [style=oplus] (2) at (1.25, -0.5) {};
		\node [style=dot] (3) at (1.25, -0) {};
		\node [style=zeroin] (4) at (0.5, -0.5) {};
		\node [style=none] (5) at (2, -0.5) {};
	\end{pgfonlayer}
	\begin{pgfonlayer}{edgelayer}
		\draw [style=simple] (5.center) to (2);
		\draw [style=simple] (2) to (4);
		\draw [style=simple] (2) to (3);
		\draw [style=simple] (3) to (1.center);
		\draw [style=simple] (3) to (0.center);
	\end{pgfonlayer}
\end{tikzpicture}
\hspace*{.5cm}
\begin{tikzpicture}[xscale=-1]
	\begin{pgfonlayer}{nodelayer}
		\node [style=none] (0) at (2, -0.25) {};
		\node [style=none] (1) at (2, -0.75) {};
		\node [style=X] (2) at (1.25, -0.5) {};
		\node [style=none] (3) at (0.5, -0.5) {};
	\end{pgfonlayer}
	\begin{pgfonlayer}{edgelayer}
		\draw [style=simple, in=180, out=34, looseness=1.00] (2) to (0.center);
		\draw [style=simple, in=180, out=-34, looseness=1.00] (2) to (1.center);
		\draw [style=simple] (2) to (3.center);
	\end{pgfonlayer}
\end{tikzpicture}
\mapsto
\begin{tikzpicture}[xscale=-1]
	\begin{pgfonlayer}{nodelayer}
		\node [style=none] (0) at (2, -0) {};
		\node [style=none] (1) at (0.5, -0) {};
		\node [style=oplus] (2) at (1.25, -0.5) {};
		\node [style=dot] (3) at (1.25, -0) {};
		\node [style=zeroin] (4) at (0.5, -0.5) {};
		\node [style=none] (5) at (2, -0.5) {};
	\end{pgfonlayer}
	\begin{pgfonlayer}{edgelayer}
		\draw [style=simple] (5.center) to (2);
		\draw [style=simple] (2) to (4);
		\draw [style=simple] (2) to (3);
		\draw [style=simple] (3) to (1.center);
		\draw [style=simple] (3) to (0.center);
	\end{pgfonlayer}
\end{tikzpicture}
\hspace*{.5cm}
\begin{tikzpicture}
	\begin{pgfonlayer}{nodelayer}
		\node [style=X] (0) at (1, -0) {};
		\node [style=none] (1) at (2, -0) {};
	\end{pgfonlayer}
	\begin{pgfonlayer}{edgelayer}
		\draw [style=simple] (1.center) to (0);
	\end{pgfonlayer}
\end{tikzpicture}
\mapsto
\begin{tikzpicture}
	\begin{pgfonlayer}{nodelayer}
		\node [style=X] (0) at (1, -0) {};
		\node [style=none] (1) at (2, -0) {};
	\end{pgfonlayer}
	\begin{pgfonlayer}{edgelayer}
		\draw [style=simple] (1.center) to (0);
	\end{pgfonlayer}
\end{tikzpicture}
\hspace*{.5cm}
\begin{tikzpicture}[xscale=-1]
	\begin{pgfonlayer}{nodelayer}
		\node [style=X] (0) at (1, -0) {};
		\node [style=none] (1) at (2, -0) {};
	\end{pgfonlayer}
	\begin{pgfonlayer}{edgelayer}
		\draw [style=simple] (1.center) to (0);
	\end{pgfonlayer}
\end{tikzpicture}
\mapsto
\begin{tikzpicture}[xscale=-1]
	\begin{pgfonlayer}{nodelayer}
		\node [style=X] (0) at (1, -0) {};
		\node [style=none] (1) at (2, -0) {};
	\end{pgfonlayer}
	\begin{pgfonlayer}{edgelayer}
		\draw [style=simple] (1.center) to (0);
	\end{pgfonlayer}
\end{tikzpicture}
$
\\\\
$
\begin{tikzpicture}
	\begin{pgfonlayer}{nodelayer}
		\node [style=Z] (0) at (1, -0) {$\pi$};
		\node [style=none] (1) at (2, -0) {};
		\node [style=none] (2) at (0, -0) {};
	\end{pgfonlayer}
	\begin{pgfonlayer}{edgelayer}
		\draw [style=simple] (1.center) to (0);
		\draw [style=simple] (0) to (2.center);
	\end{pgfonlayer}
\end{tikzpicture}
\mapsto
\begin{tikzpicture}
	\begin{pgfonlayer}{nodelayer}
		\node [style=oplus] (0) at (1, -0) {};
		\node [style=none] (1) at (2, -0) {};
		\node [style=none] (2) at (0, -0) {};
	\end{pgfonlayer}
	\begin{pgfonlayer}{edgelayer}
		\draw [style=simple] (1.center) to (0);
		\draw [style=simple] (0) to (2.center);
	\end{pgfonlayer}
\end{tikzpicture}
\hspace*{.5cm}
\begin{tikzpicture}
	\begin{pgfonlayer}{nodelayer}
		\node [style=none] (0) at (0, -0) {};
		\node [style=andin] (1) at (0.75, -0.5) {};
		\node [style=none] (2) at (1.5, -0.5) {};
		\node [style=none] (3) at (0, -1) {};
	\end{pgfonlayer}
	\begin{pgfonlayer}{edgelayer}
		\draw [style=simple] (2.center) to (1);
		\draw [style=simple, in=0, out=146, looseness=1.00] (1) to (0.center);
		\draw [style=simple, in=0, out=-146, looseness=1.00] (1) to (3.center);
	\end{pgfonlayer}
\end{tikzpicture}
\mapsto
\begin{tikzpicture}
	\begin{pgfonlayer}{nodelayer}
		\node [style=dot] (0) at (1, -0) {};
		\node [style=dot] (1) at (1, -0.5) {};
		\node [style=oplus] (2) at (1, -1) {};
		\node [style=X] (3) at (1.75, -0) {};
		\node [style=X] (4) at (1.75, -0.5) {};
		\node [style=none] (5) at (2, -1) {};
		\node [style=zeroin] (6) at (0.25, -1) {};
		\node [style=none] (7) at (0, -0) {};
		\node [style=none] (8) at (0, -0.5) {};
	\end{pgfonlayer}
	\begin{pgfonlayer}{edgelayer}
		\draw [style=simple] (5.center) to (2);
		\draw [style=simple] (2) to (6);
		\draw [style=simple] (2) to (1);
		\draw [style=simple] (1) to (0);
		\draw [style=simple] (3) to (0);
		\draw [style=simple] (0) to (7.center);
		\draw [style=simple] (8.center) to (1);
		\draw [style=simple] (1) to (4);
	\end{pgfonlayer}
\end{tikzpicture}
\hspace*{.5cm}
\begin{tikzpicture}
	\begin{pgfonlayer}{nodelayer}
		\node [style=none] (0) at (1.5, -0) {};
		\node [style=andout] (1) at (0.75, -0.5) {};
		\node [style=none] (2) at (0, -0.5) {};
		\node [style=none] (3) at (1.5, -1) {};
	\end{pgfonlayer}
	\begin{pgfonlayer}{edgelayer}
		\draw [style=simple] (2.center) to (1);
		\draw [style=simple, in=180, out=34, looseness=1.00] (1) to (0.center);
		\draw [style=simple, in=180, out=-34, looseness=1.00] (1) to (3.center);
	\end{pgfonlayer}
\end{tikzpicture}
\mapsto
\begin{tikzpicture}
	\begin{pgfonlayer}{nodelayer}
		\node [style=dot] (0) at (1, -0) {};
		\node [style=dot] (1) at (1, -0.5) {};
		\node [style=oplus] (2) at (1, -1) {};
		\node [style=X] (3) at (0.25, -0) {};
		\node [style=X] (4) at (0.25, -0.5) {};
		\node [style=none] (5) at (0, -1) {};
		\node [style=zeroout] (6) at (1.75, -1) {};
		\node [style=none] (7) at (2, -0) {};
		\node [style=none] (8) at (2, -0.5) {};
	\end{pgfonlayer}
	\begin{pgfonlayer}{edgelayer}
		\draw [style=simple] (5.center) to (2);
		\draw [style=simple] (2) to (6);
		\draw [style=simple] (2) to (1);
		\draw [style=simple] (1) to (0);
		\draw [style=simple] (3) to (0);
		\draw [style=simple] (0) to (7.center);
		\draw [style=simple] (8.center) to (1);
		\draw [style=simple] (1) to (4);
	\end{pgfonlayer}
\end{tikzpicture}
$
\end{tabular}
\end{center}

This interpretation is a strict symmetric \dag-monoidal functor.
\end{proposition}

See \S \ref{proof:TOFZXA} for the proof.

\begin{proposition}
\label{prop:ZXATOF}
Consider the interpretation $\llbracket\_\rrbracket_{\hat \TOF}:\hat \TOF\to \ZXA$ taking:

\begin{center}
\begin{tabular}{c}
$
\begin{tikzpicture}
	\begin{pgfonlayer}{nodelayer}
		\node [style=dot] (0) at (0, -0) {};
		\node [style=oplus] (1) at (0, -0.5) {};
		\node [style=dot] (2) at (0, 0.5) {};
		\node [style=none] (3) at (0.75, -0.5) {};
		\node [style=none] (4) at (0.75, -0) {};
		\node [style=none] (5) at (0.75, 0.5) {};
		\node [style=none] (6) at (-0.75, 0.5) {};
		\node [style=none] (7) at (-0.75, -0) {};
		\node [style=none] (8) at (-0.75, -0.5) {};
	\end{pgfonlayer}
	\begin{pgfonlayer}{edgelayer}
		\draw [style=simple] (3.center) to (1);
		\draw [style=simple] (1) to (0);
		\draw [style=simple] (0) to (2);
		\draw [style=simple] (2) to (5.center);
		\draw [style=simple] (4.center) to (0);
		\draw [style=simple] (0) to (7.center);
		\draw [style=simple] (6.center) to (2);
		\draw [style=simple] (1) to (8.center);
	\end{pgfonlayer}
\end{tikzpicture}
\mapsto
\begin{tikzpicture}
	\begin{pgfonlayer}{nodelayer}
		\node [style=none] (0) at (0, -0) {};
		\node [style=none] (1) at (0, -1) {};
		\node [style=none] (2) at (0, -1.5) {};
		\node [style=Z] (3) at (2, -1.5) {};
		\node [style=X] (4) at (0.5, -1) {};
		\node [style=X] (5) at (0.5, -0) {};
		\node [style=andin] (6) at (1.5, -0.5) {};
		\node [style=none] (7) at (2.75, -0) {};
		\node [style=none] (8) at (2.75, -1.5) {};
		\node [style=none] (9) at (2.75, -1) {};
	\end{pgfonlayer}
	\begin{pgfonlayer}{edgelayer}
		\draw [style=simple, in=0, out=90, looseness=0.75] (3) to (6);
		\draw [style=simple, in=-45, out=150, looseness=1.00] (6) to (5);
		\draw [style=simple] (5) to (0.center);
		\draw [style=simple] (1.center) to (4);
		\draw [style=simple] (3) to (2.center);
		\draw [style=simple] (9.center) to (4);
		\draw [style=simple] (8.center) to (3);
		\draw [style=simple, in=-150, out=45, looseness=1.00] (4) to (6);
		\draw [style=simple] (7.center) to (5);
	\end{pgfonlayer}
\end{tikzpicture}
\hspace*{.5cm}
\begin{tikzpicture}
	\begin{pgfonlayer}{nodelayer}
		\node [style=onein] (0) at (0, -0) {};
		\node [style=none] (1) at (1, -0) {};
	\end{pgfonlayer}
	\begin{pgfonlayer}{edgelayer}
		\draw [style=simple] (1.center) to (0);
	\end{pgfonlayer}
\end{tikzpicture}
\mapsto
\begin{tikzpicture}
	\begin{pgfonlayer}{nodelayer}
		\node [style=Z] (0) at (0, -0) {$\pi$};
		\node [style=none] (1) at (1, -0) {};
	\end{pgfonlayer}
	\begin{pgfonlayer}{edgelayer}
		\draw [style=simple] (1.center) to (0);
	\end{pgfonlayer}
\end{tikzpicture}
\hspace*{.5cm}
\begin{tikzpicture}[xscale=-1]
	\begin{pgfonlayer}{nodelayer}
		\node [style=onein] (0) at (0, -0) {};
		\node [style=none] (1) at (1, -0) {};
	\end{pgfonlayer}
	\begin{pgfonlayer}{edgelayer}
		\draw [style=simple] (1.center) to (0);
	\end{pgfonlayer}
\end{tikzpicture}
\mapsto
\begin{tikzpicture}
	\begin{pgfonlayer}{nodelayer}
		\node [style=Z] (0) at (1, -0) {$\pi$};
		\node [style=none] (1) at (0, -0) {};
	\end{pgfonlayer}
	\begin{pgfonlayer}{edgelayer}
		\draw [style=simple] (1.center) to (0);
	\end{pgfonlayer}
\end{tikzpicture}
$
\\\\
$
\begin{tikzpicture}
	\begin{pgfonlayer}{nodelayer}
		\node [style=X] (0) at (0, -0) {};
		\node [style=none] (1) at (1, -0) {};
	\end{pgfonlayer}
	\begin{pgfonlayer}{edgelayer}
		\draw [style=simple] (1.center) to (0);
	\end{pgfonlayer}
\end{tikzpicture}
\mapsto
\begin{tikzpicture}
	\begin{pgfonlayer}{nodelayer}
		\node [style=X] (0) at (0, -0) {};
		\node [style=none] (1) at (1, -0) {};
	\end{pgfonlayer}
	\begin{pgfonlayer}{edgelayer}
		\draw [style=simple] (1.center) to (0);
	\end{pgfonlayer}
\end{tikzpicture}
\hspace*{.5cm}
\begin{tikzpicture}[xscale=-1]
	\begin{pgfonlayer}{nodelayer}
		\node [style=X] (0) at (0, -0) {};
		\node [style=none] (1) at (1, -0) {};
	\end{pgfonlayer}
	\begin{pgfonlayer}{edgelayer}
		\draw [style=simple] (1.center) to (0);
	\end{pgfonlayer}
\end{tikzpicture}
\mapsto
\begin{tikzpicture}[xscale=-1]
	\begin{pgfonlayer}{nodelayer}
		\node [style=X] (0) at (0, -0) {};
		\node [style=none] (1) at (1, -0) {};
	\end{pgfonlayer}
	\begin{pgfonlayer}{edgelayer}
		\draw [style=simple] (1.center) to (0);
	\end{pgfonlayer}
\end{tikzpicture}
$
\end{tabular}
\end{center}

This interepretation is a strict symmetric \dag-monoidal functor.
\end{proposition}
See \S \ref{proof:ZXATOF} for the proof.

\begin{thm}
\label{theorem:TOFZXAiso}
The interpretation functors $\llbracket\_\rrbracket_{\ZXA}$ and $\llbracket\_\rrbracket_{\hat \TOF}$ are inverses, so that $\hat \TOF$ and $\ZXA$ are isomorphic as strongly compact closed props.
\end{thm}
See \S \ref{proof:TOFZXAiso} for the proof.

Recall the following proposition:

\begin{proposition}\cite[Prop. 2.6]{multirelations}\footnote{In \cite{multirelations}, they do not prove this equivalence is monoidal, but it is an obvious corollary. They also do not consider the finite case.}
The category $\Span^\sim(\FinOrd)$ equipped with the Cartesian product is monoidally equivalent to the category of (finite)  matrices over the natural numbers and the Kronecker product.
\end{proposition}

Thus,

\begin{corollary}
$\ZXA$ is complete for the prop of $2^n\times 2^m$ matrices over the natural numbers.
\end{corollary}

\section{Conclusion}
There are various other directions which could be pursued.  One could also ask if there is a normal form for $\ZXA$ induced by the presentation in terms of distributive laws and monoid maps, using the correspondence between strict factorization systems and distributive laws in spans \cite{rosebrugh2002distributive}. 
\nocite{piedeleu} 
It would also be interesting to investigate the 2-categorical structure of $\ZXA$; presenting the corresponding category of relations as a Frobenius theory \cite{functorial} using the partial order enrichment of $\TOF$.

Another immediate direction would be to add the white $\pi$ phase to $\ZXA$ to obtain an approximately universal graphical calculus for quantum computing using only distributive laws and monoid maps.  In such a fragment, one could construct the {\sf and} gate for the $X$ basis; perhaps expanding the table of distributive laws in Figure \ref{fig:table} to be complete for an approximately universal fragment of quantum computing, furthering the general programme of \cite{linrel,duncan} decomposing circuits using distributive laws.  This  approach is contrasted to considering H-boxes as primitives, as in the phase-free fragment of the $\ZH$-calculus \cite{zhpi}---in $\ZXA$+the white $\pi$ phase, the unnormalized Hadamard gate is derived. Perhaps proving the minimality of the axioms using this presentation might be easier, although we do not prove minimality in this paper.

It would also be interesting to investigate the connection to the $\ZH$-calculus and triangle fragments of the $\ZX$-calculus; in particular, in regard to natural number labelled H-boxes, as in \cite{natspiders}.  
These gates can be represented in string diagrams. The diagram of the triangle can be interpreted as the assertion  $x\wedge \neg y =  \bot$ which is equivalent to the material implication  $ x \Rightarrow y$.
\begin{figure}[H]

$$
\begin{tikzpicture}
	\begin{pgfonlayer}{nodelayer}
		\node [style=none] (0) at (1.25, -0.75) {};
		\node [style=none] (1) at (0.5, -0.75) {};
		\node [style=none] (2) at (2, -0.75) {};
		\node [style=triflip] (3) at (1.25, -0.75) {};
	\end{pgfonlayer}
	\begin{pgfonlayer}{edgelayer}
		\draw (2.center) to (0.center);
		\draw (0.center) to (1.center);
	\end{pgfonlayer}
\end{tikzpicture}
:=
\begin{tikzpicture}
	\begin{pgfonlayer}{nodelayer}
		\node [style=none] (0) at (1.5, -0.75) {};
		\node [style=none] (1) at (1.5, -1.25) {};
		\node [style=Z] (2) at (2.25, -0.75) {};
		\node [style=none] (3) at (0.5, -0.75) {};
		\node [style=andin] (4) at (1.5, -0.75) {};
		\node [style=none] (5) at (3, -1.25) {};
		\node [style=Z] (6) at (2.25, -1.25) {$\pi$};
	\end{pgfonlayer}
	\begin{pgfonlayer}{edgelayer}
		\draw [style=simple] (2.center) to (0.center);
		\draw [style=simple, in=-165, out=180, looseness=2.75] (1.center) to (0.center);
		\draw [style=simple] (3.center) to (0.center);
		\draw (5.center) to (6);
		\draw (6) to (1.center);
	\end{pgfonlayer}
\end{tikzpicture}
\hspace*{1cm}
\begin{tikzpicture}
	\begin{pgfonlayer}{nodelayer}
		\node [style=H] (0) at (1, -1) {$n$};
		\node [style=none] (1) at (0.25, -1) {};
		\node [style=none] (2) at (1.75, -1) {};
	\end{pgfonlayer}
	\begin{pgfonlayer}{edgelayer}
		\draw (2.center) to (0);
		\draw (0) to (1.center);
	\end{pgfonlayer}
\end{tikzpicture}
:=
\begin{tikzpicture}
	\begin{pgfonlayer}{nodelayer}
		\node [style=none] (0) at (-0.5, -0.75) {};
		\node [style=none] (1) at (-0.5, -1.25) {};
		\node [style=Z] (2) at (2.75, -0.75) {$\pi$};
		\node [style=none] (3) at (-1.5, -0.75) {};
		\node [style=andin] (4) at (-0.5, -0.75) {};
		\node [style=none] (5) at (3.25, -1.25) {};
		\node [style=triflip] (6) at (2, -0.75) {};
		\node [style=triflip] (7) at (1, -0.75) {};
		\node [style=none] (8) at (1.5, -0.5) {$n$};
		\node [style=Z] (9) at (0.25, -0.75) {$\pi$};
	\end{pgfonlayer}
	\begin{pgfonlayer}{edgelayer}
		\draw [style=simple, in=-165, out=180, looseness=2.75] (1.center) to (0.center);
		\draw [style=simple] (3.center) to (0.center);
		\draw [style=dotted] (6) to (7);
		\draw (6) to (2);
		\draw (7) to (0.center);
		\draw [style=simple] (1.center) to (5.center);
	\end{pgfonlayer}
\end{tikzpicture}
$$

\caption{Triangles and H-boxes in \texorpdfstring{$\ZXA$}{ZX\&}, for \texorpdfstring{$n\in \N$}{n a natural number}.}
\label{fig:gens}
\end{figure}

{%
\textbf{Acknowledgements}

I would like to thank Sivert Aasn\ae ss, Niel de Beaudrap, Robin Cockett, Lukas Heidemann, Aleks Kissinger, JS Lemay and James Hefford, Konstantinos Meichanetzidis and John van de Wetering  for useful discussions.  Especially Robin for pointing out that the conditions for the  $\STOCH$ construction and the Cartesian completion were equivalent, {\em in general}; Aleks for seeing the decomposition of the triangle; Niel for asking about the relation to natural number labeled H boxes; and John for private discussions involving the simplification of the axioms of the phase free $\ZH$-calculus.

I also gratefully acknowledge support from the Clarendon fund.
}

\nocite{coecke2008classical}
\nocite{cnot}
\nocite{tof}
\nocite{Cole}
\nocite{elltwo}
\nocite{coecke2017two}
\nocite{carboni}
\nocite{butz}
\nocite{pqp}
\nocite{lack2004composing}

\bibliographystyle{eptcs}
\bibliography{zxaarxiv}

\newpage

\appendix 

\section{Proof of Lemma \ref{lemma:xtildefaithful}}
\label{proof:xtildefaithful}

Recall the statement of the Lemma:

\textbf{Lemma  \ref{lemma:xtildefaithful}:}
\textit{
The canonical functor $\iota:\X\to \tilde \X$ is faithful.
}

\begin{proof}
Suppose that $\iota(f)\sim\iota(g)$, Then:

\begin{align*}
\begin{tikzpicture}
	\begin{pgfonlayer}{nodelayer}
		\node [style=map] (0) at (0, 0.5) {$g$};
		\node [style=none] (1) at (1, 0.5) {};
		\node [style=none] (2) at (-1, 0.5) {};
	\end{pgfonlayer}
	\begin{pgfonlayer}{edgelayer}
		\draw (1.center) to (0);
		\draw (0) to (2.center);
	\end{pgfonlayer}
\end{tikzpicture}
&=
\begin{tikzpicture}
	\begin{pgfonlayer}{nodelayer}
		\node [style=map] (0) at (0, 0.5) {$f$};
		\node [style=none] (1) at (-0.75, 0.5) {};
		\node [style=map] (2) at (1.75, 0.25) {$f^\circ$};
		\node [style=map] (3) at (2.5, 0.25) {$g$};
		\node [style=X] (4) at (0.75, 0.5) {};
		\node [style=X] (5) at (3.5, 0.5) {};
		\node [style=none] (6) at (4.5, 0.5) {};
	\end{pgfonlayer}
	\begin{pgfonlayer}{edgelayer}
		\draw (0) to (1.center);
		\draw (6.center) to (5);
		\draw [in=30, out=150, looseness=0.75] (5) to (4);
		\draw [in=180, out=-34, looseness=1.00] (4) to (2);
		\draw (0) to (4);
		\draw (3) to (2);
		\draw [in=-153, out=0, looseness=1.00] (3) to (5);
	\end{pgfonlayer}
\end{tikzpicture}
=
\begin{tikzpicture}
	\begin{pgfonlayer}{nodelayer}
		\node [style=none] (0) at (-1.5, 0.5) {};
		\node [style=map] (1) at (2.5, 0.25) {$g$};
		\node [style=X] (2) at (0.75, 0.5) {};
		\node [style=X] (3) at (3.5, 0.5) {};
		\node [style=none] (4) at (4.5, 0.5) {};
		\node [style=map] (5) at (1.75, 0.25) {$f^\circ$};
		\node [style=map] (6) at (-0.75, 0.5) {$f$};
		\node [style=map] (7) at (0, 0.5) {$f^\circ f$};
	\end{pgfonlayer}
	\begin{pgfonlayer}{edgelayer}
		\draw (4.center) to (3);
		\draw [in=30, out=150, looseness=0.75] (3) to (2);
		\draw [in=-153, out=0, looseness=1.00] (1) to (3);
		\draw [in=180, out=-34, looseness=1.00] (2) to (5);
		\draw (1) to (5);
		\draw [style=simple] (2) to (7);
		\draw [style=simple] (7) to (6);
		\draw [style=simple] (6) to (0.center);
	\end{pgfonlayer}
\end{tikzpicture}\\
&=
\begin{tikzpicture}
	\begin{pgfonlayer}{nodelayer}
		\node [style=none] (0) at (-1, 0.5) {};
		\node [style=map] (1) at (2.5, -0) {$g$};
		\node [style=X] (2) at (0.75, 0.5) {};
		\node [style=X] (3) at (3.5, 0.5) {};
		\node [style=none] (4) at (4.5, 0.5) {};
		\node [style=map] (5) at (1.75, -0) {$f^\circ$};
		\node [style=map] (6) at (-0.25, 0.5) {$f$};
		\node [style=map] (7) at (1.75, 1) {$f^\circ f$};
		\node [style=none] (8) at (2.5, 1) {};
	\end{pgfonlayer}
	\begin{pgfonlayer}{edgelayer}
		\draw (4.center) to (3);
		\draw [in=-150, out=0, looseness=1.00] (1) to (3);
		\draw [in=180, out=-34, looseness=1.00] (2) to (5);
		\draw (1) to (5);
		\draw [style=simple] (6) to (0.center);
		\draw [style=simple, in=180, out=30, looseness=1.00] (2) to (7);
		\draw [style=simple] (2) to (6);
		\draw [style=simple, in=0, out=150, looseness=1.00] (3) to (8.center);
		\draw [style=simple] (8.center) to (7);
	\end{pgfonlayer}
\end{tikzpicture}
=
\begin{tikzpicture}
	\begin{pgfonlayer}{nodelayer}
		\node [style=none] (0) at (-1, 0.5) {};
		\node [style=map] (1) at (1.75, -0) {$g$};
		\node [style=X] (2) at (0.75, 0.5) {};
		\node [style=X] (3) at (2.75, 0.5) {};
		\node [style=none] (4) at (3.75, 0.5) {};
		\node [style=map] (5) at (-0.25, 0.5) {$ff^\circ$};
		\node [style=map] (6) at (1.75, 1) {$f$};
	\end{pgfonlayer}
	\begin{pgfonlayer}{edgelayer}
		\draw (4.center) to (3);
		\draw [in=-150, out=0, looseness=1.00] (1) to (3);
		\draw [style=simple] (5) to (0.center);
		\draw [style=simple, in=180, out=30, looseness=1.00] (2) to (6);
		\draw [style=simple] (2) to (5);
		\draw [style=simple, in=-30, out=180, looseness=1.00] (1) to (2);
		\draw [style=simple, in=0, out=150, looseness=1.00] (3) to (6);
	\end{pgfonlayer}
\end{tikzpicture}\\
&=
\begin{tikzpicture}
	\begin{pgfonlayer}{nodelayer}
		\node [style=none] (0) at (-1, 0.5) {};
		\node [style=map] (1) at (1.75, -0) {$g$};
		\node [style=X] (2) at (0, 0.5) {};
		\node [style=X] (3) at (2.75, 0.5) {};
		\node [style=none] (4) at (3.75, 0.5) {};
		\node [style=map] (5) at (1.75, 1) {$f$};
		\node [style=map] (6) at (1, 1) {$ff^\circ$};
		\node [style=none] (7) at (1, -0) {};
	\end{pgfonlayer}
	\begin{pgfonlayer}{edgelayer}
		\draw (4.center) to (3);
		\draw [in=-150, out=0, looseness=1.00] (1) to (3);
		\draw [style=simple, in=0, out=150, looseness=1.00] (3) to (5);
		\draw (1) to (7.center);
		\draw [in=-30, out=180, looseness=1.00] (7.center) to (2);
		\draw (2) to (0.center);
		\draw [in=180, out=30, looseness=1.00] (2) to (6);
		\draw (6) to (5);
	\end{pgfonlayer}
\end{tikzpicture}
=
\begin{tikzpicture}
	\begin{pgfonlayer}{nodelayer}
		\node [style=none] (0) at (-1, 0.5) {};
		\node [style=map] (1) at (1, -0) {$g$};
		\node [style=X] (2) at (0, 0.5) {};
		\node [style=X] (3) at (2, 0.5) {};
		\node [style=none] (4) at (3, 0.5) {};
		\node [style=map] (5) at (1, 1) {$f$};
	\end{pgfonlayer}
	\begin{pgfonlayer}{edgelayer}
		\draw (4.center) to (3);
		\draw [in=-150, out=0, looseness=1.00] (1) to (3);
		\draw [style=simple, in=0, out=150, looseness=1.00] (3) to (5);
		\draw (2) to (0.center);
		\draw [in=-30, out=180, looseness=1.00] (1) to (2);
		\draw [in=180, out=30, looseness=1.00] (2) to (5);
	\end{pgfonlayer}
\end{tikzpicture}\\
&=
\begin{tikzpicture}
	\begin{pgfonlayer}{nodelayer}
		\node [style=none] (0) at (-1, 0.5) {};
		\node [style=map] (1) at (1, .2) {$g$};
		\node [style=X] (2) at (0, 0.5) {};
		\node [style=X] (3) at (2, 0.5) {};
		\node [style=none] (4) at (3, 0.5) {};
		\node [style=map] (5) at (1, .8) {$f$};
	\end{pgfonlayer}
	\begin{pgfonlayer}{edgelayer}
		\draw (4.center) to (3);
		\draw [in=150, out=0, looseness=1.25] (1) to (3);
		\draw [style=simple, in=0, out=-150, looseness=1.25] (3) to (5);
		\draw (2) to (0.center);
		\draw [in=30, out=180, looseness=1.25] (1) to (2);
		\draw [in=180, out=-30, looseness=1.25] (2) to (5);
	\end{pgfonlayer}
\end{tikzpicture}
=
\begin{tikzpicture}
	\begin{pgfonlayer}{nodelayer}
		\node [style=none] (0) at (-1, 0.5) {};
		\node [style=map] (1) at (1, -0) {$f$};
		\node [style=X] (2) at (0, 0.5) {};
		\node [style=X] (3) at (2, 0.5) {};
		\node [style=none] (4) at (3, 0.5) {};
		\node [style=map] (5) at (1, 1) {$g$};
	\end{pgfonlayer}
	\begin{pgfonlayer}{edgelayer}
		\draw (4.center) to (3);
		\draw [in=-150, out=0, looseness=1.00] (1) to (3);
		\draw [style=simple, in=0, out=150, looseness=1.00] (3) to (5);
		\draw (2) to (0.center);
		\draw [in=-30, out=180, looseness=1.00] (1) to (2);
		\draw [in=180, out=30, looseness=1.00] (2) to (5);
	\end{pgfonlayer}
\end{tikzpicture}\\
&=
\begin{tikzpicture}
	\begin{pgfonlayer}{nodelayer}
		\node [style=map] (0) at (0, 0.5) {$g$};
		\node [style=none] (1) at (-0.75, 0.5) {};
		\node [style=map] (2) at (1.75, 0.25) {$g^\circ$};
		\node [style=map] (3) at (2.5, 0.25) {$f$};
		\node [style=X] (4) at (0.75, 0.5) {};
		\node [style=X] (5) at (3.5, 0.5) {};
		\node [style=none] (6) at (4.5, 0.5) {};
	\end{pgfonlayer}
	\begin{pgfonlayer}{edgelayer}
		\draw (0) to (1.center);
		\draw (6.center) to (5);
		\draw [in=30, out=150, looseness=0.75] (5) to (4);
		\draw [in=180, out=-34, looseness=1.00] (4) to (2);
		\draw (0) to (4);
		\draw (3) to (2);
		\draw [in=-153, out=0, looseness=1.00] (3) to (5);
	\end{pgfonlayer}
\end{tikzpicture}
=
\begin{tikzpicture}
	\begin{pgfonlayer}{nodelayer}
		\node [style=map] (0) at (0, 0.5) {$f$};
		\node [style=none] (1) at (1, 0.5) {};
		\node [style=none] (2) at (-1, 0.5) {};
	\end{pgfonlayer}
	\begin{pgfonlayer}{edgelayer}
		\draw (1.center) to (0);
		\draw (0) to (2.center);
	\end{pgfonlayer}
\end{tikzpicture}
\end{align*}

\end{proof}


\section{Proof of lemma \ref{lemma:envstruct}}
\label{proof:envstruct}

Recall the statement of the Lemma:

\textbf{Lemma \ref{lemma:envstruct}}
\textit{Given a discrete inverse category $\X$, $c(\X)$ and $\tilde \X$ are isomorphic as discrete Cartesian restriction categories.}

\begin{proof}
Define an identity on objects functor $F:c(\X)\to \tilde \X$ in the obvious way, sending the counits to the ancillary space.
Similarly, define an identity on objects functor from $G:\tilde \X \to c(\X)$ given by plugging counits into the ancillary space.
These maps are clearly inverses to each other and preserve discrete Cartesian restriction structure; however, once again we mush show that they are actually  functors.

To see  that $F$ is a functor, it suffices to observe that every object in  $\tilde \X $ is equipped with a counital Frobenius algebra, compatible with the monoidal structure, where the unit is in the image of the freely adjoined counit under $F$.


To prove that $G$ is a functor, take some $(f,S)\sim (g,T)$ in $\tilde \X$.
Therefore, in $\tilde \X$, since the Frobenius structure is counital:
$$
\begin{tikzpicture}
	\begin{pgfonlayer}{nodelayer}
		\node [style=map] (0) at (3, -0) {$f^\circ$};
		\node [style=X] (1) at (2, 0.5) {};
		\node [style=map] (2) at (1, -0) {$f$};
		\node [style=none] (3) at (4, -0) {};
		\node [style=none] (4) at (0.25, 1) {};
		\node [style=none] (5) at (0.25, -0) {};
		\node [style=none] (6) at (3.5, -0) {};
		\node [style=none] (7) at (4, -0.5) {};
	\end{pgfonlayer}
	\begin{pgfonlayer}{edgelayer}
		\draw [style=simple] (0) to (3.center);
		\draw [style=simple] (1) to (2);
		\draw [style=simple, in=0, out=166, looseness=1.00] (1) to (4.center);
		\draw [style=simple] (5.center) to (2);
		\draw [style=simple, in=-30, out=-150, looseness=0.75] (0) to (2);
		\draw [in=153, out=0, looseness=1.00] (1) to (0);
		\draw [style=dashed, in=-60, out=180, looseness=1.00] (7.center) to (6.center);
	\end{pgfonlayer}
\end{tikzpicture}
\sim
\begin{tikzpicture}
	\begin{pgfonlayer}{nodelayer}
		\node [style=map] (0) at (2.75, -0) {$f^\circ$};
		\node [style=X] (1) at (1.75, 0.5) {};
		\node [style=X] (2) at (1, 0.5) {};
		\node [style=map] (3) at (0, -0) {$f$};
		\node [style=none] (4) at (3.75, -0) {};
		\node [style=none] (5) at (-1, 1) {};
		\node [style=none] (6) at (-1, -0) {};
		\node [style=none] (7) at (3.75, -0.5) {};
		\node [style=none] (8) at (3, 0.75) {};
	\end{pgfonlayer}
	\begin{pgfonlayer}{edgelayer}
		\draw [style=simple] (1) to (0);
		\draw [style=simple] (0) to (4.center);
		\draw [style=simple] (1) to (2);
		\draw [style=simple] (2) to (3);
		\draw [style=simple, in=0, out=166, looseness=1.00] (2) to (5.center);
		\draw [style=simple] (6.center) to (3);
		\draw [style=simple, in=-30, out=-150, looseness=0.75] (0) to (3);
		\draw [style=simple, in=0, out=180, looseness=0.75] (7.center) to (8.center);
		\draw [style=simple, in=180, out=11, looseness=1.00] (1) to (8.center);
	\end{pgfonlayer}
\end{tikzpicture}
=
\begin{tikzpicture}
	\begin{pgfonlayer}{nodelayer}
		\node [style=map] (0) at (2.75, -0) {$g^\circ$};
		\node [style=X] (1) at (1.75, 0.5) {};
		\node [style=X] (2) at (1, 0.5) {};
		\node [style=map] (3) at (0, -0) {$g$};
		\node [style=none] (4) at (3.75, -0) {};
		\node [style=none] (5) at (-1, 1) {};
		\node [style=none] (6) at (-1, -0) {};
		\node [style=none] (7) at (3.75, -0.5) {};
		\node [style=none] (8) at (3, 0.75) {};
	\end{pgfonlayer}
	\begin{pgfonlayer}{edgelayer}
		\draw [style=simple] (1) to (0);
		\draw [style=simple] (0) to (4.center);
		\draw [style=simple] (1) to (2);
		\draw [style=simple] (2) to (3);
		\draw [style=simple, in=0, out=166, looseness=1.00] (2) to (5.center);
		\draw [style=simple] (6.center) to (3);
		\draw [style=simple, in=-30, out=-150, looseness=0.75] (0) to (3);
		\draw [style=simple, in=0, out=180, looseness=0.75] (7.center) to (8.center);
		\draw [style=simple, in=180, out=11, looseness=1.00] (1) to (8.center);
	\end{pgfonlayer}
\end{tikzpicture}
\sim
\begin{tikzpicture}
	\begin{pgfonlayer}{nodelayer}
		\node [style=map] (0) at (3, -0) {$g^\circ$};
		\node [style=X] (1) at (2, 0.5) {};
		\node [style=map] (2) at (1, -0) {$g$};
		\node [style=none] (3) at (4, -0) {};
		\node [style=none] (4) at (0.25, 1) {};
		\node [style=none] (5) at (0.25, -0) {};
		\node [style=none] (6) at (3.5, -0) {};
		\node [style=none] (7) at (4, -0.5) {};
	\end{pgfonlayer}
	\begin{pgfonlayer}{edgelayer}
		\draw [style=simple] (0) to (3.center);
		\draw [style=simple] (1) to (2);
		\draw [style=simple, in=0, out=166, looseness=1.00] (1) to (4.center);
		\draw [style=simple] (5.center) to (2);
		\draw [style=simple, in=-30, out=-150, looseness=0.75] (0) to (2);
		\draw [in=153, out=0, looseness=1.00] (1) to (0);
		\draw [style=dashed, in=-60, out=180, looseness=1.00] (7.center) to (6.center);
	\end{pgfonlayer}
\end{tikzpicture}
$$

However, since the functor $\X\to \tilde \X $ is faithful by Lemma \ref{lemma:xtildefaithful}, using the alternate equivalence relation of $\tilde \X$ by  Lemma \ref{theorem:cpstartheorem}, we have that in $\X$:

$$
\begin{tikzpicture}
	\begin{pgfonlayer}{nodelayer}
		\node [style=map] (0) at (3, -0) {$f^\circ$};
		\node [style=X] (1) at (2, 0.5) {};
		\node [style=map] (2) at (1, -0) {$f$};
		\node [style=none] (3) at (4, -0) {};
		\node [style=none] (4) at (0.25, 1) {};
		\node [style=none] (5) at (0.25, -0) {};
	\end{pgfonlayer}
	\begin{pgfonlayer}{edgelayer}
		\draw [style=simple] (0) to (3.center);
		\draw [style=simple] (1) to (2);
		\draw [style=simple, in=0, out=166, looseness=1.00] (1) to (4.center);
		\draw [style=simple] (5.center) to (2);
		\draw [style=simple, in=-30, out=-150, looseness=0.75] (0) to (2);
		\draw [in=153, out=0, looseness=1.00] (1) to (0);
	\end{pgfonlayer}
\end{tikzpicture}
=
\begin{tikzpicture}
	\begin{pgfonlayer}{nodelayer}
		\node [style=map] (0) at (3, -0) {$g^\circ$};
		\node [style=X] (1) at (2, 0.5) {};
		\node [style=map] (2) at (1, -0) {$g$};
		\node [style=none] (3) at (4, -0) {};
		\node [style=none] (4) at (0, 1) {};
		\node [style=none] (5) at (0, -0) {};
	\end{pgfonlayer}
	\begin{pgfonlayer}{edgelayer}
		\draw [style=simple] (0) to (3.center);
		\draw [style=simple] (1) to (2);
		\draw [style=simple, in=0, out=166, looseness=1.00] (1) to (4.center);
		\draw [style=simple] (5.center) to (2);
		\draw [style=simple, in=-30, out=-150, looseness=0.75] (0) to (2);
		\draw [in=153, out=0, looseness=1.00] (1) to (0);
	\end{pgfonlayer}
\end{tikzpicture}
\hspace*{.1cm}\text{and thus}\hspace*{.1cm}
\begin{tikzpicture}
	\begin{pgfonlayer}{nodelayer}
		\node [style=map] (0) at (1, -0) {$f$};
		\node [style=X] (1) at (2, 0.5) {};
		\node [style=map] (2) at (3, -0) {$f^\circ$};
		\node [style=none] (3) at (0, -0) {};
		\node [style=none] (4) at (4, 1) {};
		\node [style=none] (5) at (4, -0) {};
	\end{pgfonlayer}
	\begin{pgfonlayer}{edgelayer}
		\draw [style=simple] (0) to (3.center);
		\draw [style=simple] (1) to (2);
		\draw [style=simple, in=180, out=14, looseness=1.00] (1) to (4.center);
		\draw [style=simple] (5.center) to (2);
		\draw [style=simple, in=-150, out=-30, looseness=0.75] (0) to (2);
		\draw [in=27, out=180, looseness=1.00] (1) to (0);
	\end{pgfonlayer}
\end{tikzpicture}
=
\begin{tikzpicture}
	\begin{pgfonlayer}{nodelayer}
		\node [style=map] (0) at (1, -0) {$g$};
		\node [style=X] (1) at (2, 0.5) {};
		\node [style=map] (2) at (3, -0) {$g^\circ$};
		\node [style=none] (3) at (0, -0) {};
		\node [style=none] (4) at (3.75, 1) {};
		\node [style=none] (5) at (3.75, -0) {};
	\end{pgfonlayer}
	\begin{pgfonlayer}{edgelayer}
		\draw [style=simple] (0) to (3.center);
		\draw [style=simple] (1) to (2);
		\draw [style=simple, in=180, out=14, looseness=1.00] (1) to (4.center);
		\draw [style=simple] (5.center) to (2);
		\draw [style=simple, in=-150, out=-30, looseness=0.75] (0) to (2);
		\draw [in=27, out=180, looseness=1.00] (1) to (0);
	\end{pgfonlayer}
\end{tikzpicture}
$$

Therefore in $c(\X)$:

\begin{align*}
&
\begin{tikzpicture}
	\begin{pgfonlayer}{nodelayer}
		\node [style=map] (0) at (1, -0) {$f$};
		\node [style=X] (1) at (2, 0.5) {};
		\node [style=map] (2) at (3, -0) {$f^\circ$};
		\node [style=none] (3) at (0, -0) {};
		\node [style=none] (4) at (4, 1) {};
		\node [style=none] (5) at (3.75, -0) {};
		\node [style=X] (6) at (3.75, -0) {};
	\end{pgfonlayer}
	\begin{pgfonlayer}{edgelayer}
		\draw [style=simple] (0) to (3.center);
		\draw [style=simple] (1) to (2);
		\draw [style=simple, in=180, out=14, looseness=1.00] (1) to (4.center);
		\draw [style=simple] (5.center) to (2);
		\draw [style=simple, in=-150, out=-30, looseness=0.75] (0) to (2);
		\draw [in=27, out=180, looseness=1.00] (1) to (0);
	\end{pgfonlayer}
\end{tikzpicture}
\eq{Rem. \ref{cor:copy}}
\begin{tikzpicture}
	\begin{pgfonlayer}{nodelayer}
		\node [style=map] (0) at (1, -0) {$f$};
		\node [style=X] (1) at (2, 0.5) {};
		\node [style=map] (2) at (3, -0) {$\bar {f^\circ}$};
		\node [style=none] (3) at (0, -0) {};
		\node [style=none] (4) at (4, 1) {};
		\node [style=X] (5) at (4, 0.25) {};
		\node [style=X] (6) at (4, -0.25) {};
	\end{pgfonlayer}
	\begin{pgfonlayer}{edgelayer}
		\draw [style=simple] (0) to (3.center);
		\draw [style=simple] (1) to (2);
		\draw [style=simple, in=180, out=14, looseness=1.00] (1) to (4.center);
		\draw [style=simple, in=-150, out=-30, looseness=0.75] (0) to (2);
		\draw [in=27, out=180, looseness=1.00] (1) to (0);
		\draw [in=14, out=180, looseness=1.00] (5) to (2);
		\draw [in=180, out=-14, looseness=1.00] (2) to (6);
	\end{pgfonlayer}
\end{tikzpicture}
=
\begin{tikzpicture}
	\begin{pgfonlayer}{nodelayer}
		\node [style=map] (0) at (1, -0) {$f$};
		\node [style=X] (1) at (2, 0.5) {};
		\node [style=none] (2) at (0, -0) {};
		\node [style=none] (3) at (4, 1) {};
		\node [style=X] (4) at (4, 0.25) {};
		\node [style=X] (5) at (4, -0.25) {};
		\node [style=map] (6) at (3.25, -0.5) {$\bar {f^\circ}$};
		\node [style=X] (7) at (4.75, 0.25) {};
		\node [style=X] (8) at (4.75, -0.25) {};
		\node [style=X] (9) at (2.5, 0.25) {};
		\node [style=X] (10) at (2.5, -0.25) {};
	\end{pgfonlayer}
	\begin{pgfonlayer}{edgelayer}
		\draw [style=simple] (0) to (2.center);
		\draw [style=simple, in=180, out=14, looseness=1.00] (1) to (3.center);
		\draw [in=27, out=180, looseness=1.00] (1) to (0);
		\draw [in=150, out=30, looseness=1.25] (9) to (4);
		\draw [in=150, out=30, looseness=1.25] (10) to (5);
		\draw [in=-165, out=-18, looseness=1.00] (10) to (6);
		\draw [in=150, out=-45, looseness=1.00] (9) to (6);
		\draw [in=-135, out=30, looseness=1.00] (6) to (4);
		\draw [in=-15, out=-162, looseness=1.00] (5) to (6);
		\draw (5) to (8);
		\draw (7) to (4);
		\draw (9) to (1);
		\draw [in=-30, out=180, looseness=1.00] (10) to (0);
	\end{pgfonlayer}
\end{tikzpicture}\\
&=
\begin{tikzpicture}
	\begin{pgfonlayer}{nodelayer}
		\node [style=map] (0) at (1.5, -0) {$f$};
		\node [style=none] (1) at (1, -0) {};
		\node [style=none] (2) at (5.75, 1) {};
		\node [style=X] (3) at (4, 0.25) {};
		\node [style=X] (4) at (4, -0.25) {};
		\node [style=map] (5) at (3.25, -0.5) {$\bar {f^\circ}$};
		\node [style=X] (6) at (5.5, 0.25) {};
		\node [style=X] (7) at (5.5, -0.25) {};
		\node [style=X] (8) at (2.5, 0.25) {};
		\node [style=X] (9) at (2.5, -0.25) {};
		\node [style=X] (10) at (4.75, 0.25) {};
	\end{pgfonlayer}
	\begin{pgfonlayer}{edgelayer}
		\draw [style=simple] (0) to (1.center);
		\draw [in=150, out=30, looseness=1.25] (8) to (3);
		\draw [in=150, out=30, looseness=1.25] (9) to (4);
		\draw [in=-165, out=-18, looseness=1.00] (9) to (5);
		\draw [in=150, out=-45, looseness=1.00] (8) to (5);
		\draw [in=-135, out=30, looseness=1.00] (5) to (3);
		\draw [in=-15, out=-162, looseness=1.00] (4) to (5);
		\draw (4) to (7);
		\draw (6) to (3);
		\draw [in=-30, out=180, looseness=1.00] (9) to (0);
		\draw [in=37, out=180, looseness=1.00] (2.center) to (10);
		\draw [in=30, out=180, looseness=1.00] (8) to (0);
	\end{pgfonlayer}
\end{tikzpicture}
=
\begin{tikzpicture}
	\begin{pgfonlayer}{nodelayer}
		\node [style=map] (0) at (1.5, -0) {$f$};
		\node [style=none] (1) at (1, -0) {};
		\node [style=none] (2) at (4, 1) {};
		\node [style=map] (3) at (2.25, -0) {$\bar {f^\circ}$};
		\node [style=X] (4) at (3.75, 0.25) {};
		\node [style=X] (5) at (3.75, -0.25) {};
		\node [style=X] (6) at (3, 0.25) {};
	\end{pgfonlayer}
	\begin{pgfonlayer}{edgelayer}
		\draw [style=simple] (0) to (1.center);
		\draw [in=37, out=180, looseness=1.00] (2.center) to (6);
		\draw [in=30, out=150, looseness=1.25] (3) to (0);
		\draw [in=-165, out=-15, looseness=1.25] (0) to (3);
		\draw [in=180, out=30, looseness=1.00] (3) to (6);
		\draw [in=180, out=-15, looseness=1.00] (3) to (5);
		\draw (4) to (6);
	\end{pgfonlayer}
\end{tikzpicture}
=
\begin{tikzpicture}
	\begin{pgfonlayer}{nodelayer}
		\node [style=none] (0) at (1.5, -0) {};
		\node [style=none] (1) at (4, 1) {};
		\node [style=X] (2) at (3.75, -0.25) {};
		\node [style=X] (3) at (3, 0.25) {};
		\node [style=map] (4) at (2, -0) {$f$};
		\node [style=X] (5) at (3.75, 0.25) {};
	\end{pgfonlayer}
	\begin{pgfonlayer}{edgelayer}
		\draw [in=37, out=180, looseness=1.00] (1.center) to (3);
		\draw [style=simple] (4) to (0.center);
		\draw (5) to (3);
		\draw [in=-15, out=180, looseness=1.00] (2) to (4);
		\draw [in=180, out=30, looseness=1.00] (4) to (3);
	\end{pgfonlayer}
\end{tikzpicture}
=
\begin{tikzpicture}
	\begin{pgfonlayer}{nodelayer}
		\node [style=none] (0) at (1.5, -0) {};
		\node [style=X] (1) at (3, -0.25) {};
		\node [style=map] (2) at (2, -0) {$f$};
		\node [style=none] (3) at (3, 0.25) {};
		\node [style=none] (4) at (3.5, 0.25) {};
	\end{pgfonlayer}
	\begin{pgfonlayer}{edgelayer}
		\draw [style=simple] (2) to (0.center);
		\draw [in=-15, out=180, looseness=1.00] (1) to (2);
		\draw (4.center) to (3.center);
		\draw [in=14, out=180, looseness=1.00] (3.center) to (2);
	\end{pgfonlayer}
\end{tikzpicture}
\end{align*}

So that combining the previous two equations:

\begin{align*}
\begin{tikzpicture}
	\begin{pgfonlayer}{nodelayer}
		\node [style=none] (0) at (1.5, -0) {};
		\node [style=X] (1) at (3, -0.25) {};
		\node [style=map] (2) at (2, -0) {$f$};
		\node [style=none] (3) at (3, 0.25) {};
		\node [style=none] (4) at (3.5, 0.25) {};
	\end{pgfonlayer}
	\begin{pgfonlayer}{edgelayer}
		\draw [style=simple] (2) to (0.center);
		\draw [in=-15, out=180, looseness=1.00] (1) to (2);
		\draw (4.center) to (3.center);
		\draw [in=14, out=180, looseness=1.00] (3.center) to (2);
	\end{pgfonlayer}
\end{tikzpicture}
=
\begin{tikzpicture}
	\begin{pgfonlayer}{nodelayer}
		\node [style=map] (0) at (1, -0) {$f$};
		\node [style=X] (1) at (2, 0.5) {};
		\node [style=map] (2) at (3, -0) {$f^\circ$};
		\node [style=none] (3) at (0, -0) {};
		\node [style=none] (4) at (4, 1) {};
		\node [style=none] (5) at (3.75, -0) {};
		\node [style=X] (6) at (3.75, -0) {};
	\end{pgfonlayer}
	\begin{pgfonlayer}{edgelayer}
		\draw [style=simple] (0) to (3.center);
		\draw [style=simple] (1) to (2);
		\draw [style=simple, in=180, out=14, looseness=1.00] (1) to (4.center);
		\draw [style=simple] (5.center) to (2);
		\draw [style=simple, in=-150, out=-30, looseness=0.75] (0) to (2);
		\draw [in=27, out=180, looseness=1.00] (1) to (0);
	\end{pgfonlayer}
\end{tikzpicture}
=
\begin{tikzpicture}
	\begin{pgfonlayer}{nodelayer}
		\node [style=map] (0) at (1, -0) {$g$};
		\node [style=X] (1) at (2, 0.5) {};
		\node [style=map] (2) at (3, -0) {$g^\circ$};
		\node [style=none] (3) at (0, -0) {};
		\node [style=none] (4) at (4, 1) {};
		\node [style=none] (5) at (3.75, -0) {};
		\node [style=X] (6) at (3.75, -0) {};
	\end{pgfonlayer}
	\begin{pgfonlayer}{edgelayer}
		\draw [style=simple] (0) to (3.center);
		\draw [style=simple] (1) to (2);
		\draw [style=simple, in=180, out=14, looseness=1.00] (1) to (4.center);
		\draw [style=simple] (5.center) to (2);
		\draw [style=simple, in=-150, out=-30, looseness=0.75] (0) to (2);
		\draw [in=27, out=180, looseness=1.00] (1) to (0);
	\end{pgfonlayer}
\end{tikzpicture}
=
\begin{tikzpicture}
	\begin{pgfonlayer}{nodelayer}
		\node [style=none] (0) at (1.5, -0) {};
		\node [style=X] (1) at (3, -0.25) {};
		\node [style=map] (2) at (2, -0) {$g$};
		\node [style=none] (3) at (3, 0.25) {};
		\node [style=none] (4) at (3.5, 0.25) {};
	\end{pgfonlayer}
	\begin{pgfonlayer}{edgelayer}
		\draw [style=simple] (2) to (0.center);
		\draw [in=-15, out=180, looseness=1.00] (1) to (2);
		\draw (4.center) to (3.center);
		\draw [in=14, out=180, looseness=1.00] (3.center) to (2);
	\end{pgfonlayer}
\end{tikzpicture}
\end{align*}

\end{proof}
%
%
%
%
%

\section{Proof of Lemma \ref{lemma:unitcounit}}
\label{proof:unitcounit}

Recall the statement of the Lemma:

\textbf{Lemma \ref{lemma:unitcounit}:}
\textit{
 The full subcategory of $\Span^\sim(\FinOrd)$ generated by powers of 2 is presented by the pushout,  $\hat \TOF$, of the following diagram of props:
}

$$c(\TOF)^\op \leftarrow \TOF \rightarrow c(\TOF)$$

\newcommand{\FPar}{\mathsf{FPar}}
\newcommand{\FSpan}{\mathsf{FSpan}}

\begin{proof}
Recall that $\TOF$ is presented by the subcategory $\FPinj_2$ of $(\Span^\sim (\FinOrd),\times)$ with morphisms of the form $ 2^n \xleftarrowtail{e} k \xrightarrowtail{e'} 2^m$ for arbitrary natural numbers $n,m,k$ and monics $e$ and $e'$.

Similarly, $\tilde \TOF$ is presented by the subcategory $\FPar_2$ of  $(\Span^\sim (\FinOrd),\times)$ with morphisms of the form $2^\ell \xleftarrow{f} 2^n \xleftarrowtail{e} k \xrightarrowtail{e'} 2^m$ for arbitrary natural numbers $\ell, n,m,k$ and monics $e$ and $e'$ and function $f$.
Let $\FSpan_2$ denote the full subcategory of $(\Span^\sim(\FinOrd),\times)$ generated by powers of two.
Consider the pushout $\X$ of the following diagram of props:

$$\FPar_2^\op \xleftarrowtail{}  \FPinj_2 \xrightarrowtail{} \FPar_2$$

Consider the functor $F:\X\to\FSpan_2$ induced by the universal property of the pushout.  We show that this functor is an isomorphism.
This functor is clearly the identity on objects.

For fullness consider some span $2^n \xleftarrow{f} k \xrightarrow{g} 2^m$. We can construct a function $f':2^{\lceil \log_2 k \rceil} \rightarrow 2^n$ and monic $e_f: k \xrightarrowtail{} 2^{\lceil \log_2 k \rceil}$ so that $f=ef'$.  Similarly, we can construct some  $g':2^{\lceil \log_2 k \rceil} \rightarrow 2^n$ and monic $e_g: k \xrightarrowtail{} 2^{\lceil \log_2 k \rceil}$ so that $g=e_gg'$.  Therefore:

\begin{align*}
F&\left(
\xymatrix{
         & 2^{\lceil \log_2 k \rceil} \ar[dl]_{f'} \ar@{=}[dr]\\
2^n &                                                                                 &2^{\lceil \log_2 k \rceil}
};
\xymatrix{
         & k \ar@{>->}[dl]_{e_f} \ar@{>->}[dr]^{e_m}\\
2^{\lceil \log_2 k \rceil} &                                                                                 & 2^{\lceil \log_2 k \rceil}
};
\xymatrix{
                                       & 2^{\lceil \log_2 k \rceil} \ar[dr]^{g'} \ar@{=}[dl]\\
2^{\lceil \log_2 k \rceil} &                                                                                 & 2^m
}
\right)\\
&=
\xymatrix{
         &                                                                               &                                             &  k \ar@{=}[dl] \ar@{=}[dr] \ar@/_2.0pc/[dddlll]_{f} \ar@/^2.0pc/[dddrrr]^{g}  \\
         &                                                                               &k \ar@{=}[dr] \ar@{>->}[dl]_{e_f} &                                                & k \ar@{=}[dl] \ar@{>->}[dr]^{e_g}\\
         & 2^{\lceil \log_2 k \rceil}\ar[dl]^{f'}\ar@{=}[dr]   &                                             & k \ar@{>->}[dl]_{e_f} \ar@{>->}[dr]^{e_g} &                                       & 2^{\lceil \log_2 k \rceil}\ar[dr]_{g'}\ar@{=}[dl] \\
2^n  &                                                                                & 2^{\lceil \log_2 k \rceil}     &                                                & 2^{\lceil \log_2 k \rceil} &                   & 2^m
}
\end{align*}

So $F$ is full.

For faithfulness suppose we have any two isomorphic spans in $F(\X)$:

$$
\xymatrix{
                 &                                       & k \ar@{>->}[dl]_{e_1}\ar@{->}[dddd]_\cong^{\alpha} \ar@{>->}[dr]^{e_2} \\ 
                 & 2^{n_2} \ar[dl]_{f_1}   &                                                                              & 2^{n_3} \ar[dr]^{f_2}\\ 
2^{n_1}   &                                       &                                                                              &                 & 2^{n_4}\\
                 & 2^{n_2'} \ar[ul]^{f_1'} &                                                                              & 2^{n_3'} \ar[ur]_{f_2'}\\ 
                 &                                       & k \ar@{>->}[ul]^{e_1'} \ar@{>->}[ur]_{e_2'} \\ 
}
$$

In $\X$, we have:

\begin{align*}
\xymatrix{
                & 2^{n_2} \ar[dl]_{f_1} \ar@{=}[dr] \\
2^{n_1} &                                                             & 2^{n_2}
};&
\xymatrix{
               & k \ar@{>->}[dl]_{e_1} \ar@{>->}[dr]^{e_2}\\
2^{n_2} &                                               & 2^{n_3}
};
\xymatrix{
                & 2^{n_3} \ar@{=}[dl] \ar[dr]^{f_2} \\
2^{n_3} &                                                             & 2^{n_4}
}\\
&=
\xymatrix{
                 &                                                           & k \ar@{>->}[dl]_{e_1} \ar@{=}[dr]  \ar@/_2.0pc/[ddll]_{\alpha e_1' f_1'}\\
                & 2^{n_2} \ar[dl]_{f_1} \ar@{=}[dr]   &                         & k \ar@{>->}[dl]_{e_1 } \ar@{>->}[dr]^{e_2}  \\
2^{n_1} &                                                             & 2^{n_2}          &                                                 & 2^{n_3}
};
\xymatrix{
                & 2^{n_3} \ar@{=}[dl] \ar[dr]^{f_2} \\
2^{n_3} &                                                             & 2^{n_4}
}\\
&=
\xymatrix{
                 &                                                           & k \ar@{>->}[dl]_{\alpha e_1'} \ar@{>->}[ddrr]^{e_2} \\
                & 2^{n_2'} \ar[dl]_{f_1'}                        &                         &  \\
2^{n_1} &                                                             &           &                                                 & 2^{n_3}
};
\xymatrix{
                & 2^{n_3} \ar@{=}[dl] \ar[dr]^{f_2} \\
2^{n_3} &                                                             & 2^{n_4}
}\\
&=
\xymatrix{
                & 2^{n_2'} \ar[dl]_{f_1'} \ar@{=}[dr] \\
2^{n_1} &                                                             & 2^{n_2'}
};
\xymatrix{
               & k \ar@{>->}[dl]_{\alpha e_1'} \ar@{>->}[dr]^{e_2}\\
2^{n_2'} &                                               & 2^{n_3}
};
\xymatrix{
                & 2^{n_3} \ar@{=}[dl] \ar[dr]^{f_2} \\
2^{n_3} &                                                             & 2^{n_4}
}\\
&=
\xymatrix{
                & 2^{n_2'} \ar[dl]_{f_1'} \ar@{=}[dr] \\
2^{n_1} &                                                             & 2^{n_2'}
};
\xymatrix{
                                       & k \ar@{>->}[dl]_{\alpha e_1'} \ar@{>->}[dr]^{\alpha e_2'} \ar[dd]_\cong^\alpha\\
2^{n_2'} &                                                                         & 2^{n_3'} \\
                                       & k  \ar@{>->}[ul]^{e_1'} \ar@{>->}[ur]_{e_2'}
};
\xymatrix{
                & 2^{n_3'} \ar@{=}[dl] \ar[dr]^{f_2} \\
2^{n_3'} &                                                             & 2^{n_4}
}\\
&=
\xymatrix{
                & 2^{n_2'} \ar[dl]_{f_1'} \ar@{=}[dr] \\
2^{n_1} &                                                             & 2^{n_2'}
};
\xymatrix{
               & k \ar@{>->}[dl]_{ e_1'} \ar@{>->}[dr]^{ e_2'}\\
2^{n_2'} &                                               & 2^{n_3'}
};
\xymatrix{
                & 2^{n_3'} \ar@{=}[dl] \ar[dr]^{f_2} \\
2^{n_3'} &                                                             & 2^{n_4}
}
\end{align*}

Therefore $\FSpan_2 \cong \X$.

Two show that $\hat \TOF \cong \FSpan_2$, consider the following diagram where each horizontal face is a pushout:

$$
\xymatrixrowsep{6mm}\xymatrixcolsep{4mm}
\xymatrix{
                                       & {(\FPinj_2,\times)} \ar[dl] \ar@/^.5pc/[rr] \ar@{=}[d]  &                                                  & (\FPar_2,\times) \ar[d]^{\cong} \ar[dl] \\
 (\FPar_2,\times)^\op \ar@/_1pc/[rr]  \ar[d]_{\cong}           &                   {(\FPinj_2,\times)}\ar[dl] \ar@/^.5pc/[rr]    \ar[d]^\cong                                                                       & (\FSpan_2,\times)    \ar@{-->}[d]^(.35){\cong}    & \tilde{(\FPinj_2,\times)} \ar[dl]       \ar[d]^\cong       \\
\tilde{(\FPinj_2,\times)}^\op \ar@/_1pc/[rr]            \ar[d]_{\cong}                               &      \TOF \ar[dl] \ar@/^.5pc/[rr]  \ar@{=}[d]       &                                  \ar@{-->}[d]^(.35){\cong}             & \tilde \TOF  \ar[d]_{\cong} \ar[dl]\\
\tilde{\TOF}^\op \ar@/_1pc/[rr]   \ar[d]_{\cong}   &                  \TOF \ar[dl] \ar@/^.5pc/[rr]                                                                      &  \ar@{-->}[d]^(.35){\cong}  & c(\TOF)  \ar[dl]\\
c(\TOF)^\op        \ar@/_1pc/[rr]                          &                                                                                             &          \hat\TOF &                        &            \\
}
$$

All of the rear and left faces commute. Moreover, the vertical maps are isomorphisms, therefore the maps induced by universal property of the pushout are isomorphisms.

\end{proof}

\section{Identities of \texorpdfstring{$\TOF$}{TOF}}
\label{sec:tof}

Define the category $\TOF$ \cite{tof} to be the PROP, generated by the 1 ancillary bits $| 1\rangle$ and $\langle 1|$ as well as the Toffoli gate, satisfying the identities given in Figure \ref{fig:TOF}.

\begin{figure}[h]
\noindent
\scalebox{1.0}{%
\vbox{%
\begin{mdframed}
\begin{multicols}{2}
\begin{enumerate}[label={\bf [TOF.\arabic*]}, ref={\bf [TOF.\arabic*]}, wide = 0pt, leftmargin = 2em]
\item
\label{TOF.1}
{\hfil
$
\begin{tabular}{c}
\begin{tikzpicture}
	\begin{pgfonlayer}{nodelayer}
		\node [style=nothing] (0) at (0, -0) {};
		\node [style=nothing] (1) at (0, 0.5000002) {};
		\node [style=oplus] (2) at (0.5000002, -0) {};
		\node [style=dot] (3) at (0.5000002, 0.5000002) {};
		\node [style=dot] (4) at (0.5000002, 1) {};
		\node [style=onein] (5) at (0, 1) {};
		\node [style=nothing] (6) at (1, 1) {};
		\node [style=nothing] (7) at (1, 0.5000002) {};
		\node [style=nothing] (8) at (1, -0) {};
	\end{pgfonlayer}
	\begin{pgfonlayer}{edgelayer}
		\draw (5) to (4);
		\draw (4) to (6);
		\draw (7) to (3);
		\draw (1) to (3);
		\draw (0) to (2);
		\draw (2) to (8);
		\draw (2) to (3);
		\draw (3) to (4);
	\end{pgfonlayer}
\end{tikzpicture}
=
\begin{tikzpicture}
\begin{pgfonlayer}{nodelayer}
\node [style=nothing] (0) at (0, 0) {};
\node [style=nothing] (1) at (0, 0.5) {};
\node [style=oplus] (2) at (0.5, 0) {};
\node [style=dot] (3) at (0.5, 0.5) {};
\node [style=onein] (4) at (0.5, 1) {};
\node [style=nothing] (5) at (1, 1) {};
\node [style=nothing] (6) at (1, 0.5) {};
\node [style=nothing] (7) at (1, 0) {};
\end{pgfonlayer}
\begin{pgfonlayer}{edgelayer}
\draw (1) to (3);
\draw (0) to (2);
\draw (2) to (3);
\draw (6) to (3);
\draw (4) to (5);
\draw (2) to (7);
\end{pgfonlayer}
\end{tikzpicture}\\
{}\\
\end{tabular}
$}

\item
\label{TOF.2}
{\hfil
$
\begin{tabular}{c}
\begin{tikzpicture}
	\begin{pgfonlayer}{nodelayer}
		\node [style=nothing] (0) at (0, 1.25) {};
		\node [style=nothing] (1) at (0, 0.7500001) {};
		\node [style=nothing] (2) at (2, 1.75) {};
		\node [style=nothing] (3) at (2, 1.25) {};
		\node [style=nothing] (4) at (2, 0.7500001) {};
		\node [style=dot] (5) at (1, 1.75) {};
		\node [style=dot] (6) at (1, 1.25) {};
		\node [style=oplus] (7) at (1, 0.7500001) {};
		\node [style=zeroin] (8) at (0, 1.75) {};
	\end{pgfonlayer}
	\begin{pgfonlayer}{edgelayer}
		\draw (5) to (2);
		\draw (3) to (6);
		\draw (6) to (0);
		\draw (1) to (7);
		\draw (7) to (4);
		\draw (7) to (6);
		\draw (6) to (5);
		\draw (8) to (5);
	\end{pgfonlayer}
\end{tikzpicture}
=
\begin{tikzpicture}
	\begin{pgfonlayer}{nodelayer}
		\node [style=nothing] (0) at (0, 1.25) {};
		\node [style=nothing] (1) at (0, 0.7500001) {};
		\node [style=nothing] (2) at (1.5, 1.75) {};
		\node [style=nothing] (3) at (1.5, 1.25) {};
		\node [style=nothing] (4) at (1.5, 0.7500001) {};
		\node [style=zeroin] (5) at (0, 1.75) {};
	\end{pgfonlayer}
	\begin{pgfonlayer}{edgelayer}
		\draw (5) to (2);
		\draw (0) to (3);
		\draw (1) to (4);
	\end{pgfonlayer}
\end{tikzpicture}\\
{}\\
\end{tabular}%
$}
\end{enumerate}
\end{multicols}
\
\end{mdframed}
}}
\caption{The identities of \texorpdfstring{$\TOF$}{TOF}}
\label{fig:TOF}
\end{figure}

The Toffoli gate and the 1-ancillary bits allow $\cnot$, $\Not$, $\zeroin$, $\zeroout$, and flipped $\tof$ gate and flipped $\cnot$ gate can defined in this setting:

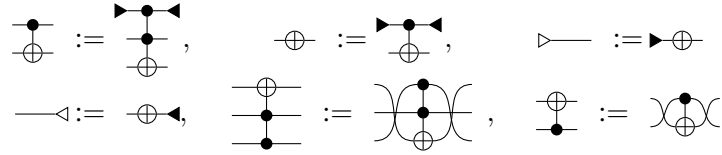
\[ \begin{array}{ccc}
\begin{tikzpicture}
	\begin{pgfonlayer}{nodelayer}
		\node [style=nothing] (0) at (0, -0) {};
		\node [style=nothing] (1) at (0, 0.5000001) {};
		\node [style=nothing] (2) at (1, 0.5000001) {};
		\node [style=nothing] (3) at (1, -0) {};
		\node [style=oplus] (4) at (0.5000001, -0) {};
		\node [style=dot] (5) at (0.5000001, 0.5000001) {};
	\end{pgfonlayer}
	\begin{pgfonlayer}{edgelayer}
		\draw (1) to (5);
		\draw (5) to (2);
		\draw (3) to (4);
		\draw (4) to (0);
		\draw (4) to (5);
	\end{pgfonlayer}
\end{tikzpicture}
:=
\begin{tikzpicture}
	\begin{pgfonlayer}{nodelayer}
		\node [style=nothing] (0) at (0, -0) {};
		\node [style=nothing] (1) at (0, 0.5000001) {};
		\node [style=nothing] (2) at (1, 0.5000001) {};
		\node [style=nothing] (3) at (1, -0) {};
		\node [style=oplus] (4) at (0.5000001, -0) {};
		\node [style=dot] (5) at (0.5000001, 0.5000001) {};
		\node [style=onein] (6) at (0, 1) {};
		\node [style=oneout] (7) at (1, 1) {};
		\node [style=dot] (8) at (0.5000001, 1) {};
	\end{pgfonlayer}
	\begin{pgfonlayer}{edgelayer}
		\draw (1) to (5);
		\draw (5) to (2);
		\draw (3) to (4);
		\draw (4) to (0);
		\draw (4) to (5);
		\draw (6) to (8);
		\draw (8) to (7);
		\draw (8) to (5);
	\end{pgfonlayer}
\end{tikzpicture}
,
&
\begin{tikzpicture}
	\begin{pgfonlayer}{nodelayer}
		\node [style=nothing] (0) at (0, -0) {};
		\node [style=nothing] (1) at (1, -0) {};
		\node [style=oplus] (2) at (0.5000001, -0) {};
	\end{pgfonlayer}
	\begin{pgfonlayer}{edgelayer}
		\draw (1) to (2);
		\draw (2) to (0);
	\end{pgfonlayer}
\end{tikzpicture}
:=
\begin{tikzpicture}
	\begin{pgfonlayer}{nodelayer}
		\node [style=nothing] (0) at (0, -0) {};
		\node [style=nothing] (1) at (1, -0) {};
		\node [style=oplus] (2) at (0.5000001, -0) {};
		\node [style=dot] (3) at (0.5000001, 0.5000001) {};
		\node [style=onein] (4) at (0, 0.5) {};
		\node [style=oneout] (5) at (1, 0.5) {};
	\end{pgfonlayer}
	\begin{pgfonlayer}{edgelayer}
		\draw (1) to (2);
		\draw (2) to (0);
		\draw (2) to (3);
		\draw (4) to (3);
		\draw (3) to (5);
	\end{pgfonlayer}
\end{tikzpicture}
,
&
\begin{tikzpicture}
	\begin{pgfonlayer}{nodelayer}
		\node [style=zeroin] (0) at (0, -0) {};
		\node [style=nothing] (1) at (1, -0) {};
	\end{pgfonlayer}
	\begin{pgfonlayer}{edgelayer}
		\draw (0) to (1);
	\end{pgfonlayer}
\end{tikzpicture}
:=
\begin{tikzpicture}
	\begin{pgfonlayer}{nodelayer}
		\node [style=nothing] (0) at (1, -0) {};
		\node [style=onein] (1) at (0, -0) {};
		\node [style=oplus] (2) at (0.5000001, -0) {};
	\end{pgfonlayer}
	\begin{pgfonlayer}{edgelayer}
		\draw (0) to (2);
		\draw (2) to (1);
	\end{pgfonlayer}
\end{tikzpicture}\\
\end{array}  \]

One can moreover construct generalized controlled not gates with arbitrarily many control wires in the obvious way.  Let $[x,X]$ denote a generalized Toffoli gate acting on the $x$th wire, controlled on the wires indexed by a set $X$. Then we can partially commute generalized controlled-not gates:

\begin{lemma} \cite[Lem. 7.2.6]{Cole}
\label{lemma:Iwama}
Let $[x,X]$ and $[y,Y]$ be generalized controlled not gates in $\TOF$ where $x\notin Y$.  We can perform the identities of Iwama et al. \cite{Iwama}, to commute them past each other with a trailing generalized controlled not gate as a side effect:
$$
 [y,{X\cup Y}] [y,{Y\sqcup\{x\}}] [x,X]
$$
\end{lemma}

In $\TOF$, one can define the diagonal map as follows:
$$

$$

\begin{lemma}\cite[\S 5.3.2]{Cole}
The diagonal map is a natural special commutative  \dag-symmetric 
monoidal  nonunital Frobenius algebra.
\end{lemma}

It is also natural on target qubits:
\begin{lemma}\cite[Lem. B.0.2 (iii)]{Cole}
\label{lemma:natoplus}
$$
%
$$
\end{lemma}

\section{The isomorphism between \texorpdfstring{$\ZXA$}{ZX\&} and  the (co)unitual completion of \texorpdfstring{$\TOF$}{TOF}}

We establish some basic properties of $\ZXA$ and  the (co)unitual completion of $\TOF$.
\subsection{Basic properties of the (co)unitual completion of \texorpdfstring{$\TOF$}{TOF}}

First, note that because $\tilde \TOF$ is a discrete Cartesian restriction category, it is a copy category and thus, for any map $f$ in $\TOF$

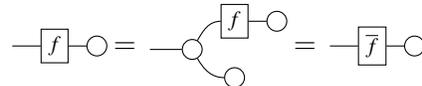
\begin{remark}
\label{cor:copy}

$$
%
$$
\end{remark}

  First, the $\cnot$ gate is its own mate on the second wire:
\begin{lemma}
\label{prop:twist}
$$
%
\end{align*}
\end{proof}

\subsection{Basic properties of \texorpdfstring{$\ZXA$}{ZX\&}}

\begin{lemma}
\label{lem:blackdot}
$$
%
\eq{\ref{ZXA.7}}
\end{align*}
\end{proof}

\begin{lemma}
The phase fusion of the black spider in $\ZXA$, 
$$
%
$$
in the presence of the other axioms is equivalent to asserting:
$$
%
$$
Or in other terms, the phase fusion of the black spider is equivalent to the interaction of the unit for and and the counit for copying as a bialgebra.
\end{lemma}

\begin{proof}
For the one direction, suppose that phase fusion holds:

\begin{align*}
%
\end{align*}

Conversely if the unit part of the bialgebra rule holds:

\begin{align*}
%
$$
\end{proof}

\subsection{Proof of Proposition \ref{prop:TOFZXA}}
\label{proof:TOFZXA}
Recall the statement of Proposition \ref{prop:TOFZXA}:

\textbf{Proposition \ref{prop:TOFZXA}:} \textit{The interpretation $\llbracket\_\rrbracket_{\ZXA}:\ZXA\to\hat \TOF$ is a strict symmetric monoidal functor.}

\begin{proof}
We prove that all of the axioms of $\ZXA$ hold in $\hat \TOF$:
\begin{enumerate}
\item[\ref{ZXA.1}:]
\begin{description}
\item[Unitality:] By Lemma \ref{lemma:whiteunit}:

\begin{align*}
\left\llbracket
%
\right\rrbracket_{\ZXA}
\end{align*}

\item[\ref{ZXA.3}:]
This is immediate.

\item[\ref{ZXA.4}:]
This is immediate.

\item[\ref{ZXA.5}:]
\begin{align*}
\left\llbracket
%
\right\rrbracket_{\ZXA}
$$

\item[\ref{ZXA.7}:]
This is immediate.

\item[\ref{ZXA.8}:]
\begin{align*}
\left\llbracket
%
\right\rrbracket_{\ZXA}
\end{align*}

\item[\ref{ZXA.16}:]
This is precisely \ref{TOF.7}.

\item[\ref{ZXA.17}:]
\begin{align*}
\left\llbracket
%
\right\rrbracket_{\ZXA}
\end{align*}

\end{enumerate}

\end{proof}

\subsection{Proof of Proposition \ref{prop:ZXATOF}}
\label{proof:ZXATOF}
Recall the statement of Proposition \ref{prop:ZXATOF}:

\textbf{Proposition \ref{prop:ZXATOF}:} \textit{The interpretation $\llbracket\_\rrbracket_{\hat \TOF}:\hat \TOF\to\ZXA$ is a strict symmetric monoidal functor.}

\begin{proof}
First, observe:
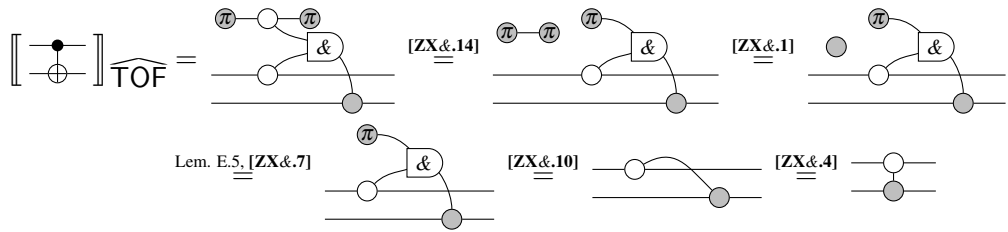
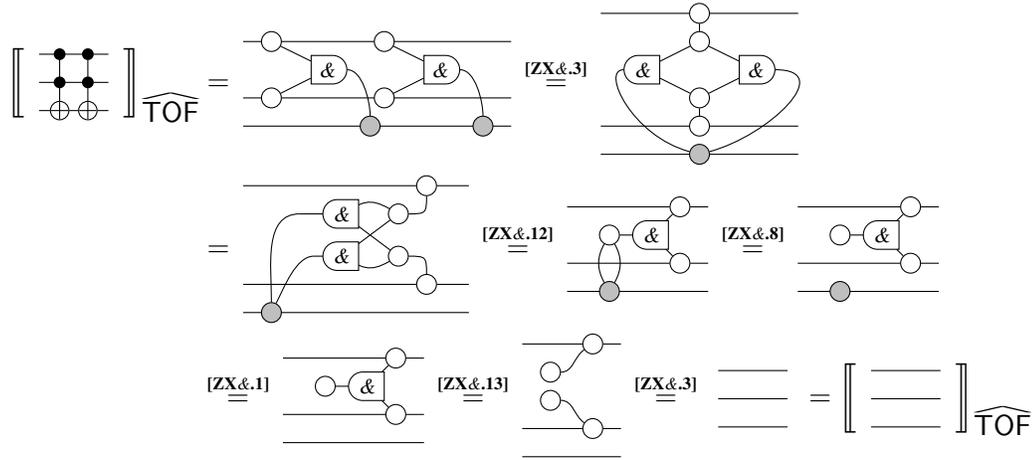
\begin{align*}
\left\llbracket
%
\end{align*}

We prove that all of the axioms of $\hat \TOF$ hold in $\ZXA$ :
\begin{enumerate}
\item[\ref{TOF.1}:]
\begin{align*}
\left\llbracket
%
\right\rrbracket_{\hat{\TOF}}
\end{align*}

\item[\ref{TOF.3}:]
This follows from the spider law.

\item[\ref{TOF.4}:]
This follows from the spider law.

\item[\ref{TOF.5}:]
This follows from the spider law.

\item[\ref{TOF.6}:]
This follows from the spider law.

\item[\ref{TOF.7}:]
\begin{align*}
\left\llbracket
%
\right\rrbracket_{\hat{\TOF}}
\end{align*}

\item[\ref{TOF.8}:]
This follows immediately from Lemma \ref{lem:blackdot} and \ref{ZXA.7}.

\item[\ref{TOF.9}:]

\begin{align*}
\left\llbracket
%
\right\rrbracket_{\hat{\TOF}}
\end{align*}

\item[\ref{TOF.10}:]  It is easier to prove that $\ref{TOF.10}$ is redundant.  Given \ref{TOF.9},  \ref{TOF.6} and \ref{TOF.12}, \ref{TOF.10} is equivalent to the following:
$$
%
\right\rrbracket_{\hat{\TOF}}
\end{align*}
\endgroup

\end{enumerate}
Where unitality and counitality follow from the fact that the white spiders are Frobenius algebras.  Also, we must also note that both Frobenius algebras induce the same compact closed structure, as is implied by the spider law;  this is immediate.

\end{proof}

\subsection{Proof of Theorem \ref{theorem:TOFZXAiso}}
\label{proof:TOFZXAiso}

\textbf{Theorem \ref{theorem:TOFZXAiso}} \textit{The interpretations $\llbracket\_\rrbracket_{\ZXA}$ and $\llbracket\_\rrbracket_{\hat \TOF}$ are inverses, so that $\hat \TOF$ and $\ZXA$ are isomorphic as strongly compact closed props.}

\begin{proof}
First we show that $\llbracket\llbracket\_\rrbracket_{\ZXA}\rrbracket_{\hat \TOF}=1$:
\begin{description}
\item[For the white spider:]
The case for the unit and counit is trivial.  For the (co)multiplication we have:
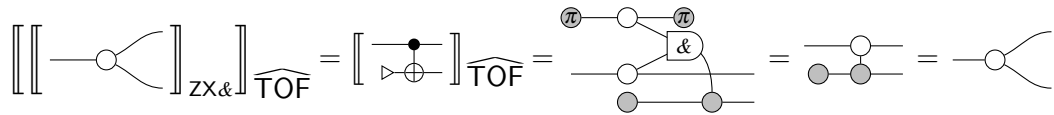
\begin{align*}
\left\llbracket\left\llbracket
\begin{tikzpicture}
	\begin{pgfonlayer}{nodelayer}
		\node [style=X] (0) at (5, -0) {};
		\node [style=none] (1) at (6, 0.5) {};
		\node [style=none] (2) at (6, -0.5) {};
		\node [style=none] (3) at (4, -0) {};
	\end{pgfonlayer}
	\begin{pgfonlayer}{edgelayer}
		\draw [in=-27, out=180, looseness=1.00] (2.center) to (0);
		\draw [in=180, out=27, looseness=1.00] (0) to (1.center);
		\draw (0) to (3.center);
	\end{pgfonlayer}
\end{tikzpicture}
\right\rrbracket_{\ZXA}\right\rrbracket_{\hat \TOF}
&=
\left\llbracket
\begin{tikzpicture}
	\begin{pgfonlayer}{nodelayer}
		\node [style=none] (0) at (5.5, 0.5) {};
		\node [style=none] (1) at (5.5, -0) {};
		\node [style=none] (2) at (4.25, 0.5) {};
		\node [style=dot] (3) at (5, 0.5) {};
		\node [style=oplus] (4) at (5, -0) {};
		\node [style=zeroin] (5) at (4.5, -0) {};
	\end{pgfonlayer}
	\begin{pgfonlayer}{edgelayer}
		\draw (1.center) to (4);
		\draw (4) to (5);
		\draw (4) to (3);
		\draw (3) to (0.center);
		\draw (3) to (2.center);
	\end{pgfonlayer}
\end{tikzpicture}
\right\rrbracket_{\hat \TOF}
=
\begin{tikzpicture}
	\begin{pgfonlayer}{nodelayer}
		\node [style=andin] (0) at (2.5, 1.5) {};
		\node [style=X] (1) at (1.5, 2) {};
		\node [style=X] (2) at (1.5, 1) {};
		\node [style=Z] (3) at (3, 0.5) {};
		\node [style=none] (4) at (2.5, 1.5) {};
		\node [style=Z] (5) at (0.5, 2) {$\pi$};
		\node [style=Z] (6) at (2.5, 2) {$\pi$};
		\node [style=none] (7) at (3.75, 0.5) {};
		\node [style=none] (8) at (3.75, 1) {};
		\node [style=none] (9) at (0.5, 1) {};
		\node [style=Z] (10) at (1.5, 0.5) {};
	\end{pgfonlayer}
	\begin{pgfonlayer}{edgelayer}
		\draw [in=0, out=90, looseness=1.00] (3) to (4.center);
		\draw (4.center) to (1);
		\draw (2) to (4.center);
		\draw (6) to (1);
		\draw (1) to (5);
		\draw (8.center) to (2);
		\draw (2) to (9.center);
		\draw (7.center) to (3);
		\draw (3) to (10);
	\end{pgfonlayer}
\end{tikzpicture}
=
\begin{tikzpicture}
	\begin{pgfonlayer}{nodelayer}
		\node [style=X] (0) at (2.25, 1) {};
		\node [style=Z] (1) at (2.25, 0.5) {};
		\node [style=none] (2) at (3, 0.5) {};
		\node [style=none] (3) at (3, 1) {};
		\node [style=none] (4) at (1.25, 1) {};
		\node [style=Z] (5) at (1.5, 0.5) {};
	\end{pgfonlayer}
	\begin{pgfonlayer}{edgelayer}
		\draw (3.center) to (0);
		\draw (0) to (4.center);
		\draw (2.center) to (1);
		\draw (1) to (5);
		\draw (1) to (0);
	\end{pgfonlayer}
\end{tikzpicture}
=
\begin{tikzpicture}
	\begin{pgfonlayer}{nodelayer}
		\node [style=X] (0) at (4.75, 0) {};
		\node [style=none] (1) at (5.75, 0.5) {};
		\node [style=none] (2) at (5.75, -0.5) {};
		\node [style=none] (3) at (4, 0) {};
	\end{pgfonlayer}
	\begin{pgfonlayer}{edgelayer}
		\draw [in=-27, out=180] (2.center) to (0);
		\draw [in=180, out=27] (0) to (1.center);
		\draw (0) to (3.center);
	\end{pgfonlayer}
\end{tikzpicture}
\end{align*}

\item[For the grey spider:]
The cases for the unit, counit and $\pi$ phase are trivial.  For the (co) multiplication we have:

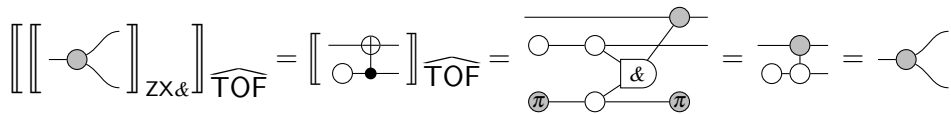
\begin{align*}
\left\llbracket\left\llbracket
\begin{tikzpicture}
	\begin{pgfonlayer}{nodelayer}
		\node [style=Z] (0) at (5, -0) {};
		\node [style=none] (1) at (5.75, 0.5) {};
		\node [style=none] (2) at (5.75, -0.5) {};
		\node [style=none] (3) at (4.5, -0) {};
	\end{pgfonlayer}
	\begin{pgfonlayer}{edgelayer}
		\draw [in=-27, out=180, looseness=1.00] (2.center) to (0);
		\draw [in=180, out=27, looseness=1.00] (0) to (1.center);
		\draw (0) to (3.center);
	\end{pgfonlayer}
\end{tikzpicture}
\right\rrbracket_{\ZXA}\right\rrbracket_{\hat \TOF}
&=
\left\llbracket
\begin{tikzpicture}
	\begin{pgfonlayer}{nodelayer}
		\node [style=none] (0) at (5.5, 0.5) {};
		\node [style=none] (1) at (5.5, -0) {};
		\node [style=none] (2) at (4.25, 0.5) {};
		\node [style=dot] (3) at (5, -0) {};
		\node [style=oplus] (4) at (5, 0.5) {};
		\node [style=X] (5) at (4.5, -0) {};
	\end{pgfonlayer}
	\begin{pgfonlayer}{edgelayer}
		\draw (4) to (3);
		\draw (3) to (1.center);
		\draw (5) to (3);
		\draw (4) to (0.center);
		\draw (4) to (2.center);
	\end{pgfonlayer}
\end{tikzpicture}
\right\rrbracket_{\hat \TOF}
=
 \begin{tikzpicture}
	\begin{pgfonlayer}{nodelayer}
		\node [style=X] (0) at (7, 2) {};
		\node [style=X] (1) at (7, 1) {};
		\node [style=none] (2) at (7.75, 1.5) {};
		\node [style=Z] (3) at (8.5, 2.5) {};
		\node [style=Z] (4) at (6, 1) {$\pi$};
		\node [style=Z] (5) at (8.5, 1) {$\pi$};
		\node [style=X] (6) at (6, 2) {};
		\node [style=none] (7) at (5.75, 2.5) {};
		\node [style=none] (8) at (9, 2.5) {};
		\node [style=none] (9) at (9, 2) {};
		\node [style=andin] (10) at (7.75, 1.5) {};
	\end{pgfonlayer}
	\begin{pgfonlayer}{edgelayer}
		\draw (4) to (1);
		\draw (1) to (2.center);
		\draw (5) to (1);
		\draw (2.center) to (0);
		\draw (2.center) to (3);
		\draw (9.center) to (0);
		\draw (0) to (6);
		\draw (7.center) to (3);
		\draw (3) to (8.center);
	\end{pgfonlayer}
\end{tikzpicture}
=
\begin{tikzpicture}
	\begin{pgfonlayer}{nodelayer}
		\node [style=X] (0) at (7, 2) {};
		\node [style=Z] (1) at (7, 2.5) {};
		\node [style=X] (2) at (6.5, 2) {};
		\node [style=none] (3) at (6.25, 2.5) {};
		\node [style=none] (4) at (7.5, 2.5) {};
		\node [style=none] (5) at (7.5, 2) {};
	\end{pgfonlayer}
	\begin{pgfonlayer}{edgelayer}
		\draw (5.center) to (0);
		\draw (0) to (2);
		\draw (3.center) to (1);
		\draw (1) to (4.center);
		\draw (0) to (1);
	\end{pgfonlayer}
\end{tikzpicture}
=
\begin{tikzpicture}
	\begin{pgfonlayer}{nodelayer}
		\node [style=Z] (0) at (5, -0) {};
		\node [style=none] (1) at (5.75, 0.5) {};
		\node [style=none] (2) at (5.75, -0.5) {};
		\node [style=none] (3) at (4.5, -0) {};
	\end{pgfonlayer}
	\begin{pgfonlayer}{edgelayer}
		\draw [in=-27, out=180, looseness=1.00] (2.center) to (0);
		\draw [in=180, out=27, looseness=1.00] (0) to (1.center);
		\draw (0) to (3.center);
	\end{pgfonlayer}
\end{tikzpicture}
\end{align*}

\item[For the {\sf and} gate:]
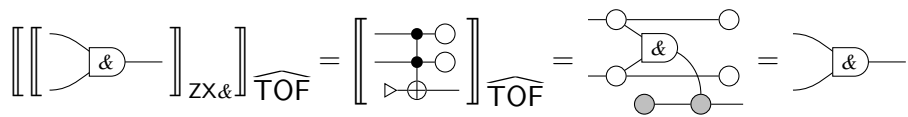
\begin{align*}
\left\llbracket\left\llbracket
\begin{tikzpicture}
	\begin{pgfonlayer}{nodelayer}
		\node [style=none] (0) at (5, -0) {};
		\node [style=none] (1) at (4, 0.5) {};
		\node [style=none] (2) at (4, -0.5) {};
		\node [style=none] (3) at (6, -0) {};
		\node [style=andin] (4) at (5, -0) {};
	\end{pgfonlayer}
	\begin{pgfonlayer}{edgelayer}
		\draw [in=-153, out=0, looseness=1.00] (2.center) to (0.center);
		\draw [in=0, out=153, looseness=1.00] (0.center) to (1.center);
		\draw (0.center) to (3.center);
	\end{pgfonlayer}
\end{tikzpicture}
\right\rrbracket_{\ZXA}\right\rrbracket_{\hat \TOF}
&=
\left\llbracket
\begin{tikzpicture}
	\begin{pgfonlayer}{nodelayer}
		\node [style=dot] (0) at (6, 2) {};
		\node [style=dot] (1) at (6, 1.5) {};
		\node [style=oplus] (2) at (6, 1) {};
		\node [style=zeroin] (3) at (5.5, 1) {};
		\node [style=X] (4) at (6.5, 2) {};
		\node [style=X] (5) at (6.5, 1.5) {};
		\node [style=none] (6) at (6.75, 1) {};
		\node [style=none] (7) at (5.25, 2) {};
		\node [style=none] (8) at (5.25, 1.5) {};
	\end{pgfonlayer}
	\begin{pgfonlayer}{edgelayer}
		\draw (5) to (1);
		\draw (0) to (1);
		\draw (1) to (2);
		\draw (2) to (3);
		\draw (2) to (6.center);
		\draw (4) to (0);
		\draw (0) to (7.center);
		\draw (8.center) to (1);
	\end{pgfonlayer}
\end{tikzpicture}
\right\rrbracket_{\hat \TOF}
=
\begin{tikzpicture}
	\begin{pgfonlayer}{nodelayer}
		\node [style=X] (0) at (2, 2) {};
		\node [style=X] (1) at (2, 1) {};
		\node [style=none] (2) at (2.75, 1.5) {};
		\node [style=Z] (3) at (3.5, 0.5) {};
		\node [style=Z] (4) at (2.5, 0.5) {};
		\node [style=X] (5) at (4, 1) {};
		\node [style=X] (6) at (4, 2) {};
		\node [style=none] (7) at (4.25, 0.5) {};
		\node [style=none] (8) at (1.5, 2) {};
		\node [style=none] (9) at (1.5, 1) {};
		\node [style=andin] (10) at (2.75, 1.5) {};
	\end{pgfonlayer}
	\begin{pgfonlayer}{edgelayer}
		\draw (2.center) to (0);
		\draw (1) to (2.center);
		\draw (3) to (4);
		\draw [in=0, out=90, looseness=1.00] (3) to (2.center);
		\draw (5) to (1);
		\draw (1) to (9.center);
		\draw (8.center) to (0);
		\draw (0) to (6);
		\draw (7.center) to (3);
	\end{pgfonlayer}
\end{tikzpicture}
=
\begin{tikzpicture}
	\begin{pgfonlayer}{nodelayer}
		\node [style=none] (0) at (5, -0) {};
		\node [style=none] (1) at (4, 0.5) {};
		\node [style=none] (2) at (4, -0.5) {};
		\node [style=none] (3) at (6, -0) {};
		\node [style=andin] (4) at (5, -0) {};
	\end{pgfonlayer}
	\begin{pgfonlayer}{edgelayer}
		\draw [in=-153, out=0, looseness=1.00] (2.center) to (0.center);
		\draw [in=0, out=153, looseness=1.00] (0.center) to (1.center);
		\draw (0.center) to (3.center);
	\end{pgfonlayer}
\end{tikzpicture}
\end{align*}
\end{description}

Next, we show that $\llbracket\llbracket\_\rrbracket_{\hat \TOF}\rrbracket_{\ZXA}=1$:
The ancillae are trivial.  For the Toffoli gate:
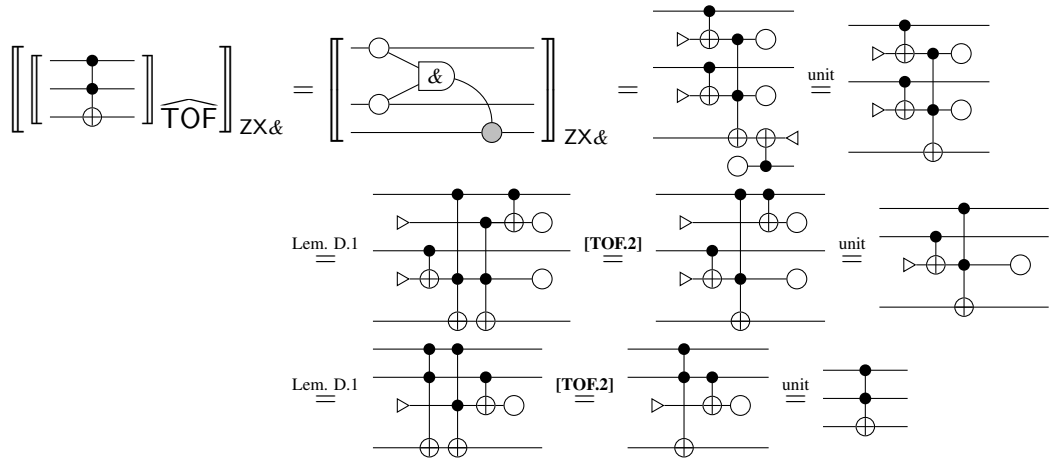
\begin{align*}
\left\llbracket\left\llbracket
\begin{tikzpicture}
	\begin{pgfonlayer}{nodelayer}
		\node [style=dot] (0) at (5, 2) {};
		\node [style=dot] (1) at (5, 1.5) {};
		\node [style=oplus] (2) at (5, 1) {};
		\node [style=none] (3) at (5.75, 1) {};
		\node [style=none] (4) at (4.25, 1) {};
		\node [style=none] (5) at (4.25, 2) {};
		\node [style=none] (6) at (4.25, 1.5) {};
		\node [style=none] (7) at (5.75, 1.5) {};
		\node [style=none] (8) at (5.75, 2) {};
	\end{pgfonlayer}
	\begin{pgfonlayer}{edgelayer}
		\draw (2) to (1);
		\draw (1) to (0);
		\draw (8.center) to (0);
		\draw (0) to (5.center);
		\draw (6.center) to (1);
		\draw (1) to (7.center);
		\draw (3.center) to (2);
		\draw (2) to (4.center);
	\end{pgfonlayer}
\end{tikzpicture}
\right\rrbracket_{\hat \TOF}\right\rrbracket_{\ZXA}
&=
\left\llbracket
\begin{tikzpicture}
	\begin{pgfonlayer}{nodelayer}
		\node [style=X] (0) at (2, 2) {};
		\node [style=X] (1) at (2, 1) {};
		\node [style=none] (2) at (3, 1.5) {};
		\node [style=Z] (3) at (4, 0.5) {};
		\node [style=none] (4) at (1.5, 0.5) {};
		\node [style=none] (5) at (4.75, 0.5) {};
		\node [style=none] (6) at (4.75, 1) {};
		\node [style=none] (7) at (4.75, 2) {};
		\node [style=none] (8) at (1.5, 2) {};
		\node [style=none] (9) at (1.5, 1) {};
		\node [style=andin] (10) at (3, 1.5) {};
	\end{pgfonlayer}
	\begin{pgfonlayer}{edgelayer}
		\draw (3) to (5.center);
		\draw (3) to (4.center);
		\draw (9.center) to (1);
		\draw (1) to (6.center);
		\draw (7.center) to (0);
		\draw (0) to (8.center);
		\draw (0) to (2.center);
		\draw [in=90, out=0, looseness=1.00] (2.center) to (3);
		\draw (2.center) to (1);
	\end{pgfonlayer}
\end{tikzpicture}
\right\rrbracket_{\ZXA}
=
\begin{tikzpicture}
	\begin{pgfonlayer}{nodelayer}
		\node [style=dot] (0) at (6.5, 2) {};
		\node [style=dot] (1) at (6.5, 1) {};
		\node [style=oplus] (2) at (6.5, 0.25) {};
		\node [style=X] (3) at (7, 2) {};
		\node [style=X] (4) at (7, 1) {};
		\node [style=oplus] (5) at (7, 0.25) {};
		\node [style=dot] (6) at (7, -0.25) {};
		\node [style=X] (7) at (6.5, -0.25) {};
		\node [style=none] (8) at (7.5, -0.25) {};
		\node [style=dot] (9) at (6, 2.5) {};
		\node [style=oplus] (10) at (6, 2) {};
		\node [style=zeroin] (11) at (5.5, 2) {};
		\node [style=oplus] (12) at (6, 1) {};
		\node [style=dot] (13) at (6, 1.5) {};
		\node [style=zeroin] (14) at (5.5, 1) {};
		\node [style=none] (15) at (7.5, 1.5) {};
		\node [style=none] (16) at (7.5, 2.5) {};
		\node [style=none] (17) at (5, 1.5) {};
		\node [style=none] (18) at (5, 2.5) {};
		\node [style=none] (19) at (5, 0.25) {};
		\node [style=zeroout] (20) at (7.5, 0.25) {};
	\end{pgfonlayer}
	\begin{pgfonlayer}{edgelayer}
		\draw (8.center) to (6);
		\draw (6) to (5);
		\draw (5) to (2);
		\draw (7) to (6);
		\draw (2) to (1);
		\draw (1) to (0);
		\draw (0) to (3);
		\draw (4) to (1);
		\draw (10) to (11);
		\draw (10) to (9);
		\draw (12) to (14);
		\draw (12) to (13);
		\draw (1) to (12);
		\draw (10) to (0);
		\draw (16.center) to (9);
		\draw (9) to (18.center);
		\draw (17.center) to (13);
		\draw (13) to (15.center);
		\draw (20) to (5);
		\draw (2) to (19.center);
	\end{pgfonlayer}
\end{tikzpicture}
\eq{unit}
\begin{tikzpicture}
	\begin{pgfonlayer}{nodelayer}
		\node [style=dot] (0) at (6.5, 2) {};
		\node [style=dot] (1) at (6.5, 1) {};
		\node [style=oplus] (2) at (6.5, 0.25) {};
		\node [style=X] (3) at (7, 2) {};
		\node [style=X] (4) at (7, 1) {};
		\node [style=none] (5) at (7.5, 0.25) {};
		\node [style=dot] (6) at (6, 2.5) {};
		\node [style=oplus] (7) at (6, 2) {};
		\node [style=zeroin] (8) at (5.5, 2) {};
		\node [style=oplus] (9) at (6, 1) {};
		\node [style=dot] (10) at (6, 1.5) {};
		\node [style=zeroin] (11) at (5.5, 1) {};
		\node [style=none] (12) at (7.5, 1.5) {};
		\node [style=none] (13) at (7.5, 2.5) {};
		\node [style=none] (14) at (5, 1.5) {};
		\node [style=none] (15) at (5, 2.5) {};
		\node [style=none] (16) at (5, 0.25) {};
	\end{pgfonlayer}
	\begin{pgfonlayer}{edgelayer}
		\draw (2) to (1);
		\draw (1) to (0);
		\draw (0) to (3);
		\draw (4) to (1);
		\draw (7) to (8);
		\draw (7) to (6);
		\draw (9) to (11);
		\draw (9) to (10);
		\draw (1) to (9);
		\draw (7) to (0);
		\draw (13.center) to (6);
		\draw (6) to (15.center);
		\draw (14.center) to (10);
		\draw (10) to (12.center);
		\draw (2) to (16.center);
		\draw (5.center) to (2);
	\end{pgfonlayer}
\end{tikzpicture}\\
&\eq{Lem. \ref{lemma:Iwama}}
\begin{tikzpicture}
	\begin{pgfonlayer}{nodelayer}
		\node [style=dot] (0) at (6.25, 2) {};
		\node [style=dot] (1) at (6.25, 1) {};
		\node [style=oplus] (2) at (6.25, 0.25) {};
		\node [style=X] (3) at (7.25, 2) {};
		\node [style=X] (4) at (7.25, 1) {};
		\node [style=none] (5) at (7.75, 0.25) {};
		\node [style=dot] (6) at (6.75, 2.5) {};
		\node [style=oplus] (7) at (6.75, 2) {};
		\node [style=zeroin] (8) at (4.75, 2) {};
		\node [style=oplus] (9) at (5.25, 1) {};
		\node [style=dot] (10) at (5.25, 1.5) {};
		\node [style=zeroin] (11) at (4.75, 1) {};
		\node [style=none] (12) at (7.75, 1.5) {};
		\node [style=none] (13) at (7.75, 2.5) {};
		\node [style=none] (14) at (4.25, 1.5) {};
		\node [style=none] (15) at (4.25, 2.5) {};
		\node [style=none] (16) at (4.25, 0.25) {};
		\node [style=dot] (17) at (5.75, 2.5) {};
		\node [style=dot] (18) at (5.75, 1) {};
		\node [style=oplus] (19) at (5.75, 0.25) {};
	\end{pgfonlayer}
	\begin{pgfonlayer}{edgelayer}
		\draw (2) to (1);
		\draw (1) to (0);
		\draw (0) to (3);
		\draw (4) to (1);
		\draw (7) to (8);
		\draw (7) to (6);
		\draw (9) to (11);
		\draw (9) to (10);
		\draw (1) to (9);
		\draw (7) to (0);
		\draw (13.center) to (6);
		\draw (6) to (15.center);
		\draw (14.center) to (10);
		\draw (10) to (12.center);
		\draw (2) to (16.center);
		\draw (5.center) to (2);
		\draw (19) to (18);
		\draw (18) to (17);
	\end{pgfonlayer}
\end{tikzpicture}
\eq{\ref{TOF.2}}
\begin{tikzpicture}
	\begin{pgfonlayer}{nodelayer}
		\node [style=X] (0) at (6.75, 2) {};
		\node [style=X] (1) at (6.75, 1) {};
		\node [style=none] (2) at (7.25, 0.25) {};
		\node [style=dot] (3) at (6.25, 2.5) {};
		\node [style=oplus] (4) at (6.25, 2) {};
		\node [style=zeroin] (5) at (4.75, 2) {};
		\node [style=oplus] (6) at (5.25, 1) {};
		\node [style=dot] (7) at (5.25, 1.5) {};
		\node [style=zeroin] (8) at (4.75, 1) {};
		\node [style=none] (9) at (7.25, 1.5) {};
		\node [style=none] (10) at (7.25, 2.5) {};
		\node [style=none] (11) at (4.25, 1.5) {};
		\node [style=none] (12) at (4.25, 2.5) {};
		\node [style=none] (13) at (4.25, 0.25) {};
		\node [style=dot] (14) at (5.75, 2.5) {};
		\node [style=dot] (15) at (5.75, 1) {};
		\node [style=oplus] (16) at (5.75, 0.25) {};
	\end{pgfonlayer}
	\begin{pgfonlayer}{edgelayer}
		\draw (4) to (5);
		\draw (4) to (3);
		\draw (6) to (8);
		\draw (6) to (7);
		\draw (10.center) to (3);
		\draw (3) to (12.center);
		\draw (11.center) to (7);
		\draw (7) to (9.center);
		\draw (16) to (15);
		\draw (15) to (14);
		\draw (2.center) to (13.center);
		\draw (6) to (1);
		\draw (0) to (4);
	\end{pgfonlayer}
\end{tikzpicture}
\eq{unit}
\begin{tikzpicture}
	\begin{pgfonlayer}{nodelayer}
		\node [style=X] (0) at (6.75, 1) {};
		\node [style=none] (1) at (7.25, 0.25) {};
		\node [style=oplus] (2) at (5.25, 1) {};
		\node [style=dot] (3) at (5.25, 1.5) {};
		\node [style=zeroin] (4) at (4.75, 1) {};
		\node [style=none] (5) at (7.25, 1.5) {};
		\node [style=none] (6) at (7.25, 2) {};
		\node [style=none] (7) at (4.25, 1.5) {};
		\node [style=none] (8) at (4.25, 2) {};
		\node [style=none] (9) at (4.25, 0.25) {};
		\node [style=dot] (10) at (5.75, 2) {};
		\node [style=dot] (11) at (5.75, 1) {};
		\node [style=oplus] (12) at (5.75, 0.25) {};
	\end{pgfonlayer}
	\begin{pgfonlayer}{edgelayer}
		\draw (2) to (4);
		\draw (2) to (3);
		\draw (7.center) to (3);
		\draw (3) to (5.center);
		\draw (12) to (11);
		\draw (11) to (10);
		\draw (1.center) to (9.center);
		\draw (2) to (0);
		\draw (6.center) to (8.center);
	\end{pgfonlayer}
\end{tikzpicture}\\
&\eq{Lem.  \ref{lemma:Iwama}}
\begin{tikzpicture}
	\begin{pgfonlayer}{nodelayer}
		\node [style=X] (0) at (6.75, 1) {};
		\node [style=none] (1) at (7.25, 0.25) {};
		\node [style=oplus] (2) at (6.25, 1) {};
		\node [style=dot] (3) at (6.25, 1.5) {};
		\node [style=zeroin] (4) at (4.75, 1) {};
		\node [style=none] (5) at (7.25, 1.5) {};
		\node [style=none] (6) at (7.25, 2) {};
		\node [style=none] (7) at (4.25, 1.5) {};
		\node [style=none] (8) at (4.25, 2) {};
		\node [style=none] (9) at (4.25, 0.25) {};
		\node [style=dot] (10) at (5.75, 2) {};
		\node [style=dot] (11) at (5.75, 1) {};
		\node [style=oplus] (12) at (5.75, 0.25) {};
		\node [style=dot] (13) at (5.25, 2) {};
		\node [style=dot] (14) at (5.25, 1.5) {};
		\node [style=oplus] (15) at (5.25, 0.25) {};
	\end{pgfonlayer}
	\begin{pgfonlayer}{edgelayer}
		\draw (2) to (4);
		\draw (2) to (3);
		\draw (7.center) to (3);
		\draw (3) to (5.center);
		\draw (12) to (11);
		\draw (11) to (10);
		\draw (1.center) to (9.center);
		\draw (2) to (0);
		\draw (6.center) to (8.center);
		\draw (15) to (13);
	\end{pgfonlayer}
\end{tikzpicture}
\eq{\ref{TOF.2}}
\begin{tikzpicture}
	\begin{pgfonlayer}{nodelayer}
		\node [style=X] (0) at (6.25, 1) {};
		\node [style=none] (1) at (6.75, 0.25) {};
		\node [style=oplus] (2) at (5.75, 1) {};
		\node [style=dot] (3) at (5.75, 1.5) {};
		\node [style=zeroin] (4) at (4.75, 1) {};
		\node [style=none] (5) at (6.75, 1.5) {};
		\node [style=none] (6) at (6.75, 2) {};
		\node [style=none] (7) at (4.25, 1.5) {};
		\node [style=none] (8) at (4.25, 2) {};
		\node [style=none] (9) at (4.25, 0.25) {};
		\node [style=dot] (10) at (5.25, 2) {};
		\node [style=dot] (11) at (5.25, 1.5) {};
		\node [style=oplus] (12) at (5.25, 0.25) {};
	\end{pgfonlayer}
	\begin{pgfonlayer}{edgelayer}
		\draw (2) to (4);
		\draw (2) to (3);
		\draw (7.center) to (3);
		\draw (3) to (5.center);
		\draw (1.center) to (9.center);
		\draw (2) to (0);
		\draw (6.center) to (8.center);
		\draw (12) to (10);
	\end{pgfonlayer}
\end{tikzpicture}
\eq{unit}
\begin{tikzpicture}
	\begin{pgfonlayer}{nodelayer}
		\node [style=dot] (0) at (5, 2) {};
		\node [style=dot] (1) at (5, 1.5) {};
		\node [style=oplus] (2) at (5, 1) {};
		\node [style=none] (3) at (5.75, 1) {};
		\node [style=none] (4) at (4.25, 1) {};
		\node [style=none] (5) at (4.25, 2) {};
		\node [style=none] (6) at (4.25, 1.5) {};
		\node [style=none] (7) at (5.75, 1.5) {};
		\node [style=none] (8) at (5.75, 2) {};
	\end{pgfonlayer}
	\begin{pgfonlayer}{edgelayer}
		\draw (2) to (1);
		\draw (1) to (0);
		\draw (8.center) to (0);
		\draw (0) to (5.center);
		\draw (6.center) to (1);
		\draw (1) to (7.center);
		\draw (3.center) to (2);
		\draw (2) to (4.center);
	\end{pgfonlayer}
\end{tikzpicture}
\end{align*}
\end{proof}

\section{The identities of \texorpdfstring{$\CNOT$}{CNOT}}

The category $\CNOT$  \cite{cnot} is the $\dag$-symmetric monoidal subcategory of $\TOF$ generated by the controlled not gate and ancillary bits $|1\rangle$, $\langle 1|$.  A complete set of identities is presented in the following figure, because some of the identities are used in the translation between $\ZXA$ and the (co)unital completion of $\TOF$.

\begin{figure}[H]
	\noindent
	\scalebox{1.0}{%
		\vbox{%
			\begin{mdframed}
				\begin{multicols}{2}
					\begin{enumerate}[label={\bf [CNOT.\arabic*]}, ref={\bf [CNOT.\arabic*]}, wide = 0pt, leftmargin = 2em]
						\item
						\label{CNOT.1}
						{\hfil
							$
							\begin{tikzpicture}
							\begin{pgfonlayer}{nodelayer}
							\node [style=nothing] (0) at (0, 0) {};
							\node [style=nothing] (1) at (0, .5) {};
							\node [style=oplus] (2) at (.5, 0) {};
							\node [style=dot] (3) at (.5, .5) {};
							\node [style=dot] (4) at (1, 0) {};
							\node [style=oplus] (5) at (1, .5) {};
							\node [style=oplus] (6) at (1.5, 0) {};
							\node [style=dot] (7) at (1.5, .5) {};
							\node [style=nothing] (8) at (2, 0) {};
							\node [style=nothing] (9) at (2, .5) {};
							\end{pgfonlayer}
							\begin{pgfonlayer}{edgelayer}
							\draw [style=simple] (0) to (8);
							\draw [style=simple] (1) to (9);
							\draw [style=simple] (2) to (3);
							\draw [style=simple] (4) to (5);
							\draw [style=simple] (6) to (7);
							\end{pgfonlayer}
							\end{tikzpicture}
							=
							\begin{tikzpicture}
	\begin{pgfonlayer}{nodelayer}
		\node [style=nothing] (0) at (0, -0) {};
		\node [style=nothing] (1) at (0, 0.5000002) {};
		\node [style=nothing] (2) at (1, 0.5000002) {};
		\node [style=nothing] (3) at (1, -0) {};
	\end{pgfonlayer}
	\begin{pgfonlayer}{edgelayer}
		\draw [in=180, out=0, looseness=1.25] (1) to (3);
		\draw [in=180, out=0, looseness=1.25] (0) to (2);
	\end{pgfonlayer}
\end{tikzpicture}$}

						\item
						\label{CNOT.2}
						\hfil{
							$
							\begin{tikzpicture}
							\begin{pgfonlayer}{nodelayer}
							\node [style=nothing] (0) at (0, 0) {};
							\node [style=nothing] (1) at (0, .5) {};
							\node [style=oplus] (2) at (.5, 0) {};
							\node [style=dot] (3) at (.5, .5) {};
							\node [style=oplus] (6) at (1, 0) {};
							\node [style=dot] (7) at (1, .5) {};
							\node [style=nothing] (8) at (1.5, 0) {};
							\node [style=nothing] (9) at (1.5, .5) {};
							\end{pgfonlayer}
							\begin{pgfonlayer}{edgelayer}
							\draw [style=simple] (0) to (8);
							\draw [style=simple] (1) to (9);
							\draw [style=simple] (2) to (3);
							\draw [style=simple] (6) to (7);
							\end{pgfonlayer}
							\end{tikzpicture}
							=
							\begin{tikzpicture}
							\begin{pgfonlayer}{nodelayer}
							\node [style=nothing] (0) at (0, 0) {};
							\node [style=nothing] (1) at (0, .5) {};
							\node [style=nothing] (3) at (1.5, 0) {};
							\node [style=nothing] (4) at (1.5, .5) {};
							\end{pgfonlayer}
							\begin{pgfonlayer}{edgelayer}
							\draw [style=simple] (0) to (3);
							\draw [style=simple] (1) to (4);
							\end{pgfonlayer}
							\end{tikzpicture}
							$}
						
						\item
						\label{CNOT.3}
						\hfil{
							$
							\begin{tikzpicture}
							\begin{pgfonlayer}{nodelayer}
							\node [style=nothing] (0) at (0, 1) {};
							\node [style=nothing] (1) at (0, .5) {};
							\node [style=nothing] (2) at (0, 0) {};
							\node [style=oplus] (3) at (.75, 1) {};
							\node [style=dot] (4) at (.75, .5) {};
							\node [style=dot] (5) at (1.25, .5) {};
							\node [style=oplus] (6) at (1.25, 0) {};
							\node [style=nothing] (7) at (2, 1) {};
							\node [style=nothing] (8) at (2, .5) {};
							\node [style=nothing] (9) at (2, 0) {};
							\end{pgfonlayer}
							\begin{pgfonlayer}{edgelayer}
							\draw [style=simple] (0) to (7);
							\draw [style=simple] (1) to (8);
							\draw [style=simple] (2) to (9);
							\draw [style=simple] (3) to (4);
							\draw [style=simple] (5) to (6);
							\end{pgfonlayer}
							\end{tikzpicture}
							=
							\begin{tikzpicture}
							\begin{pgfonlayer}{nodelayer}
							\node [style=nothing] (0) at (0, 1) {};
							\node [style=nothing] (1) at (0, .5) {};
							\node [style=nothing] (2) at (0, 0) {};
							\node [style=oplus] (3) at (1.25, 1) {};
							\node [style=dot] (4) at (1.25, .5) {};
							\node [style=dot] (5) at (.75, .5) {};
							\node [style=oplus] (6) at (.75, 0) {};
							\node [style=nothing] (7) at (2, 1) {};
							\node [style=nothing] (8) at (2, .5) {};
							\node [style=nothing] (9) at (2, 0) {};
							\end{pgfonlayer}
							\begin{pgfonlayer}{edgelayer}
							\draw [style=simple] (0) to (7);
							\draw [style=simple] (1) to (8);
							\draw [style=simple] (2) to (9);
							\draw [style=simple] (3) to (4);
							\draw [style=simple] (5) to (6);
							\end{pgfonlayer}
							\end{tikzpicture}
							$}
						
						\item 
						\label{CNOT.4}
						\hfil{
							\begin{tabular}{c}
							$
							\begin{tikzpicture}
							\begin{pgfonlayer}{nodelayer}
							\node [style=onein] (0) at (0, .5) {};
							\node [style=nothing] (1) at (0, 0) {};
							\node [style=dot] (2) at (.5, .5) {};
							\node [style=oplus] (3) at (.5, 0) {};
							\node [style=nothing] (4) at (1, .5) {};
							\node [style=nothing] (5) at (1, 0) {};
							\end{pgfonlayer}
							\begin{pgfonlayer}{edgelayer}
							\draw [style=simple] (0) to (4);
							\draw [style=simple] (1) to (5);
							\draw [style=simple] (2) to (3);
							\end{pgfonlayer}
							\end{tikzpicture}
							=
							\begin{tikzpicture}
							\begin{pgfonlayer}{nodelayer}
							\node [style=onein] (0) at (0, .5) {};
							\node [style=nothing] (1) at (0, 0) {};
							\node [style=dot] (2) at (.5, .5) {};
							\node [style=oplus] (3) at (.5, 0) {};
							\node [style=oneout] (4) at (1, .5) {};
							\node [style=nothing] (5) at (2, 0) {};
							\node [style=onein] (6) at (1.5, .5) {};
							\node [style=nothing] (7) at (2, 0.5) {};
							\end{pgfonlayer}
							\begin{pgfonlayer}{edgelayer}
							\draw [style=simple] (0) to (4);
							\draw [style=simple] (1) to (5);
							\draw [style=simple] (2) to (3);
							\draw [style=simple] (6) to (7);
							\end{pgfonlayer}
							\end{tikzpicture}$\\
							$ $\\
							$\begin{tikzpicture}
							\begin{pgfonlayer}{nodelayer}
							\node [style=nothing] (0) at (0, .5) {};
							\node [style=nothing] (1) at (0, 0) {};
							\node [style=dot] (2) at (.5, .5) {};
							\node [style=oplus] (3) at (.5, 0) {};
							\node [style=oneout] (4) at (1, .5) {};
							\node [style=nothing] (5) at (1, 0) {};
							\end{pgfonlayer}
							\begin{pgfonlayer}{edgelayer}
							\draw [style=simple] (0) to (4);
							\draw [style=simple] (1) to (5);
							\draw [style=simple] (2) to (3);
							\end{pgfonlayer}
							\end{tikzpicture}
							=
							\begin{tikzpicture}
							\begin{pgfonlayer}{nodelayer}
							\node [style=oneout] (0) at (2, .5) {};
							\node [style=nothing] (1) at (2, 0) {};
							\node [style=dot] (2) at (1.5, .5) {};
							\node [style=oplus] (3) at (1.5, 0) {};
							\node [style=onein] (4) at (1, .5) {};
							\node [style=nothing] (5) at (0, 0) {};
							\node [style=oneout] (6) at (.5, .5) {};
							\node [style=nothing] (7) at (0, 0.5) {};
							\end{pgfonlayer}
							\begin{pgfonlayer}{edgelayer}
							\draw [style=simple] (0) to (4);
							\draw [style=simple] (1) to (5);
							\draw [style=simple] (2) to (3);
							\draw [style=simple] (6) to (7);
							\end{pgfonlayer}
							\end{tikzpicture}$
							\end{tabular}
							}
						
						\item 
						\label{CNOT.5}
						\hfil{
							$
							\begin{tikzpicture}
							\begin{pgfonlayer}{nodelayer}
							\node [style=nothing] (0) at (0, 1) {};
							\node [style=nothing] (1) at (0, .5) {};
							\node [style=nothing] (2) at (0, 0) {};
							\node [style=dot] (3) at (.75, 1) {};
							\node [style=oplus] (4) at (.75, .5) {};
							\node [style=oplus] (5) at (1.25, .5) {};
							\node [style=dot] (6) at (1.25, 0) {};
							\node [style=nothing] (7) at (2, 1) {};
							\node [style=nothing] (8) at (2, .5) {};
							\node [style=nothing] (9) at (2, 0) {};
							\end{pgfonlayer}
							\begin{pgfonlayer}{edgelayer}
							\draw [style=simple] (0) to (7);
							\draw [style=simple] (1) to (8);
							\draw [style=simple] (2) to (9);
							\draw [style=simple] (3) to (4);
							\draw [style=simple] (5) to (6);
							\end{pgfonlayer}
							\end{tikzpicture}
							=
							\begin{tikzpicture}
							\begin{pgfonlayer}{nodelayer}
							\node [style=nothing] (0) at (0, 1) {};
							\node [style=nothing] (1) at (0, .5) {};
							\node [style=nothing] (2) at (0, 0) {};
							\node [style=dot] (3) at (1.25, 1) {};
							\node [style=oplus] (4) at (1.25, .5) {};
							\node [style=oplus] (5) at (.75, .5) {};
							\node [style=dot] (6) at (.75, 0) {};
							\node [style=nothing] (7) at (2, 1) {};
							\node [style=nothing] (8) at (2, .5) {};
							\node [style=nothing] (9) at (2, 0) {};
							\end{pgfonlayer}
							\begin{pgfonlayer}{edgelayer}
							\draw [style=simple] (0) to (7);
							\draw [style=simple] (1) to (8);
							\draw [style=simple] (2) to (9);
							\draw [style=simple] (3) to (4);
							\draw [style=simple] (5) to (6);
							\end{pgfonlayer}
							\end{tikzpicture}
							$}
						
						\item 
						\label{CNOT.6}
						\hfil{
							$
							\begin{tikzpicture}
							\begin{pgfonlayer}{nodelayer}
							\node [style=onein] (0) at (0, 0) {};
							\node [style=oneout] (1) at (1, 0) {};
							\end{pgfonlayer}
							\begin{pgfonlayer}{edgelayer}
							\draw [style=simple] (0) to (1);
							\end{pgfonlayer}
							\end{tikzpicture}
							=
							\begin{tikzpicture}
							\begin{pgfonlayer}{nodelayer}
							\node [style=rn] (0) at (0, 0) {};
							\node [style=rn] (1) at (1, 0) {};
							\end{pgfonlayer}
							\begin{pgfonlayer}{edgelayer}
							\end{pgfonlayer}
							\end{tikzpicture}
							$}
						
						\item 
						\label{CNOT.7}
						\hfil{
							\begin{tabular}{c}
							$\begin{tikzpicture}
							\begin{pgfonlayer}{nodelayer}
							\node [style=onein] (0) at (0, 1) {};
							\node [style=onein] (1) at (0, .5) {};
							\node [style=nothing] (2) at (0, 0) {};
							\node [style=dot] (3) at (.5, 1) {};
							\node [style=oplus] (4) at (.5, .5) {};
							\node [style=dot] (5) at (1, .5) {};
							\node [style=oplus] (6) at (1, 0) {};
							\node [style=oneout] (7) at (1, 1) {};
							\node [style=nothing] (8) at (1.5, .5) {};
							\node [style=nothing] (9) at (1.5, 0) {};
							\end{pgfonlayer}
							\begin{pgfonlayer}{edgelayer}
							\draw [style=simple] (0) to (7);
							\draw [style=simple] (1) to (8);
							\draw [style=simple] (2) to (9);
							\draw [style=simple] (3) to (4);
							\draw [style=simple] (5) to (6);
							\end{pgfonlayer}
							\end{tikzpicture}
							=
							\begin{tikzpicture}
							\begin{pgfonlayer}{nodelayer}
							\node [style=onein] (0) at (0, 1) {};
							\node [style=onein] (1) at (0, .5) {};
							\node [style=nothing] (2) at (0, 0) {};
							\node [style=dot] (3) at (.5, 1) {};
							\node [style=oplus] (4) at (.5, .5) {};
							\node [style=oneout] (7) at (1, 1) {};
							\node [style=nothing] (8) at (1.5, .5) {};
							\node [style=nothing] (9) at (1.5, 0) {};
							\end{pgfonlayer}
							\begin{pgfonlayer}{edgelayer}
							\draw [style=simple] (0) to (7);
							\draw [style=simple] (1) to (8);
							\draw [style=simple] (2) to (9);
							\draw [style=simple] (3) to (4);
							\end{pgfonlayer}
							\end{tikzpicture}$\\
							$ $\\
							$\begin{tikzpicture}
							\begin{pgfonlayer}{nodelayer}
							\node [style=oneout] (0) at (1.5, 1) {};
							\node [style=oneout] (1) at (1.5, .5) {};
							\node [style=nothing] (2) at (1.5, 0) {};
							\node [style=dot] (3) at (1, 1) {};
							\node [style=oplus] (4) at (1, .5) {};
							\node [style=dot] (5) at (.5, .5) {};
							\node [style=oplus] (6) at (.5, 0) {};
							\node [style=onein] (7) at (.5, 1) {};
							\node [style=nothing] (8) at (0, .5) {};
							\node [style=nothing] (9) at (0, 0) {};
							\end{pgfonlayer}
							\begin{pgfonlayer}{edgelayer}
							\draw [style=simple] (0) to (7);
							\draw [style=simple] (1) to (8);
							\draw [style=simple] (2) to (9);
							\draw [style=simple] (3) to (4);
							\draw [style=simple] (5) to (6);
							\end{pgfonlayer}
							\end{tikzpicture}
							=
							\begin{tikzpicture}
							\begin{pgfonlayer}{nodelayer}
							\node [style=oneout] (0) at (1.5, 1) {};
							\node [style=oneout] (1) at (1.5, .5) {};
							\node [style=nothing] (2) at (1.5, 0) {};
							\node [style=dot] (3) at (1, 1) {};
							\node [style=oplus] (4) at (1, .5) {};
							\node [style=onein] (7) at (.5, 1) {};
							\node [style=nothing] (8) at (0, .5) {};
							\node [style=nothing] (9) at (0, 0) {};
							\end{pgfonlayer}
							\begin{pgfonlayer}{edgelayer}
							\draw [style=simple] (0) to (7);
							\draw [style=simple] (1) to (8);
							\draw [style=simple] (2) to (9);
							\draw [style=simple] (3) to (4);
							\end{pgfonlayer}
							\end{tikzpicture}$
							\end{tabular}
							}
						
						\item 
						\label{CNOT.8}
						\hfil{
							$
							\begin{tikzpicture}
							\begin{pgfonlayer}{nodelayer}
							\node [style=nothing] (0) at (0, 1) {};
							\node [style=nothing] (1) at (0, .5) {};
							\node [style=nothing] (2) at (0, 0) {};
							\node [style=dot] (3) at (.5, 1) {};
							\node [style=oplus] (4) at (.5, .5) {};
							\node [style=dot] (5) at (1, .5) {};
							\node [style=oplus] (6) at (1, 0) {};
							\node [style=dot] (7) at (1.5, 1) {};
							\node [style=oplus] (8) at (1.5, .5) {};
							\node [style=nothing] (9) at (2, 1) {};
							\node [style=nothing] (10) at (2, .5) {};
							\node [style=nothing] (11) at (2, 0) {};
							\end{pgfonlayer}
							\begin{pgfonlayer}{edgelayer}
							\draw [style=simple] (0) to (9);
							\draw [style=simple] (1) to (10);
							\draw [style=simple] (2) to (11);
							\draw [style=simple] (3) to (4);
							\draw [style=simple] (5) to (6);
							\draw [style=simple] (7) to (8);
							\end{pgfonlayer}
							\end{tikzpicture}
							=
							\begin{tikzpicture}
							\begin{pgfonlayer}{nodelayer}
							\node [style=nothing] (0) at (0, 1) {};
							\node [style=nothing] (1) at (0, .5) {};
							\node [style=nothing] (2) at (0, 0) {};
							\node [style=dot] (5) at (.5, .5) {};
							\node [style=oplus] (6) at (.5, 0) {};
							\node [style=dot] (7) at (1, 1) {};
							\node [style=oplus] (8) at (1, 0) {};
							\node [style=nothing] (9) at (1.5, 1) {};
							\node [style=nothing] (10) at (1.5, .5) {};
							\node [style=nothing] (11) at (1.5, 0) {};
							\end{pgfonlayer}
							\begin{pgfonlayer}{edgelayer}
							\draw [style=simple] (0) to (9);
							\draw [style=simple] (1) to (10);
							\draw [style=simple] (2) to (11);
							\draw [style=simple] (5) to (6);
							\draw [style=simple] (7) to (8);
							\end{pgfonlayer}
							\end{tikzpicture}
							$}
						
						\item 
						\label{CNOT.9}
						\hfil{
							$
							\begin{tikzpicture}
							\begin{pgfonlayer}{nodelayer}
							\node [style=onein] (0) at (0, 1) {};
							\node [style=onein] (1) at (0, .5) {};
							\node [style=nothing] (2) at (0, 0) {};
							\node [style=dot] (3) at (.5, 1) {};
							\node [style=oplus] (4) at (.5, .5) {};
							\node [style=oneout] (7) at (1, 1) {};
							\node [style=oneout] (8) at (1, .5) {};
							\node [style=nothing] (9) at (1, 0) {};
							\end{pgfonlayer}
							\begin{pgfonlayer}{edgelayer}
							\draw [style=simple] (0) to (7);
							\draw [style=simple] (1) to (8);
							\draw [style=simple] (2) to (9);
							\draw [style=simple] (3) to (4);
							\end{pgfonlayer}
							\end{tikzpicture}
							=
							\begin{tikzpicture}
							\begin{pgfonlayer}{nodelayer}
							\node [style=onein] (0) at (0, 1) {};
							\node [style=onein] (1) at (0, .5) {};
							\node [style=nothing] (2) at (0, 0) {};
							\node [style=dot] (3) at (1, 1) {};
							\node [style=oplus] (4) at (1, .5) {};
							\node [style=oneout] (7) at (2, 1) {};
							\node [style=oneout] (8) at (2, .5) {};
							\node [style=nothing] (9) at (2, 0) {};
							\node [style=oneout] (10) at (0.75, 0) {};
							\node [style=onein] (11) at (1.25, 0) {};
							\end{pgfonlayer}
							\begin{pgfonlayer}{edgelayer}
							\draw [style=simple] (0) to (7);
							\draw [style=simple] (1) to (8);
							\draw [style=simple] (2) to (10);
							\draw [style=simple] (11) to (9);
							\draw [style=simple] (3) to (4);
							\end{pgfonlayer}
							\end{tikzpicture}
							$}
					\end{enumerate}
				\end{multicols}
				\
			\end{mdframed}
	}}
	\caption{The identities of \texorpdfstring{$\CNOT$}{CNOT}}
	\label{fig:CNOT}
\end{figure}

\end{document}